\documentclass[12pt]{article}

\usepackage[margin=1in]{geometry}
\usepackage{color}
\usepackage{graphicx}
\usepackage{amsmath}
\usepackage{amssymb}
\usepackage{amsthm}
\usepackage{rotating}
\usepackage{verbatim}
\usepackage{hyperref}
\usepackage{bm}
\usepackage[authoryear]{natbib}
\usepackage{bbm}
\usepackage{algorithm,algorithmic}
\usepackage{tikz}
\usetikzlibrary{decorations.pathreplacing,arrows}
\usepackage[hang,flushmargin]{footmisc}

\def\independenT#1#2{\mathrel{\rlap{$#1#2$}\mkern2mu{#1#2}}}
\newcommand\independent{\protect\mathpalette{\protect\independenT}{\perp}}
\newcommand{\pr}[1]{\mathbb{P}\left(#1\right)}
\newcommand{\prwrt}[2]{\mathbb{P}_{#1}\left(#2\right)}
\newcommand{\prst}[2]{\mathbb{P}\left(#1 \ \middle| \ #2\right)}
\newcommand{\E}[1]{\mathbb{E}\left[#1\right]}
\newcommand{\Est}[2]{\mathbb{E}\left[#1 \ \middle| \ #2\right]}
\newcommand{\Ewrt}[2]{\mathbb{E}_{#1}\left[#2\right]}

\DeclareMathOperator*{\argmin}{arg\,min}
\def\R{\mathbb{R}}
\newcommand{\One}[1]{\mathbbm{1}\left\{#1\right\}}
\newcommand{\one}[1]{\mathbbm{1}_{#1}}
\newcommand{\norm}[1]{\|#1\|}
\newcommand{\opnorm}[1]{\|#1\|_{\textnormal{op}}}
\newcommand{\fronorm}[1]{\|#1\|_{\textnormal{F}}}
\newcommand{\bignorm}[1]{\big\|#1\big\|}
\newcommand{\Bignorm}[1]{\Big\|#1\Big\|}
\newcommand{\Xcal}{\mathcal{X}}
\newcommand{\thetah}{\widehat{\theta}}
\newcommand{\thetahi}{\thetah_{\textnormal{init}}}
\newcommand{\thetamle}{\widehat{\theta}_{\textnormal{MLE}}}
\newcommand{\Xt}{\widetilde{X}}
\newcommand{\Pt}{\widetilde{P}}

\newcommand{\bigo}{\mathcal{O}}
\newcommand{\bigot}{\widetilde{\mathcal{O}}}
\newcommand{\littleo}{\mathrm{o}}
\newcommand{\bigop}{\mathcal{O}_{\mathrm{P}}}

\newcommand{\dtv}{\mathsf{d}_{\mathsf{TV}}}
\newcommand{\dex}{\mathsf{d}_{\mathsf{exch}}}
\newcommand{\iidsim}{\stackrel{\textnormal{iid}}{\sim}}
\newcommand{\eps}{\varepsilon}
\newcommand{\Ecal}{\mathcal{E}}
\newcommand{\Lcal}{\mathcal{L}}
\newcommand{\ident}{\mathbf{I}}
\newcommand{\normal}{\mathcal{N}}
\newcommand{\ball}{\mathbb{B}}

\newcommand{\giv}{\,|\,}
\newcommand{\leb}{\textnormal{Leb}}
\newcommand{\pen}{\mathcal{R}}

\newtheorem{theorem}{Theorem}
\newtheorem{lemma}{Lemma}
\newtheorem{assumption}{Assumption}

\newtheorem{definition}{Definition}
\newtheorem{setting}{Model Class}
\newtheorem{problem}{Problem Domain}
\newtheorem{example}{Example}

\numberwithin{equation}{section}

\title{Testing goodness-of-fit and conditional independence\\ with approximate co-sufficient sampling}
\author{Rina Foygel Barber\thanks{Department of Statistics, University of Chicago} \ and Lucas Janson\thanks{Department of Statistics, Harvard University}}
\date{}

\begin{document}
\maketitle


\begin{abstract}
Goodness-of-fit (GoF) testing is ubiquitous in statistics, with direct ties to model selection, confidence interval construction, conditional independence testing, and multiple testing, just to name a few applications. While testing the GoF of a simple (point) null hypothesis provides an analyst great flexibility in the choice of test statistic while still ensuring validity, most GoF tests for composite null hypotheses are far more constrained, as the test statistic must have a tractable distribution over the entire null model space. A notable exception is \emph{co-sufficient sampling} (CSS): resampling the data conditional on a sufficient statistic for the null model guarantees valid GoF testing using any test statistic the analyst chooses. But CSS testing requires the null model to have a compact (in an information-theoretic sense) sufficient statistic, which only holds for a very limited class of models; even for a null model as simple as logistic regression, CSS testing is powerless. In this paper,
we leverage the concept of \emph{approximate sufficiency} to generalize CSS testing to essentially any parametric model with an asymptotically-efficient estimator; we call our extension ``approximate CSS'' (aCSS) testing. We quantify the finite-sample Type I error inflation of aCSS testing and show that it is vanishing under standard maximum likelihood asymptotics, for any choice of test statistic. We apply our proposed procedure both theoretically and in simulation to a number of models of interest to demonstrate its finite-sample Type I error and power.
\end{abstract}



\section{Introduction}\label{sec:intro}
Suppose we observe data $X$ belonging to some sample space $\Xcal$, and would like to test whether it comes from some parametric null model $\{P_\theta:\theta\in\Theta\}$, where $\Theta\subseteq\R^d$, versus a more complex (usually higher-dimensional) model. This problem of so-called ``goodness-of-fit" (GoF) testing is one of the most fundamental in statistics, with a vast literature exhibiting applications and theoretical and methodological development. We pause here to highlight a few of the many areas of statistics to which GoF testing is directly applicable, including some problems that are not obviously or commonly associated with  GoF.

\begin{problem}[Standard goodness-of-fit testing]\label{ex:goft}
GoF testing is commonly used to test a postulated model or distributional property, often as a precursor to further statistical analysis that assumes the postulated model/property to be correct. Such null models/properties include standard distributional families, nonparametric properties such as symmetry or log-concavity, time-series properties such as stationarity, and relational properties such as independence.
\end{problem}

\begin{problem}[Model selection]\label{ex:ms}
GoF testing can also be used to select a best-fitting model through simultaneously testing a family of models. For instance, this could be choosing a sparse model in regression, selecting the number of clusters or principle components in unsupervised learning, or identifying change points in a time series.
\end{problem}

\begin{problem}[Confidence interval construction]\label{ex:cic}
Suppose the data $X$ is distributed according to a known model $\{P_{\gamma,\theta}: (\gamma,\theta)\in\Gamma\times\Theta\}$, where $\Gamma\subseteq\R^m$ and $\Theta\subseteq\R^d$, and the goal is to construct a confidence region for $\gamma$ in the presence of the nuisance parameter $\theta$. 
If for any $\gamma_0$, we can construct a GoF test for the null model $\{P_{\gamma_0,\theta}:\theta\in\Theta\}$, then the set of $\gamma_0$ at which we fail to reject the test constitutes a valid confidence region. 
\end{problem}

\begin{problem}[Conditional independence testing]\label{ex:cit}
In many regression and graphical modeling problems, the primary question of interest is whether a given triple of random variables $(X,Y,Z)$ satisfies conditional independence, i.e., $Y\independent X\giv Z$. 
If $X\giv Z$ is distributed according to a known model $\{P_\theta(\cdot\giv Z):\theta\in\Theta\}$, then testing conditional independence can be formulated as a GoF test where $\{P_\theta(\cdot\giv Z):\theta\in\Theta\}$ is the null model for $X\giv Z,Y$. (Note that, when we apply a GoF test to the conditional independence problem in this way,
we implicitly treat $Y$ and $Z$ as fixed, and check for goodness-of-fit of $X$'s conditional model.)
\end{problem}

\begin{problem}[Multiple testing]\label{ex:mt}
In multiple testing, the goal is to reject a subset of a fixed family of null hypotheses. When each hypothesis test is a GoF test (i.e., its null is lower-dimensional than its alternative), testing any intersection of null hypotheses (i.e., testing the ``global null" for any subset of hypotheses) constitutes a GoF test as well, and combining such intersection GoF tests through a closed testing procedure \citep{marcus1976closed} produces a subset to reject that controls the familywise error rate.
\end{problem}

The key challenges of any GoF testing problem are to find a test that is valid, in that it controls the Type I error at a prespecified significance level, and that is powerful, in that it rejects the null model as often as possible when it does not hold. 
For parametric null models (the focus of this paper), there are many existing methods for testing GoF, with canonical choices including the popular score, likelihood ratio, and Wald tests. The standard approach for these tests and many others is to prescribe a test statistic (chosen to be powerful under a given alternative model) and derive a (often asymptotic) null distribution for it. Such tests require certain regularity conditions on the alternative model (in order to construct a well-behaved test statistic) and on the null model (in order to establish the validity of the null distribution for the test statistic) that are generally quite similar to those needed for the maximum likelihood estimator under both the null and alternative to be asymptotically normal. While these tests are extremely popular and have been fruitfully applied through much of the history of statistics to many problems, the regularity assumptions placed on the alternative distribution in particular limit the ability to fully leverage domain knowledge to maximize the statistical power. To elaborate, consider the following cases which often arise in practice when applying parametric GoF tests.
\begin{itemize}
    \item {\bf Some prior information is available about the relative plausibility of different regions of the alternative model.} Ideally we would like to incorporate this prior information into our test statistic in order to maximize power (e.g., through a test statistic derived from Bayesian inference), but standard GoF tests only provide the null distribution for a test statistic which is determined by the entire alternative space, and give little flexibility to incorporate prior knowledge into that test statistic while still retaining a valid null distribution. An extreme case would be when certain regions of the alternative are known to be completely implausible, i.e., we would like to remove them from the alternative model entirely, yet removing them from the model would violate the regularity conditions required for the alternative model. For example, we may know that under the alternative some $k$-dimensional parameter has at most $d<k$ non-zero entries, but we do not know which ones. Such a sparse alternative model would violate the usual assumption that the parameter space is convex, forcing one to ignore the sparsity and instead operate under (and hence derive a test statistic from) the larger, mostly implausible, $k$-dimensional alternative model. As we see in the next scenario, if $k$ is too large, even this route is not feasible.
    \item {\bf The alternative model is high- or even infinite-dimensional.} Since standard GoF tests treat their prescribed test statistics as fixed in the asymptotic regime in which they prove validity, those test statistics can only be designed to be powerful against fixed- (and finite-) dimensional alternatives. If we have a high-dimensional alternative (i.e., whose dimension is not assumed negligible relative to the sample size, which includes any nonparametric alternative), we would ideally like to choose a test statistic which changes with the sample size to be powerful against a sequence of alternatives which changes as the sample size grows asymptotically. But standard GoF tests cannot accommodate such a choice, forcing one to instead operate under a fixed-dimensional alternative that may represent a vanishing fraction of plausible alternatives, or a very coarse approximation to the space of realistic alternatives.
\end{itemize}
The common thread in these cases is that the test statistic that would be most powerful to use given the domain knowledge at hand 
is not accompanied by theoretical guarantees or any known (exact or approximate) null distribution. In our simulations in Section~\ref{sec:examples}, we will study some examples where 
standard tests
can be applied (and will compare aCSS to the score test for those examples), and others where, as in the scenarios above, 
standard tests
cannot be applied and thus a more flexible method like aCSS is necessary.

However, constructing a valid test around an arbitrary test statistic $T(X)$ is possible only in very limited settings.
In particular, if
$\Theta$ is simple, i.e., it contains only a single point so that there are no unknown/nuisance parameters in the null model, then any test statistic $T(X)$'s null distribution can be arbitrarily-well approximated computationally by repeatedly independently sampling copies $\Xt$ of $X$ from the null distribution and recomputing the same test statistic on the copies. 
To be concrete, if the statistic $T(X)$ is chosen such that larger (positive) values are seen as evidence against the null, we can draw $M$ i.i.d.~copies $\Xt^{(1)},\dots,\Xt^{(M)}$ from the null distribution,
and define a (discretized) p-value
\begin{equation}\label{eqn:pval}
\textnormal{pval} = \textnormal{pval}_T(X,\Xt^{(1)},\dots,\Xt^{(M)}) = \frac{1}{M+1}\left(1+\sum_{m=1}^M \One{T(\Xt^{(m)})\ge T(X)}\right),
\end{equation}
which is guaranteed to satisfy $\pr{\textnormal{pval}\leq \alpha}\leq \alpha$ under the null, for any predefined rejection level $\alpha$.

More generally, when $\Theta$ is not a singleton set (i.e., the null hypothesis is composite), in principle we can 
still construct a p-value of the form~\eqref{eqn:pval}
as long as we are able to sample a set of copies $\Xt^{(1)},\dots,\Xt^{(M)}$ of $X$ so that $X,\Xt^{(1)},\dots,\Xt^{(M)}$ are \emph{exchangeable} under the null.
We emphasize that this exchangeability property continues to enable the analyst to use any desired test statistic $T(X)$, as the validity of the p-value is 
unaffected. Of course, to achieve high power, we should aim to choose a function $T(X)$ that is likely to be large under the alternative hypothesis.
Note that we absorb everything that is not $X$ into the definition of the function $T$, e.g., for testing conditional independence $X\independent Y \giv Z$ as in Problem Domain~\ref{ex:cit}, $T$ can depend arbitrarily on $Y$ and $Z$ as well (since, after conditioning 
on $Y$ and $Z$, they are treated as fixed and nonrandom).

To summarize, we have seen that\[\begin{array}{l}\textnormal{The problem of GoF testing with arbitrary test statistics can be reduced to}\\\textnormal{one of sampling copies of $X$ that are exchangeable under the null.}\end{array}\]
These copies $\Xt^{(1)},\dots,\Xt^{(M)}$ then act as a ``control group'' for the real data $X$, and we can
compare the real statistic $T(X)$ against the ``control group'' values $T(\Xt^{(m)})$ to test the null.
Of course, aside from the setting of a simple null,
sampling exchangeable copies is not necessarily a straightforward task. 
In particular,
in order to have power, we must ensure our null-exchangeable copies do not remain exchangeable under the alternative; for instance sampling $\Xt^{(1)}=\cdots=\Xt^{(M)}=X$ trivially satisfies the exchangeability property under the null, but also under any alternative, and hence clearly equation~\eqref{eqn:pval} produces a useless $p$-value of 1 with probability 1 (for any choice of $T$).

One way to sample exchangeable copies when $\Theta$ is composite is by conditioning on a sufficient statistic $S(X)$ for $\theta$, since then by definition the conditional distribution $X\giv S(X)$ does not depend on $\theta$. By drawing the copies $\Xt^{(m)}$ from this conditional distribution, we achieve exchangeability of $X,\Xt^{(1)},\dots,\Xt^{(M)}$ (more concretely, $X$ and its copies are i.i.d.~conditional on $S(X)$). This approach is known as \emph{co-sufficient sampling} (CSS) \citep{stephens2012goodness}. However, this approach is viable in
only a limited range of settings. 
In particular, many null models do not admit a compact (in an information-theoretic sense) sufficient statistic, meaning any sufficient statistic for the null model will remain sufficient for many alternative models as well. In such cases, which we term \emph{degenerate}, CSS testing runs into the problem described at the end of the preceding paragraph---the copies $\Xt^{(m)}$ will still be exchangeable with $X$ under the alternative, resulting in a  powerless test. This situation arises quite often, and we will give a number of common examples shortly in Section~\ref{sec:applications}.

As an alternative approach, we might consider the {\em parametric bootstrap} \citep{efron1994introduction}, where after estimating the true parameter $\theta$ via some estimator $\thetah$ (e.g., the maximum likelihood estimate), the copies $\Xt^{(1)},\dots,\Xt^{(M)}$ are sampled from $P_{\thetah}$. While this widely-used approach often works well in practice, the parametric bootstrap does not create exchangeable copies of the data, and is not guaranteed to achieve the desired Type I error level when paired with an arbitrary test statistic $T$---in fact, it may even lead to substantial error inflation. To take a simple example, consider a Gaussian linear regression setting where $P_\theta$ is given by the distribution $X\sim \normal(\theta\cdot Z, \mathbf{I}_d )$, where $Z\in\R^n$ is a fixed covariate. Suppose we are interested in testing whether $X$ in fact has more dependence with another covariate $Y\in\R^n$, and so our test statistic is given by 
\[T(X) = \frac{(X^\top Y)^2}{(X^\top Z)^2}.\]
Figure~\ref{fig:parboot} compares the parametric bootstrap against co-sufficient sampling (full details for this simulation are given in Appendix~\ref{app:details_parboot}). The results show that CSS results in a uniform distribution of p-values, while the parametric bootstrap results in a highly non-uniform distribution, and could lead to substantially inflated Type I error.
 Therefore, we would instead prefer to extend the CSS framework in order to enjoy theoretically guaranteed error control.

\begin{figure}\centering
\includegraphics[width=\textwidth]{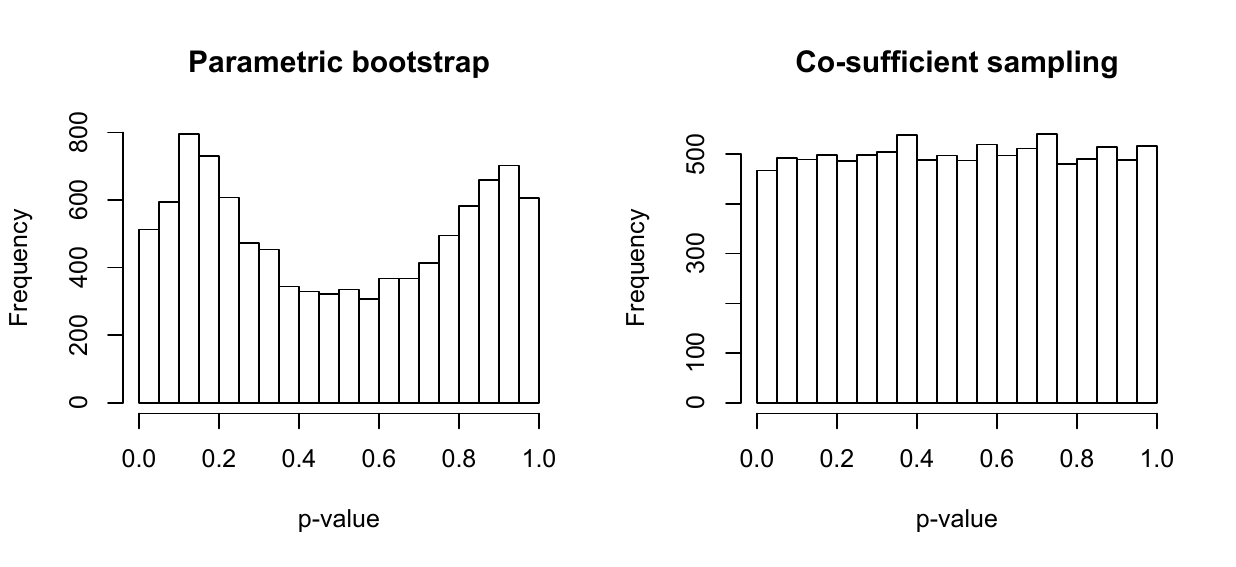}
\caption{Comparison of the parametric bootstrap versus CSS, in a Gaussian linear model example where the null hypothesis is true. We can see that CSS yields uniformly distributed p-values, but the parametric bootstrap does not, resulting in an inflated Type I error rate. (See Section~\ref{sec:intro} for details.)}
\label{fig:parboot}
\end{figure}

\subsection{Our contribution}
In this paper,
we demonstrate how to escape the problem of zero power in the degenerate setting,
by introducing a new generalization of CSS testing that conditions only on an approximately sufficient statistic \citep{lecam1960locally,van2000asymptotic,le2012asymptotic}.
We call such a test an ``approximate co-sufficient sampling'' test. 
This paper makes four main contributions:
\begin{enumerate}
\item We propose \emph{approximate co-sufficient sampling} (aCSS), which samples approximately exchangeable copies of the data by conditioning on an approximately sufficient statistic and plugging in a consistent estimator for the unknown parameter.
\item Under weak conditions closely resembling those for standard maximum likelihood asymptotics, we provide a finite-sample upper-bound for the total variation (TV) distance from exchangeability of our aCSS samples.
\item We show that the aforementioned TV bound translates directly to a bound on the Type I error inflation of an aCSS test that holds uniformly over the choice of test statistic, and we apply this bound to a number of important models to prove the inflation vanishes asymptotically as the sample size approaches infinity.
\item We provide general algorithms for aCSS and demonstrate their use in a series of simulations that exhibit the validity and power of aCSS testing.
\end{enumerate} 

\subsection{Applications}\label{sec:applications}

The problem of zero power for CSS testing arises surprisingly often---while we call such settings ``degenerate'',
they are not extreme cases but rather constitute a large portion of  common statistical models of interest.
To illustrate this, we will consider the following settings in which CSS testing is powerless, while aCSS testing can still be quite powerful and remains asymptotically valid for any test statistic. 
(Our theoretical results later on will quantify its finite-sample Type I error).

\begin{setting}[Data with associated covariates]\label{cl:dwac}
Suppose $X=(X_1,\dots,X_n)$ where the $X_i$'s are independent, and
each $X_i$ has an associated covariate $Z_i$ (i.e., the null distribution of each $X_i$ depends on this $Z_i$). In this setting, the distribution of $X$ 
(i.e., the joint distribution of $X_1,\dots,X_n$) will often have the data $X$ itself as a minimal sufficient statistic. This is even true when $X_i\giv Z_i$ follows 
logistic regression: for generic values of the $Z_i$'s (e.g., if each $Z_i\in\R^d$ is drawn from some continuous distribution), the minimal sufficient statistic is equivalent to $X$ itself because the value of $Z^\top X\in\R^d$ determines $X\in\{0,1\}^n$ uniquely, and hence is sufficient under any alternative as well (and hence any CSS test must be powerless). This problem class applies not just for GoF testing for conditional models for $X$ (including the logistic model), but also for conditional independence testing as described in Problem Domain~\ref{ex:cit}.
\end{setting}

\begin{setting}[Curved exponential families]\label{cl:cef}
Consider a null model that is a curved exponential family, i.e., a full-rank $k$-parameter exponential family with an added constraint that reduces the dimension of the parameter space to some $d<k$ and is nonlinear in the canonical parameters. In this setting, the minimal sufficient statistic is generally the same as that for the unconstrained (full-rank) exponential family. 
This means that any CSS test must be powerless against any alternative that lies in the larger exponential family, for example, if we want to test whether the parameter constraint
holds or not. A classical example is the Behrens--Fisher problem of testing equality of (unknown) means between independent normal samples having different (unknown) variances: any CSS test will be powerless for every alternative pair of means and variances.
The same issue arises in, 
e.g., the study of contingency tables (where the canonical family is multinomial and the null hypothesis imposes a constraint on its probabilities), and spatial and time-series models (where the null hypothesis imposes a spatial or temporal structure on the canonical parameters of an exponential family).
\end{setting}

\begin{setting}[Heavy-tailed models]\label{cl:htm}
Many heavy-tailed models are not exponential families and do not admit compact sufficient statistics. For instance, the Cauchy location family's minimal sufficient statistic is given by the order statistics.
Any CSS test of this null is therefore powerless against any i.i.d. alternative, since the order statistics are sufficient for this alternative as well. As another
example, the Student's $t$ scale family's minimal sufficient statistic is the order statistics of the absolute values, so any CSS test for it must be powerless against any i.i.d. symmetrical alternative.
\end{setting}

\begin{setting}[Models with latent variables]\label{cl:mwlv}
Many popular models capture domain-specific properties through latent variables. In such models,
if we condition on the latent variable, then the data often comes from a well-behaved distribution with compact minimal sufficient statistic. However, when the latent variable is unobserved, we are forced to perform inference unconditionally, and the unconditional model rarely has a compact minimal sufficient statistic. Examples include hidden Markov models, mixture distributions, data with missing values, errors-in-variables models, and factor models.
\end{setting}

\noindent Later on in Section~\ref{sec:examples}, we will return to
Model Classes~\ref{cl:dwac},~\ref{cl:cef}, and~\ref{cl:htm} and give concrete examples of models where aCSS can be applied.
We leave Model Class~\ref{cl:mwlv} for future work.

\subsection{Related work}\label{sec:relatedwork}
GoF testing dates back to the very early days of the field of statistics, and the literature even on just parametric GoF is far too numerous to cite. Instead, we focus our literature review on the subfield of CSS testing, which distinguishes itself within the broader field of parametric GoF testing by guaranteeing Type I error control with \emph{any} test statistic under a parametric null model, the potential advantages of which have been described earlier in this section. In fact, some of the most foundational nonparametric tests can be thought of as CSS tests, including the permutation test (conditions on the order statistics for an i.i.d. null). The formal idea of a CSS test seems to date back to at least \cite{bartlett1937properties}, although the value of sufficient statistics for GoF testing in the presence of nuisance parameters has also been used in many other ways, e.g., \cite{durbin1961some,kumar1977sufficiency,bell1984inference} decompose the data into a minimal sufficient statistic and an ancillary statistic and construct a GoF test based on the parameter-free distribution of the ancillary statistic. 

CSS testing has gained substantial interest in the last 30 years, though with a focus on non-degenerate hypothesis testing settings \citep{agresti1992,engen1997stochastic,agresti2001exact,o2006conditional,lockhart2007use,lockhart2009exact,lindqvist2011monte,broniatowski2012conditional,lockhart2012conditional,stephens2012goodness,lindqvist2013exact,hazra2013exact,beltran2019goodness,santos2019metropolis,contreras2019use}. Our work differs from the existing work in CSS testing by allowing for degenerate (and non-degenerate) models by conditioning only on an approximately sufficient statistic.
Similar techniques have been used to obtain exact confidence intervals in the presence of nuisance parameters \citep{lillegard1999exact}, again in non-degenerate settings. As an example, \cite{RP:1984,kolassa2003algorithms} study conditional independence testing of $X\independent Y\giv Z$ where $X\giv Z$ follows a logistic regression model and $Z$ is discrete; logistic regression is degenerate when $Z$ has a continuous distribution but non-degenerate when $Z$ is discrete. 

For conditional independence testing (Problem Class~\ref{ex:cit}), in the setting where $H_0$ is a simple null hypothesis (i.e., $X\giv Z$ has a known 
distribution), \cite{EC-ea:2018} study procedures of the form~\eqref{eqn:pval} under the name ``conditional randomization test'';
their work also constructs the model-X knockoffs framework for simultaneously testing multiple conditional independence hypotheses (i.e., variable selection in a multivariate regression). This construction provides an exact ``swap-exchangeability'' property that enables variable selection, i.e., simultaneous conditional independence testing of many covariates
when the multivariate covariates $X$ come from a non-degenerate model, and  leads to exact false discovery rate control  \citep{barber2015controlling}. Generalizing to the setting where $H_0$ is not simple,
\cite{DH-LJ:2019} construct model-X knockoffs \citep{EC-ea:2018} conditional on a sufficient statistic, retaining the exact ``swap-exchangeability'' property and exact false discovery rate control; we can think of this as a knockoffs-analogue of CSS testing.

Moving beyond CSS testing, we are only aware of a few works which take a similarly approximate approach to that of the present paper.
First, the most related to our approach is the work of \cite{lillegaard2001tests}, 
where they mention the possibility of an aCSS-type test to solve the Behrens--Fisher problem
(i.e., testing for a  difference of means between two Gaussian samples, described above in Model Class~\ref{cl:cef}), but
conclude that such an approach would be computationally intractable; they instead propose a heuristic sampling procedure which they support with simulations but no theory. Second, 
\cite{kalbfleisch1970application,cox1987parameter}  focus on parametric likelihood-based testing in the presence of nuisance parameters,
but study the case where the nuisance parameters are orthogonal or asymptotically orthogonal to the parameter of interest.
Finally, approximate Bayesian computation, also known as likelihood-free inference, conducts Bayesian inference conditional on a compact non-sufficient statistic, but for computational, as opposed to statistical reasons, since in the Bayesian framework there is no statistical downside to conditioning on as much as possible (see, e.g., \cite{kousathanas2016likelihood} for such a paper that explicitly addresses the role of sufficiency).

\subsection{Notation}
We will write $\norm{\!\cdot\!}$ to denote the usual Euclidean ($\ell_2$) norm on vectors, and to denote the operator norm (i.e., spectral norm) on matrices.
For a matrix $M$, $\lambda_{\max}(M)$ denotes its largest eigenvalue in the positive direction.
We write $(x)_+$ to denote $\max\{x,0\}$. 
We will write $\mathbb{E}_{\theta}$ and $\mathbb{P}_{\theta}$ to denote expectation or probability taken
with respect to $X\sim P_\theta$, where the parametric family $\{P_\theta : \theta\in\Theta\}$ is our null model.


\section{Method}
The goal of approximate co-sufficient sampling is to generate copies $\Xt^{(1)},\dots,\Xt^{(M)}$ of the observed data $X$ such that if the null hypothesis 
\begin{equation}\label{eqn:H0}H_0 : X\sim P_\theta\textnormal{ for some $\theta\in\Theta$}\end{equation}
is true, then the random variables $X,\Xt^{(1)},\dots,\Xt^{(M)}$ are approximately exchangeable.
Recalling the p-value defined in~\eqref{eqn:pval}, we can then test the null hypothesis 
using any desired test statistic $T(X)$. The choice of statistic is completely unconstrained, and this flexibility enables us to 
design very powerful tests in many settings. 
Note that, although this flexibility allows us to design any form of function $T$, $T$ itself (as a function, i.e., before seeing its argument) cannot depend on $X$. For example, if $T(X)$ uses $X$ to tune parameters in a neural network and then computes a statistic of that neural network applied to $X$, then $T(\Xt^{(m)})$ cannot compute a statistic on the same ($X$-tuned) neural network applied to $\Xt^{(m)}$, but must use $\Xt^{(m)}$ to tune the parameters of a new neural network and compute a statistic of that ($\Xt^{(m)}$-tuned) neural network applied to $\Xt^{(m)}$.

To quantify our goal of generating approximately exchangeable copies of the data, we begin by defining a ``distance to exchangeability'':
\begin{definition}\label{def:dex}
For any integer $k\geq1$ and any set of random variables $A_1,\dots,A_k$ with a joint distribution, define
\[\dex(A_1,\dots,A_k) = \inf\left\{\dtv\big((A_1,\dots,A_k), (B_1,\dots,B_k)\big) : \textnormal{$B_1,\dots,B_k$ are exchangeable}\right\}.\]
Here $\dtv$ denotes the total variation distance,
and  the infimum is taken over all sets of $k$ random variables $B_1,\dots,B_k$ with an exchangeable joint distribution.
\end{definition}
\noindent Of course, if $A_1,\dots,A_k$ are exchangeable, then $\dex(A_1,\dots,A_k)=0$.
When we say informally that variables $A_1,\dots,A_k$ are ``approximately exchangeable'', we mean that the distance to exchangeability is small.

Now we will see how this distance $\dex$ relates to the problem of testing the null hypothesis~\eqref{eqn:H0}
(\cite{berrett2019conditional} use a similar argument in a permutation test setting).
Fix a threshold $\alpha\in[0,1]$ and a function $T:\Xcal\rightarrow\R$ (the test statistic).
For any exchangeable random variables $(B_0,\dots,B_M)$, by definition of exchangeability we have $\pr{\textnormal{pval}_T(B_0,B_1,\dots,B_M)\leq \alpha}\leq \alpha$,
and therefore,
\[\pr{\textnormal{pval}_T(X,\Xt^{(1)},\dots,\Xt^{(M)})\leq\alpha}
\leq \alpha + \dtv\big((X,\Xt^{(1)},\dots,\Xt^{(M)}),(B_0,B_1\dots,B_M)\big).\]
Taking an infimum over all exchangeable distributions on $(B_0,B_1\dots,B_M)$, we have shown that
\[\pr{\textnormal{pval}_T(X,\Xt^{(1)},\dots,\Xt^{(M)})\leq\alpha}
\leq \alpha +\dex(X,\Xt^{(1)},\dots,\Xt^{(M)})\]
under the null hypothesis $H_0$.

Therefore, we can see that, if we are able to construct copies of the data such that $\dex(X,\Xt^{(1)},\dots,\Xt^{(M)})$ is small,
then we can construct an approximately-valid test of $H_0$ using any desired test statistic $T$.
From this point on, then, our task is to determine how we can use approximate sufficiency to generate such copies.

\subsection{Overview}
Consider any function $S = S(X)$ of the data, which is sufficient under the null hypothesis that $X\sim P_\theta$ for some $\theta\in\Theta$.
Let $P(\cdot\giv s)$ be the conditional distribution of $X$ given $S=s$ 
(sufficiency of $S(X)$ ensures that this distribution does not depend on $\theta$).
As described in Section~\ref{sec:intro},
the co-sufficient sampling (CSS) method operates by drawing copies from this conditional distribution. That is,
the joint distribution of the data and the copies, under the CSS method, is given by:
\[\begin{cases}
X\sim P_{\theta_0},\\
S = S(X),\\
\Xt^{(1)},\dots,\Xt^{(M)}\mid X, S \iidsim P(\cdot\giv S),\end{cases}\]
where $\theta_0$ is the unknown true parameter.
Clearly, the real and fake data $X,\Xt^{(1)},\dots,\Xt^{(M)}$ are i.i.d.~conditional on $S$, and are therefore exchangeable,
meaning that the $\Xt^{(m)}$'s provide a valid ``control group'' for the real data $X$ regardless of the unknown $\theta_0$.

As discussed above, this framework is limited to only certain specific models, since many common models are ``degenerate''
(such as the model classes described in Section~\ref{sec:applications}), where
 any sufficient statistic $S=S(X)$ reveals so much information about $X$ that it leads to a completely powerless procedure against the alternative hypothesis of interest.
We can instead consider statistics $S=S(X)$ that are not sufficient, but are {\em approximately sufficient}, meaning
that the distribution $P_{\theta}(\cdot\giv S)$ is {\em approximately} unaffected by the value of $\theta$---more concretely, if we can estimate
$\theta$ with a consistent estimator $\thetah$, then we only need to ensure that $P_{\thetah}(\cdot\giv S)\approx P_{\theta}(\cdot\giv S)$. 
In fact, for many settings, maximum likelihood estimation is known to provide an asymptotically sufficient statistic \citep{lecam1960locally,van2000asymptotic,le2012asymptotic}. Thus, we can take $S=S(X)$ to simply
be $\thetamle(X)$, or more generally, any other estimator of $\theta_0$ that is asymptotically sufficient. 

In this setting, we write $P_{\theta_0}(\cdot\giv\thetah)$ to denote the conditional distribution of $X\giv\thetah$ when the data is 
distributed as $X\sim P_{\theta_0}$ and we calculate $\thetah = \thetamle(X)$. Of course, we cannot draw the copies from this distribution since $\theta_0$ is unknown,
but if $\thetah=\thetamle(X)$ is approximately sufficient,
then the distribution $P_{\theta_0}(\cdot\giv\thetah)$ should depend only slightly on $\theta_0$. In particular, we will use $\thetah$ itself
as a plug-in estimate for $\theta_0$, 
leading to the joint model
\[\begin{cases}
X\sim P_{\theta_0},\\
\thetah = \thetamle(X),\\
\Xt^{(1)},\dots,\Xt^{(M)}\mid X,  \thetah \iidsim P_{\thetah}(\cdot\giv\thetah).\end{cases}\]
These copies form an approximately-valid control group as long as 
$P_{\thetah}(\cdot\giv\thetah)\approx P_{\theta_0}(\cdot\giv\thetah)$.

In our aCSS algorithm, 
we will replace the deterministic step $\thetah=\thetah(X)$ with a randomized estimator (essentially, adding a small random perturbation into the likelihood maximization problem).
Adding noise is beneficial for computational reasons, since the set of $x\in\Xcal$ whose MLE is exactly
equal to $\thetamle(X)$ may be a challenging set to sample from.  For certain examples, adding noise can also be beneficial from the statistical point of view, as for, e.g., the logistic regression setting, described
in Model Class~\ref{cl:dwac}, where conditioning on the exact MLE, $\thetamle(X)$, may lead to a zero-power scenario. (We will discuss the role of $\sigma$ further in Section~\ref{sec:choosing_sigma} below.) In addition, we will also allow adding a twice-differentiable regularization function $\pen(\theta)$ to the likelihood maximization
problem, for instance $\pen(\theta)\propto \norm{\theta}^2$ for ridge regression, which may be beneficial in some applications.

Informally, our proposed aCSS algorithm takes the following form:
\[\begin{cases}
X\sim P_{\theta_0},\\
W\sim \normal(0,\tfrac{1}{d}\ident_d),\\
\thetah = \thetah(X,W) = \argmin_{\theta\in\Theta}\left\{-\log f(X;\theta) + \pen(\theta) + \sigma W^\top\theta\right\},\\
\Xt^{(1)},\dots,\Xt^{(M)}\mid X,  \thetah \iidsim P_{\thetah}(\cdot\giv\thetah),\end{cases}\]
where again $P_{\theta_0}(\cdot\giv\thetah)$ denotes the conditional distribution of $X\giv\thetah$ when the data is 
distributed as $X\sim P_{\theta_0}$, and $P_{\thetah}(\cdot\giv\thetah)$ is a plug-in estimate.

However, in many settings 
 the penalized negative log-likelihood may not be strongly convex,
or might even be nonconvex, in which case
we will need to modify this procedure---while it is the case that, in many statistical problems,
many tools exist that are likely to find the (perturbed) MLE
with high probability (e.g., by carefully choosing a good initialization point), we will need to account for the fact
that finding the global optimum is not guaranteed. Furthermore, in order to construct the copies $\Xt^{(1)},\dots,\Xt^{(M)}$,
 we are implicitly assuming that we are able to generate i.i.d.~samples from the conditional distribution of $X\giv \thetah$. In practice,
 sampling directly from this density may be impossible, so we may need to turn to techniques such as Markov Chain Monte Carlo (MCMC),
which can introduce dependence between the samples.
Our next task, then, is to develop a more general and rigorous form of this simple algorithm, so that we can provide a practical method that can be deployed 
in a broad range of settings.

\subsection{Algorithm for approximate co-sufficient sampling}

In this section, we will formally define our aCSS algorithm.
Below, we define our noisy estimator $\thetah$ (Section~\ref{sec:draw_thetah}), calculate the conditional 
distribution of $X\giv \thetah$ (Section~\ref{sec:compute_conditional}), and 
describe how to sample the copies $\Xt^{(1)},\dots,\Xt^{(M)}$ from the estimated conditional distribution
(Section~\ref{sec:sample_copies}).

\subsubsection{Sampling the estimator}\label{sec:draw_thetah}
Recall that the  estimator $\thetah$ is intended to be approximately equal to the MLE,
even though it includes a regularization function and a random perturbation into the likelihood maximization problem. Writing
\[\Lcal(\theta;x)  = -\log f(x;\theta) + \pen(\theta) ,\]
consider the optimization problem
\begin{equation}\label{eqn:thetah_def} \argmin_{\theta\in\Theta}\Lcal(\theta;X,W)\textnormal{\  \ where \ } \Lcal(\theta;x,w)=\Lcal(\theta;x)  + \sigma w^\top\theta,\end{equation}
where $W\sim\normal(0,\tfrac{1}{d}\ident_d)$ is independent Gaussian noise, $\sigma>0$ determines the noise level of the random
perturbation, and $\pen:\Theta\rightarrow\R$ is a twice-differentiable regularization function.
In order to accommodate the penalized and unpenalized estimators with a single unified presentation, we can view the unpenalized 
version as a special case by simply taking $\pen(\theta)\equiv 0$.
(This type of randomly perturbed log-likelihood 
was previously studied by \cite{XT-JT:2018}, with the different aim of enabling selective inference on a high-dimensional parameter $\theta$.
In their work, the object of interest is the distribution of $\thetah$, to enable inference on $\theta_0$,
whereas in our setting $\theta_0$ is essentially a nuisance parameter.)

In the general setting where the negative log-likelihood might be nonconvex, the optimization problem~\eqref{eqn:thetah_def} may be challenging---in particular, in the presence of nonconvexity, how would we find a global minimizer, and is a global minimizer even guaranteed to exist? In many settings, any available algorithm would only be able to guarantee that we find a first-order stationary point to~\eqref{eqn:thetah_def} (if it even converges at all). 
To address this, we modify our procedure to allow $\thetah$ to only \emph{usually} be a well-behaved \emph{local} optimum of \eqref{eqn:thetah_def}. This enables aCSS to draw on the vast literature on optimizing penalized maximum likelihoods. Although the random perturbation by $W$ makes \eqref{eqn:thetah_def} slightly non-standard for penalized maximum likelihood, the perturbation is linear in $\theta$ and hence has no impact on Hessians or convexity and only adds a fixed, trivially-computable constant vector to the gradient. Thus, although large linear perturbations can ``tip over" an otherwise well-behaved basin of attraction, our theory will ensure this never happens asymptotically and in practice one can always choose $\sigma$ small enough to make this astronomically unlikely; see Appendix~\ref{app:optimization} for a more detailed discussion. In summary, we expect that any algorithm that empirically-often (it need not be provably-often) finds a local optimum for the unperturbed penalized maximum likelihood problem will suffice with almost no modification to solve \eqref{eqn:thetah_def} for the purposes required by the theory in this paper.

In particular, we will define $\thetah$ to be any measurable function mapping a (data, noise) pair $(x,w)$ to an estimate, i.e., 
\[\thetah: \Xcal\times\R^d\rightarrow\Theta,\]
and we will later assume this map is likely to return a strict second-order stationary point (SSOSP) of the minimization problem~\eqref{eqn:thetah_def}.
Here we say that $\theta$ is a SSOSP of $\Lcal(\theta;x,w)$ if two conditions are satisfied:
\begin{itemize}
\item $\theta$ is a first-order stationary point (FOSP) of $\Lcal(\theta;x,w)$, meaning that $\nabla_\theta \Lcal(\theta;x,w) = 0$
or equivalently $w=-\frac{\nabla_\theta \Lcal(\theta;x)}{\sigma}$.
\item The objective function is strictly convex at $\theta$, i.e., $\nabla_\theta^2\Lcal(\theta;x,w)\succ 0$
or equivalently $\nabla_\theta^2\Lcal(\theta;x)\succ 0$.
\end{itemize}
We should think of $\thetah(x,w)$ as the output of some optimization algorithm, such as gradient descent, being run to convergence on the 
minimization problem~\eqref{eqn:thetah_def}.

From this point on, abusing notation, depending on context we may write $\thetah$ to denote the map $\thetah:\Xcal\times\R^d\rightarrow\Theta$, or may
also write $\thetah$ to denote $\thetah(X,W)$, the random variable obtained by applying this map to the data.

\subsubsection{Calculating the distribution conditioned on the estimator}\label{sec:compute_conditional}

Our next step is to calculate the conditional distribution of $X\giv\thetah$, where $\thetah = \thetah(X,W)$
for random Gaussian noise $W\sim\normal(0,\tfrac{1}{d}\ident_d)$.
As it turns out, it is generally not possible to do this exactly---in the rare degenerate case where $\thetah(X,W)$ may fail to find a SSOSP
of the optimization problem~\eqref{eqn:thetah_def}, we do not know the distribution of $\thetah\giv X$ and therefore cannot calculate the distribution of $X\giv \thetah$.
We will avoid this degeneracy by conditioning on the event that $\thetah(X,W)$ returns a SSOSP.

First, we assume some standard conditions on the parametric family, and a differentiability condition on the model and the regularization function (we
 will also assume implicitly that all the functions defined so far, namely, $\thetah$, $p$, $\Lcal$ and its derivatives,
 are measurable with respect to $\nu_\Xcal\times\leb$ or $\nu_\Xcal$ or $\leb$, as appropriate):
\begin{assumption}[Regularity conditions]\label{asm:family}
The family $\{P_\theta:\theta\in\Theta\}$ and regularization function $\pen(\theta)$ satisfy:
\begin{itemize}
\item $\Theta\subseteq\R^d$ is a convex and open subset;
\item For each $\theta\in\Theta$, $P_\theta$ has density $f(x;\theta)>0$ with respect to the base measure $\nu_\Xcal$;
\item For each $x\in\Xcal$, the function $\theta\mapsto \Lcal(\theta;x)=-\log f(x;\theta)+\pen(\theta)$ is continuously twice differentiable.
\end{itemize}
\end{assumption}
\noindent We are now ready to calculate the conditional distribution of $X\giv\thetah$. 
\begin{lemma}\label{lem:compute_conditional}
Suppose Assumption~\ref{asm:family} holds.
Fix any $\theta_0\in\Theta$, and let $(X,\thetah)$ be drawn from the joint model
\begin{equation}\label{eqn:joint_model_XWthetah}\begin{cases} X\sim P_{\theta_0},\\ W\sim \normal(0,\tfrac{1}{d}\ident_d),\\\thetah = \thetah(X,W).\end{cases}\end{equation}
Suppose the event that $\thetah$ is a SSOSP of $\Lcal(\theta;X,W)$ has positive 
probability. 

Then, conditional on this event, the conditional distribution of $X\giv \thetah$ has density 
\begin{equation}\label{eqn:density_realX}p_{\theta_0}(\cdot\giv \thetah) \propto f(x;\theta_0)\cdot \exp\left\{ - \frac{\norm{\nabla_\theta\Lcal(\thetah;x)}^2}{2\sigma^2/d}\right\}\cdot \det\left(\nabla^2_\theta \Lcal(\thetah;x)\right)\cdot\one{x\in\Xcal_{\thetah}}\end{equation}
with respect to
 the base measure $\nu_\Xcal$,
where
\begin{equation}\label{eqn:density_support}
\Xcal_\theta = \left\{x\in\Xcal : \textnormal{ for some $w\in\R^d$, $\theta=\thetah(x,w)$ is a SSOSP of $\Lcal(\theta;x,w)$}\right\}.\end{equation} 
\end{lemma}
\noindent The proof of this lemma is given  in Appendix~\ref{app:proof_lem:compute_conditional}.
For intuition, we can consider the terms appearing in the calculation~\eqref{eqn:density_realX}: the first term $f(x;\theta_0)$ expresses the original distribution of $X$ (before conditioning), the second term $\exp\{\dots\}$ comes from the density of the multivariate normal distribution of $W$, the third term $\det(\dots)$ arises from a change-of-variables calculation when we move from the joint distribution of $(X,W)$ to that of $(X,\thetah)$, and the final term $\one{x\in\Xcal_{\thetah}}$ handles potential technical issues such as failure to find a SSOSP. In particular, the form of the second term is due to our choice of the multivariate normal distribution for the noise $W$; if we instead chose a different noise distribution, the results of this lemma would still hold if we make the appropriate changes to this second term (and the method would yield the same types of theoretical results as long as the distribution of $W$ is continuous, supported everywhere on $\R^d$, and has similar concentration properties for $\norm{W}$). In this work, we choose a multivariate normal distribution since the outcome of the procedure will therefore be invariant to rotations of the parameter space $\Theta$; in settings where the choice of the basis for $\Theta$ is meaningful (e.g., we expect sparsity), it may be interesting to instead consider a non-rotationally-invariant noise distribution.

\subsubsection{Sampling the copies}\label{sec:sample_copies}
We next need to specify how to sample the copies $\Xt^{(1)},\dots,\Xt^{(M)}$.
Below we describe several different approaches---which one we use will depend on the computational complexity of the problem at hand.

\paragraph{The i.i.d.~sampling case}
In order to construct copies $\Xt^{(1)},\dots,\Xt^{(M)}$ that are exchangeable with the data $X$,
we would like to sample the copies $\Xt^{(1)},\dots,\Xt^{(M)}$ i.i.d. from the density $p_{\theta_0}(\cdot\giv \thetah)$,
which by Lemma~\ref{lem:compute_conditional} specifies the exact conditional distribution of $X\giv\thetah$.
Since $\theta_0$ is unknown we will use $\thetah$ as a plug-in estimator.
Our procedure is the following:
after observing the data $X$, 
\begin{equation}\label{eqn:alg_density_iid}
\begin{cases} \textnormal{Draw $W\sim\normal(0,\tfrac{1}{d}\ident_d)$ and define $\thetah = \thetah(X,W)$.}\\
\textnormal{If $\thetah$ is a SSOSP of $\Lcal(\theta;X,W)$, then draw $\Xt^{(1)},\dots,\Xt^{(M)}\iidsim p_{\thetah}(\cdot\giv \thetah)$,}\\
\textnormal{\hspace{2.17in} otherwise
 return $\Xt^{(1)}= \dots \Xt^{(M)} = X$.}\end{cases}
\end{equation}
Here our estimated density for the conditional distribution of $X\giv\thetah$ is given by 
\begin{equation}\label{eqn:density_hat} p_{\thetah}(x\giv \thetah) \propto f(x;\thetah)\cdot \exp\left\{ - \frac{\norm{\nabla_\theta\Lcal(\thetah;x)}^2}{2\sigma^2/d}\right\}\cdot \det\left(\nabla^2_\theta \Lcal(\thetah;x)\right)\cdot\one{x\in\Xcal_{\thetah}}\end{equation}
with respect to
 the base measure $\nu_\Xcal$. (Lemma~\ref{lem:density_hat}, in Appendix~\ref{app:check_density}, will verify that this expression indeed defines a valid density.)

Of course, in order to implement the sampling algorithm given in~\eqref{eqn:alg_density_iid}, we are implicitly assuming that it is computationally
feasible to generate i.i.d.~samples from $p_{\thetah}(\cdot|\thetah)$. To avoid making this assumption, we next consider a more general 
framework.

\paragraph{The MCMC sampling case}
In the general case where sampling directly from $p_{\thetah}(\cdot\giv \thetah)$ may not be possible, we can instead use MCMC or any other strategy
that ensures exchangeability. To be concrete, 
we will consider two schemes from \cite{besag1989generalized} for constructing the copies with MCMC sampling.
Given $\thetah$, let $\Pi(\cdot;x)$ be any collection of transition distributions, such that the density $p_{\thetah}(\cdot\giv \thetah)$ defines a stationary distribution.
Assume that $\Pi$ defines a reversible Markov chain. Given $\Pi$, we define two different schemes for generating the copies. (See Figure~\ref{fig:mcmc} for an illustration
of these schemes.)
\begin{itemize}
\item {\bf Hub-and-spoke sampler}. Given $X$ and $\thetah$, we sample the copies as follows:
\begin{itemize}
\item Initialize at $X$, and run the Markov chain for $L$  steps to define the ``hub'' $\Xt^*$.
\item Independently for $m=1,\dots,M$, initialize at $\Xt^*$ and  run the Markov chain for $L$  steps to define the ``spoke'' $\Xt^{(m)}$.
\end{itemize}
\item {\bf Permuted serial sampler}. Given $X$ and $\thetah$, we sample the copies as follows:
\begin{itemize}
\item Draw a uniform permutation $\pi$ on $\{0,\dots,M\}$ and find $m^*\in\{0,\dots,M\}$ such that $\pi(m^*)=0$.
\item Initialize at $X$, and run the Markov chain for $Lm^*$  steps, stopping every $L$-th step to define the copies $\Xt^{(\pi(m^*-1))},\dots,\Xt^{(\pi(0))}$.
\item Independently, initialize at $X$, and run the Markov chain for $L(M-m^*)$  steps, stopping every $L$-th step to define the copies $\Xt^{(\pi(m^*+1))},\dots,\Xt^{(\pi(M))}$.
\end{itemize}
\end{itemize}
Later on, we will give concrete examples of how to implement these sampling schemes for specific models.

\begin{figure}[t] 
\begin{tikzpicture}

\begin{scope}[xshift=0cm,yshift=0cm]

\node (X) at (-2,0) {};
\node (Xhub) at (0,0) {};
\node (X1) at (-1.41,1.41) {};
\node (X2) at (0,2) {};
\node (X3) at (1.41,1.41) {};
\node (Xdots) at (2,0) {};
\node (XM2) at (1.41,-1.41) {};
\node (XM1) at (0,-2) {};
\node (XM) at (-1.41,-1.41) {};
\node (Xtext) at (-2,0) {$X$};
\node (Xhubtext) at (0,0) {$\Xt^*$};
\node (X1text) at (-1.41,1.41) {$\Xt^{(1)}$};
\node (X2text) at (0,2) {$\Xt^{(2)}$};
\node (X3text) at (1.41,1.41) {$\Xt^{(3)}$};
\node (Xdotstext) at (2,0) {$\dots$};
\node (XM2text) at (1.41,-1.41) {$\Xt^{(M-2)}$};
\node (XM1text) at (0,-2) {$\Xt^{(M-1)}$};
\node (XMtext) at (-1.41,-1.41) {$\Xt^{(M)}$};
\node (hublabel) at (2,2.5) {\footnotesize\color{red} latent hub};
\draw[line width=0.5mm,shorten >=5pt,shorten <=5pt] (X.east) -- (Xhub.west);
\draw[line width=0.5mm,shorten >=5pt,shorten <=5pt] (X1.south east) -- (Xhub.north west);
\draw[line width=0.5mm,shorten >=5pt,shorten <=5pt] (X2.south) -- (Xhub.north);
\draw[line width=0.5mm,shorten >=5pt,shorten <=5pt] (X3.south west) -- (Xhub.north east);
\draw[line width=0.5mm,shorten >=5pt,shorten <=5pt] (Xdots.west) -- (Xhub.east);
\draw[line width=0.5mm,shorten >=5pt,shorten <=5pt] (XM2.north west) -- (Xhub.south east);
\draw[line width=0.5mm,shorten >=5pt,shorten <=5pt] (XM1.north) -- (Xhub.south);
\draw[line width=0.5mm,shorten >=5pt,shorten <=5pt] (XM.north east) -- (Xhub.south west);
\draw [<-,red] (Xhub.north east) to [out=60,in=210] (hublabel.south west);
\end{scope}

\begin{scope}[xshift=1.5cm,yshift=0cm]

\node (P1) at (0.4,0) {};
\node (P2) at (2,0) {};
\node (P3) at (3.45,0) {};
\node (P4) at (4.8,0) {};
\node (P5) at (6.4,0) {};
\node (P6) at (8,0) {};
\node (P7) at (9.6,0) {};
\node (P2text) at (2,0) {$\Xt^{(2)}$};
\node (P3text) at (3.45,0) {$X$};
\node (P4text) at (4.8,0) {$\Xt^{(1)}$};
\node (P5text) at (6.4,0) {$\dots$};
\node (P6text) at (8,0) {$\Xt^{(M)}$};
\node (P7text) at (9.6,0) {$\Xt^{(3)}$};
\node (Pleft) at (1.5,0) {};
\node (Pright) at (10.1,0) {};

\draw[line width=0.5mm,shorten >=3pt,shorten <=8pt] (P2.east) -- (P3.west);
\draw[line width=0.5mm,shorten >=8pt,shorten <=3pt] (P3.east) -- (P4.west);
\draw[line width=0.5mm,shorten >=8pt,shorten <=8pt] (P4.east) -- (P5.west);
\draw[line width=0.5mm,shorten >=8pt,shorten <=8pt] (P5.east) -- (P6.west);
\draw[line width=0.5mm,shorten >=8pt,shorten <=8pt] (P6.east) -- (P7.west);
\draw[decoration={brace,mirror,raise=20pt},decorate,color=red] (Pleft) -- node[below=25pt] {\footnotesize\color{red} Random permutation of $M+1$ positions} (Pright);
\end{scope}

\end{tikzpicture}
\caption{Left: the hub-and-spoke sampler. Right: the permuted serial sampler. In both diagrams, each thick black line represents running the reversible Markov chain for $L$ steps.}
\label{fig:mcmc}
\end{figure}
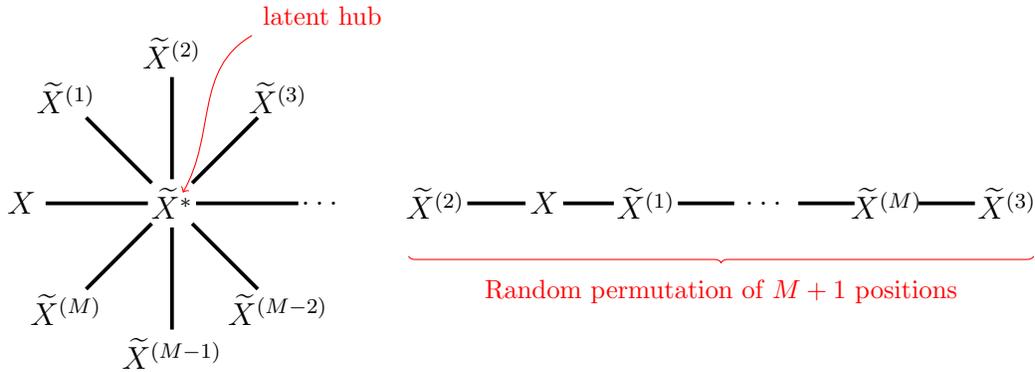

\paragraph{A unified definition}
To generalize our various options (i.i.d.~sampling, hub-and-spoke MCMC sampling, and permuted serial MCMC sampling),
we will write $\Pt_M(\cdot;X,\thetah)$ to denote 
the distribution of the collection of copies $(\Xt^{(1)},\dots,\Xt^{(M)})$ conditional on $X$ and $\thetah$.
For all three cases, our aCSS procedure for sampling the copies is the following:
\begin{equation}\label{eqn:aCSS_alg}
\begin{cases} \textnormal{Draw $W\sim\normal(0,\tfrac{1}{d}\ident_d)$ and define $\thetah = \thetah(X,W)$.}\\
 \textnormal{If $\thetah$ is a SSOSP of $\Lcal(\theta;X,W)$, then draw $(\Xt^{(1)},\dots,\Xt^{(M)})\sim \Pt_M(\cdot;X,\thetah)$,}\\
\textnormal{\hspace{2.17in} otherwise
 return $\Xt^{(1)}= \dots \Xt^{(M)} = X$.}\end{cases}
\end{equation}

In the i.i.d.~sampling case, $\Pt_M(\cdot;X,\thetah)$ is simply equal to sampling from the product density $p_{\thetah}(\cdot|\thetah)\times \dots \times p_{\thetah}(\cdot|\thetah)$,
and therefore depends on $\thetah$ but not on $X$,
while for the two MCMC samplers, there is dependence between the data and the copies
even after conditioning on $\thetah$ (although, if the chain length $L$ is sufficiently long, we would expect this dependence to be weak).
Despite this dependence, all three of these sampling schemes satisfy the following exchangeability condition:
for all $\theta\in\Theta$ with $\nu_\Xcal(\Xcal_\theta)>0$,
\begin{equation}\label{eqn:exch_some_theta}
\textnormal{\begin{tabular}{c}
If $X\sim p_\theta(\cdot\giv \theta)$ and $(\Xt^{(1)},\dots,\Xt^{(M)})\mid X\sim  \Pt_M(\cdot;X,\theta)$, then\\
 the random vector $(X,\Xt^{(1)},\dots,\Xt^{(M)})$ is exchangeable.\end{tabular}}\end{equation}
Note that $\Pt_M(\cdot;X,\theta)$ replaces \emph{all} instances of $\thetah$ in the definition of $\Pt_M(\cdot;X,\thetah)$ with $\theta$'s. Of course, it may be of interest to examine other sampling schemes, aside from the three described above. Our theoretical results
below apply to any algorithm of the form~\eqref{eqn:aCSS_alg} as long as the distribution $\Pt_M$ for drawing the copies
is chosen to satisfy~\eqref{eqn:exch_some_theta}.

\section{Theoretical results}

In this section, we
present our main result, proving a bound on the excess Type I error of any aCSS testing procedure.

\subsection{Main result: Type I error bound}
Before presenting the theorem, we will need a few more assumptions on the model
and on the noisy estimator $\thetah$.  
First,
we need to assume
that $\thetah$ is (typically) an accurate estimator of the unknown true $\theta_0$, and that $\thetah$ will (typically) return a SSOSP for the optimization problem~\eqref{eqn:thetah_def}:
\begin{assumption}\label{asm:thetah}
For any $\theta_0\in\Theta$, the estimator $\thetah:\Xcal\times\R^d\rightarrow\Theta$ satisfies
\begin{equation}\label{eqn:asm:thetah}
\pr{\norm{\thetah(X,W)-\theta_0}\leq r(\theta_0)\textnormal{, and $\thetah(X,W)$ is a SSOSP of $\Lcal(\theta;X,W)$}}\geq 1 - \delta(\theta_0),
\end{equation}
where the probability is taken with respect to the distribution  $(X,W)\sim P_{\theta_0} \times \normal(0,\tfrac{1}{d}\ident_d)$.
\end{assumption}
\noindent For many parametric families,
the maximum likelihood estimator (or a penalized MLE) is typically shown to satisfy this type of condition with $r(\theta_0)=\bigot(n^{-1/2})$  (here $\bigot(\cdot)$ denotes that the scaling holds up to powers of $\log n$).
  This assumption has essentially the same flavor,
except that our estimator $\thetah$ is a random perturbation
of the penalized MLE. We discuss this assumption in more detail in Appendix~\ref{app:proofs_for_examples}.

Next, we place some assumptions on the derivatives of the log-likelihood.
Let $H(\theta;x) = -\nabla^2_\theta\log f(x;\theta)$ and
let $H(\theta) = \Ewrt{\theta_0}{H(\theta;X)}$ (in particular, $H(\theta_0)$ is the Fisher information).
\begin{assumption}\label{asm:hessian}
For any $\theta_0\in\Theta$, the expectation $H(\theta)$ exists for all $\theta\in\ball(\theta_0,r(\theta_0))\cap\Theta$, and furthermore
\begin{equation}\label{eqn:asm:hessian_conc1}\Ewrt{\theta_0}{\sup_{\theta\in\ball(\theta_0,r(\theta_0))\cap\Theta}r(\theta_0)^2\cdot\big(\lambda_{\max}(H(\theta)-H(\theta;X))\big)_+} \leq \eps(\theta_0)\end{equation}
and
\begin{equation}\label{eqn:asm:hessian_conc2}\Ewrt{\theta_0}{\exp\left\{\sup_{\theta\in\ball(\theta_0,r(\theta_0))\cap\Theta}r(\theta_0)^2\cdot\big(\lambda_{\max}(H(\theta;X)-H(\theta))\big)_+\right\}} \leq e^{\eps(\theta_0)}.\end{equation}
\end{assumption}
\noindent 
Here $r(\theta_0)$ is the same constant as appears in Assumption~\ref{asm:thetah} (which, as mentioned above, will scale as 
$r(\theta_0)=\bigot( n^{-1/2})$ in many settings).
To interpret our assumption, we  note that assumptions of the form
\[\norm{H(\theta;X)-H(\theta)} = \bigop(n^{1/2})\]
are standard for establishing classical results such as asymptotic normality of the MLE; even with a bound as weak as $r(\theta_0) = \littleo(n^{-1/4})$,
 this type of assumption will immediately imply that the
first bound~\eqref{eqn:asm:hessian_conc1} holds. 
However, this type of condition is not quite sufficient for the theoretical arguments we need to establish, and we instead need 
the condition~\eqref{eqn:asm:hessian_conc2}, which implies the same
rate of convergence but with stronger control of the tails.

With our assumptions in place, we state the main result, which bounds the distance to exchangeability---and therefore, the Type I error---of any aCSS 
procedure.
\begin{theorem}\label{thm:main}
Suppose Assumptions~\ref{asm:family},~\ref{asm:thetah}, and~\ref{asm:hessian} all hold.
After observing the data $X$, suppose we run the aCSS algorithm~\eqref{eqn:aCSS_alg},
where the distribution $\Pt_M$ is chosen to satisfy~\eqref{eqn:exch_some_theta}.
Then, if $X\sim P_{\theta_0}$ for some $\theta_0\in\Theta$, 
 the copies $\Xt^{(1)},\dots,\Xt^{(M)}$ are approximately exchangeable with $X$, satisfying
\[\dex(X,\Xt^{(1)},\dots,\Xt^{(M)}) \leq 3\sigma \cdot r(\theta_0) + \delta(\theta_0)+ \eps(\theta_0).\]
In particular, this implies that for any predefined test statistic $T:\Xcal\rightarrow\R$ and rejection threshold $\alpha\in[0,1]$,
the p-value defined in~\eqref{eqn:pval} satisfies
\[\pr{\textnormal{pval}_T(X,\Xt^{(1)},\dots,\Xt^{(M)})\leq \alpha} \leq \alpha + 3\sigma \cdot r(\theta_0) + \delta(\theta_0)+ \eps(\theta_0).\]
\end{theorem}
\noindent The proof of this theorem is given in Appendix~\ref{app:proof_thm:main}.

\subsection{The asymptotic view}\label{sec:theory_asymptotic_view}
The theoretical guarantee given in Theorem~\ref{thm:main} is nonasymptotic, but
 it typically implies asymptotic control of the Type I error.
In particular, in many standard settings where the observed data
arises from an independent sample of size $n$, 
the terms $r(\theta_0)$, $\delta(\theta_0)$, and $\eps(\theta_0)$ are all vanishing, and in particular
we will expect to see  $r(\theta_0)=\bigot( n^{-1/2})$. Thus, if we choose noise level $\sigma \asymp n^a$ for some $a<\tfrac{1}{2}$, this
will lead to asymptotic Type I error control, i.e.,  $\pr{\textnormal{pval}\leq \alpha}  = \alpha + \littleo(1)$. 

 Furthermore, the Type I error bound in Theorem~\ref{thm:main} gives insight into the role of approximate (or asymptotic) sufficiency in the method--- $\thetah(X,W)$ is essentially a MLE (assuming $\sigma=\littleo(n^{1/2})$ as before)---this is because the size of the perturbation of the negative log-likelihood,
 $\norm{\nabla_\theta\Lcal(\theta;X,W)-\nabla_\theta\Lcal(\theta;X)}=\sigma \norm{W} = \littleo(n^{1/2})$, is vanishing relative to $\norm{\nabla_\theta\Lcal(\theta_0;X)}\asymp n^{1/2}$. Thus under standard assumptions,  $\thetah(X,W)$
 is asymptotically efficient, and inherits the asymptotic sufficiency properties of the MLE.
At a high level, this means that the distributions $p_{\theta_0}(\cdot\giv\thetah)$ of $X\giv \thetah$ and $p_{\thetah}(\cdot\giv\thetah)$ of $\Xt^{(m)}\giv\thetah$
are asymptotically equal (i.e., the total variation distance between them is vanishing), 
leading to asymptotic exchangeability between $X$ and its copies, and consequently an asymptotic Type I error bound
at the nominal level $\alpha$ as shown in Theorem~\ref{thm:main}.

\subsection{Choosing $\sigma$}\label{sec:choosing_sigma}
It may seem odd that we have advocated for $\sigma>0$ and yet the Type 1 error bound in our main result gets \emph{worse} as $\sigma$ increases. Indeed, increasing $\sigma$ will generally degrade the Type 1 error of aCSS testing due to the fact that, as $\sigma$ is increased, the method moves farther from conditioning on a sufficient statistic. And in fact, taking the limit as $\sigma\rightarrow 0$ in Theorem 1 gives the tightest possible Type 1 error bound (only Assumption~\ref{asm:thetah} depends on $\sigma$, and in general we would expect it to be even more plausible for smaller $\sigma$). In addition, as discussed in Section~\ref{sec:draw_thetah} and in Appendix~\ref{app:optimization}, increasing $\sigma$ can decrease the probability of finding an SSOSP for the optimization \eqref{eqn:thetah_def}, which will not negatively impact the Type 1 error, but will decrease the power of the test by increasing the probability of returning a p-value of 1. However, despite these two downsides, there are two critical reasons why it is advantageous, and arguably necessary, to take $\sigma>0$, and this is why we allow for it in Theorem~\ref{thm:main}.

First, note that if we took $\sigma=0$, aCSS would need to sample from a distribution supported on a level set of the MLE function of $X$. This level set is a low-dimensional (and hence measure-zero) subset of $\Xcal$, and thus it is generally computationally intractable to sample from exactly. There is some work on sampling a random variable conditional on the value of a function of it (e.g., \cite{diaconis2013sampling}), but only in very limited settings. Thus in most applications of aCSS, we are not aware of a computationally tractable approach that does not take $\sigma>0$. Once we accept that $\sigma>0$ is computationally necessary, the choice of its value represents a power-computation trade-off within the MCMC samplers we propose in this paper. This trade-off is discussed more in Appendix~\ref{app:sampling_randomizations}, but essentially as $\sigma$ approaches zero, it will take increasingly many MCMC steps (and associated computation) for the sampler to move ``away" from the original $X$ towards conditional independence. The more the sampler can move ``away" from $X$, the higher the power of aCSS testing will tend to be, since a small p-value is obtained exactly when $X$ stands out among the sampled copies.

Second, for models in which the MLE is sufficient for $\theta$ as well as for the parameters in a higher-dimensional supermodel of $\{P_{\theta}:\theta\in\Theta\}$ (e.g., in the logistic regression example the MLE is equivalent to $X$ itself and thus is sufficient for all the parameters in any model), taking $\sigma=0$ would lead to a completely powerless test for all alternatives in that supermodel. Exactly how large $\sigma$ needs to be to break this degeneracy will likely need to be worked out on a case-by-case basis, and we defer a general treatment to future work. 
However, we see in Section~\ref{sec:examples} that for the logistic regression setting, described
in Model Class~\ref{cl:dwac}, we can easily achieve high power with a $\sigma$ value that is still sufficiently small to have no visible impact on the Type 1 error.


\section{Examples}\label{sec:examples}
To provide further insight into the generality and practicality of aCSS testing, we establish that the necessary assumptions hold for four specific models.
Example~\ref{ex:canonical_glm} (generalized linear models with canonical parameters) is an example
of a regression model containing data with associated covariates, as discussed in Model Class~\ref{cl:dwac}.
Example~\ref{ex:behrens_fisher} (the Behrens--Fisher problem) and Example~\ref{ex:gaussian_spatial} (a Gaussian spatial process) are both examples of curved exponential families,
discussed earlier in Model Class~\ref{cl:cef}.
Example~\ref{ex:multivariate_t} (a multivariate t distribution) is a heavy-tailed model, and is thus an instance of Model Class~\ref{cl:htm}.

In each case, we will see that the assumptions of Theorem~\ref{thm:main} are satisfied
with $r(\theta_0) = \bigot(n^{-1/2})$, and with vanishing $\delta(\theta_0)$ and $\eps(\theta_0)$.
In particular, choosing a noise level $\sigma\asymp n^a$ for any $a<\tfrac{1}{2}$ is sufficient to ensure that the Type I error is asymptotically bounded by the nominal
level $\alpha$. We will then show simulation results for each of the four examples in Section~\ref{sec:examples_sims} below.

\subsection{Canonical generalized linear models (GLMs)}\label{sec:ex1}
\begin{example}\label{ex:canonical_glm}\normalfont
We begin with the setting of a generalized linear model (GLM) with canonical parameters.
Consider a logistic regression model with covariates $Z_i\in\R^d$ associated with each $X_i\in\R$, so that 
\[f(x;\theta) = \prod_{i=1}^n \left(\frac{e^{Z_i^\top\theta}}{1+e^{Z_i^\top\theta}}\right)^{x_i}\cdot\left(\frac{1}{1+e^{Z_i^\top\theta}}\right)^{1-x_i},\]
parametrized by $\theta\in\Theta = \R^d$. (We interpret $f(x;\theta)$ as a density with respect to the base
measure $\nu_\Xcal$ on $\Xcal=\R^n$ that places mass 1 on each point $x\in\{0,1\}^n$.) We can rewrite this in the notation of a generalized linear model (GLM),
\[f(x;\theta) = \exp\left\{ x^\top Z \theta - \sum_{i=1}^n \log(1 + e^{Z_i^\top\theta})\right\},\]
where $Z\in\R^{n\times d}$ is the matrix with rows $Z_i$.
As discussed above, for $X\sim P_\theta$, the random vector $S(X) = Z^\top X\in\R^d$ provides a sufficient statistic; however,
if the rows $Z_i$ are in general position, then $Z^\top X$ will determine $X\in\{0,1\}^n$ uniquely, meaning that
$X$ is no longer random after we condition on $S(X) = Z^\top X$. In other words, co-sufficient sampling (CSS) would lead to zero power, and we
therefore need to turn to aCSS testing.

More generally, we can consider any canonical GLM, of the form
\[f(x;\theta) = \exp\left\{ x^\top Z \theta - \sum_{i=1}^n a(Z_i^\top\theta)\right\},\]
with respect to some  base measure $\nu_\Xcal = \mu\times\dots\times \mu$ on $\Xcal = \R^d$, where $\mu$ is a measure
on $\R$. The function $a$ is known as the {\em partition function}, and is strictly convex on its domain, which must be an open subset of $\R$.
As for logistic regression, $Z^\top X$ is a sufficient statistic for $X\sim P_\theta$, but in the case of a discrete distribution (e.g., Poisson),
CSS will again lead to zero power and so we should instead consider aCSS.

Suppose that the sample size $n$ tends to infinity,
while the parameter $\theta_0$ is held constant (in particular, this implies that dimension $d$ is held constant---we
leave the high-dimensional setting for future work). 
For this example, and all the others below, we will consider the unpenalized version of the method, i.e., $\pen(\theta)\equiv 0$.
Assume the covariates are entrywise bounded, i.e., $\max_{i,j}\norm{Z_{ij}}_{\infty}$ is bounded by a constant,
and $\frac{1}{n} Z^\top Z\succeq \lambda_0 \ident_d$ for a positive constant $\lambda_0$.
We treat the covariates as fixed (i.e., the theory holds conditional on the covariates).
Then, as we will show in Appendix~\ref{app:proofs_for_examples}, for an appropriately-chosen initial estimator
 this example satisfies Assumptions~\ref{asm:family},~\ref{asm:thetah}, and~\ref{asm:hessian}
with $r(\theta_0)=\bigot(n^{-1/2})$, $\delta(\theta_0)=\bigo(n^{-1})$, and $\eps(\theta_0) = 0$.
\end{example}

\subsection{The Behrens--Fisher problem}\label{sec:ex2}
\begin{example}\label{ex:behrens_fisher}\normalfont
Next we consider the classical example of the Behrens--Fisher problem. 
Consider data 
\[X^{(0)}_1,\dots,X^{(0)}_{n^{(0)}} \iidsim \normal(\mu^{(0)},\gamma^{(0)}),\quad X^{(1)}_1,\dots,X^{(1)}_{n^{(1)}}\iidsim \normal(\mu^{(1)},\gamma^{(1)}),\]
with the two samples drawn independently.
We are interested in testing the null hypothesis $H_0: \mu^{(0)} = \mu^{(1)}$, and therefore the family of distributions
can be parameterized by $\theta = (\mu,\gamma^{(0)},\gamma^{(1)}) \in \Theta = \R\times\R_+\times\R_+\subseteq\R^3$, yielding a  family
$\{P_\theta : \theta\in\Theta\}$ where $P_\theta$ has density
\[f(x;\theta)  = f(x;(\mu,\gamma^{(0)},\gamma^{(1)})) =  \prod_{i=1}^{n^{(0)}} \frac{1}{\sqrt{2\pi\gamma^{(0)}}}e^{-(X^{(0)}_i-\mu)^2 / 2\gamma^{(0)}} \cdot \prod_{i=1}^{n^{(1)}} \frac{1}{\sqrt{2\pi\gamma^{(1)}}}e^{-(X^{(1)}_i-\mu)^2 / 2\gamma^{(1)}}\]
with respect to the Lebesgue measure on $\Xcal = \R^{n^{(0)}+n^{(1)}}$.

This problem is an example of a curved exponential family (Problem Domain~\ref{cl:cef}), 
for which the larger model is parametrized by $(\mu^{(0)},\gamma^{(0)},\mu^{(1)},\gamma^{(1)})$---note that 
the constraint $\mu^{(0)} = \mu^{(1)}$ is a nonlinear constraint once
we transform to the canonical parameters, which are given by $(\gamma^{(\ell)})^{-1}\mu^{(\ell)}$, $(\gamma^{(\ell)})^{-1}$ for each $\ell\in\{0,1\}$.
For this problem, under the null model (i.e., parametrized by $\theta=(\mu,\gamma^{(0)},\gamma^{(1)})$), the minimal sufficient statistic is nonetheless four-dimensional---for example,
the sample means and sample standard deviations of $\{X^{(0)}_i\}$ and of $\{X^{(1)}_i\}$ form a minimal sufficient statistic.
Of course, this statistic is also sufficient for the larger alternative model (where $\mu^{(0)}\neq \mu^{(1)}$); once we condition on this sufficient statistic,
the remaining randomness in the data carries no information about the parameters $\mu^{(0)}$ and $\mu^{(1)}$. Therefore, CSS would lead to 
a completely powerless procedure, and we instead turn to aCSS.
(As mentioned earlier in Section~\ref{sec:relatedwork}, \cite{lillegaard2001tests} mention the possibility of, but do not pursue, an aCSS-like procedure
for this specific example.) 

Suppose that the sample size $n$ tends to infinity,
while the parameter $\theta_0$ is held constant and  the ratio $\frac{\max\{n^{(0)},n^{(1)}\}}{\min\{n^{(0)},n^{(1)}\}}$ is bounded
by a constant.
Then, as we will show in Appendix~\ref{app:proofs_for_examples}, for an appropriately-chosen initial estimator
 this example satisfies Assumptions~\ref{asm:family},~\ref{asm:thetah}, and~\ref{asm:hessian}
with $r(\theta_0)\asymp \bigot(n^{-1/2})$, $\delta(\theta_0)\asymp \bigo(n^{-1})$, and $\eps(\theta_0) = \bigot(n^{-1})$.
\end{example}

\subsection{A Gaussian spatial process}\label{sec:ex3}
\begin{example}\label{ex:gaussian_spatial}\normalfont
For our next example, we will work in a dependent data setting---unlike the other three examples, we do not have independent observations.
Our model is a Gaussian spatial process. 
Suppose that $X\in\R^n$ is distributed according to a multivariate Gaussian,
\[X \sim \normal(0,\Sigma_\theta),\]
where the covariance matrix $\Sigma_\theta$ is parametrized by a scalar $\theta\in\R$. Specifically, we will consider
a spatial Gaussian process where
\[(\Sigma_\theta)_{ij} = \exp\left\{ - \theta \cdot D_{ij} \right\},\]
where $(D_{ij})\in\R^{n\times n}$ is a pairwise distance matrix among $n$ spatial points. In other words, we can think of
the observation $X_i$ as corresponding to a location
$z_i\in\R^k$ for some ambient dimension $k$, and the correlation between $X_i$ and $X_j$ is a decaying function of the distance between locations $z_i$ and $z_j$, i.e., $D_{ij}=\norm{z_i-z_j}$.
We assume that the distances $D_{ij}$ are known, and the parameter $\theta\in \Theta=(0,\infty)\subseteq\R$ is the only unknown.
This example, like Example~\ref{ex:behrens_fisher}, is an instance of a curved exponential family. In this case, the larger model is given by $X\sim \normal(0,\Omega^{-1})$,
where the inverse covariance $\Omega$ is the canonical parameter. The nonlinear constraints introduced by the spatial model take the form
\[(D_{ij})^{-1} \log (\Omega^{-1})_{ij} =(D_{k\ell})^{-1} \log (\Omega^{-1})_{kl}\]
for all indices $i,j,k,\ell$ (since the expression on each side of this equation should be equal to the same value $\theta$).
As in Example~\ref{ex:multivariate_t}, the minimal sufficient statistic for our curved exponential null model is the same as that for the larger exponential family---in this case, it is given by the (uncentered) sample covariance---and therefore CSS would result in a powerless procedure
for testing against any mean-zero multivariate Gaussian alternative.

Now we turn to aCSS for this example. In this setting, the distribution $P_\theta$ has density
\[f(x;\theta) = \frac{1}{(2\pi)^{n/2}\det(\Sigma_\theta)^{1/2}}e^{-x^\top \Sigma_\theta^{-1}x / 2},\]
with respect to the Lebesgue measure on $\R^n$.
The negative log-likelihood $\theta\mapsto - \log f(x;\theta)$ is therefore nonconvex, due to the nature of the 
map $\theta\mapsto \Sigma_\theta$.

 It is known, however, that in the special case where the locations $z_i$ are on a regular integer 
lattice, standard results such as asymptotic normality of the MLE can be obtained \citep{bachoc2014asymptotic},
and so we will work in this setting. Consider the integer grid $\{z_1,\dots,z_n\} = \{1,\dots,N\}^k$, where $n=N^k$.
As above, the distances $D_{ij}$ are given by $\norm{z_i-z_j}$.
Suppose that the grid size $N$ tends to infinity,
while the dimension $k$ and the parameter $\theta_0$ are held constant.
Then, as we will show in Appendix~\ref{app:proofs_for_examples}, for an appropriately-chosen initial estimator
 this example satisfies Assumptions~\ref{asm:family},~\ref{asm:thetah}, and~\ref{asm:hessian}
with $r(\theta_0) = \bigot(n^{-1/2})$, $\delta(\theta_0)=\bigo(n^{-1})$, and $\eps(\theta_0)=\bigot(n^{-1/2})$.
\end{example}

\subsection{The multivariate t distribution with unknown covariance}\label{sec:ex4}
\begin{example}\label{ex:multivariate_t}\normalfont
Our last example will demonstrate that our methodology can be applied even in settings where the data is extremely heavy-tailed---specifically,
the multivariate t distribution. 
We consider a setting with $n$ i.i.d.~draws from a zero-mean multivariate t distribution,
\[X_i \iidsim t_\gamma(0,\theta^{-1}),\]
where $\theta^{-1}\in\R^{k\times k}$ is an unknown covariance matrix while $\gamma>0$ is the known degrees-of-freedom parameter.
(Breaking with standard notation, we will use a lowercase $\theta$ to denote a matrix parameter, to agree with our notation throughout this paper.)
Our family of distributions is therefore given by $\{P_\theta:\theta\in\Theta\}$, where
$\Theta\subseteq\R^{k\times k}$ is the set of positive definite $k\times k$ matrices. We can view $\Theta$ as a convex open subset of $\R^d$ with $d=\frac{k(k+1)}{2}$, 
by considering the upper triangle of a positive definite matrix $\theta$. The density is
\[f(x;\theta) = \prod_{i=1}^n c_{k,\gamma}\det(\theta)^{1/2} \left(\gamma + x_i^\top \theta x_i\right)^{-\frac{\gamma+k}{2}},\]
with respect to the Lebesgue measure on $\Xcal = (\R^k)^n$,
where $c_{k,\gamma}$ depends only on the dimension $k$ and the degrees-of-freedom parameter $\gamma$, and not on the unknown parameter $\theta$.
Unlike a GLM, we cannot write the log-density $\log f(x;\theta)$ in the form (function of $x$)$\cdot$(function of $\theta$).
In fact, we can see that, up to permutation and/or multiplication by $-1$ of the data points $i=1,\dots,n$, the data $X$ itself is a minimal sufficient statistic for $\theta$, so there is no sufficient statistic that would not
essentially fully specify the data. Thus for instance, CSS testing would be powerless against any i.i.d. alternative that is invariant to reflection through the origin.
However, the approximate sufficiency framework is well-suited for this example.

Suppose that the sample size $n$ tends to infinity,
while the degrees-of-freedom parameter $\gamma$ and the unknown matrix parameter $\theta_0$ are held constant (in particular, this implies that the dimension $k$ is held constant---we
leave the high-dimensional setting for future work).
Then, as we will show in Appendix~\ref{app:proofs_for_examples}, for an appropriately-chosen initial estimator
 this example satisfies Assumptions~\ref{asm:family},~\ref{asm:thetah}, and~\ref{asm:hessian}
with $r(\theta_0)=\bigot(n^{-1/2})$, $\delta(\theta_0) = \bigo(n^{-1})$, and $\eps(\theta_0)= \bigot(n^{-1/2})$.
\end{example}

\subsection{Simulations}\label{sec:examples_sims}
We now demonstrate the performance of aCSS for each of the four examples described above; code to reproduce
the simulations is available at \url{http://www.stat.uchicago.edu/~rina/code/aCSS.zip}.
We will first show two examples in Section~\ref{sec:examples_sim_with_score} with relatively simple parametric alternative models, for which competing methods exist; in these examples, we will see aCSS testing is as powerful as the most powerful established method, namely, the score test. Then, in Section~\ref{sec:examples_sim_without_score}, we will consider two more complex examples 
 exhibiting alternative models which elude standard approaches, and for which we are unaware of any existing test that would be powerful; we will see that aCSS testing can be powerful in such settings through the choice of a relatively sophisticated test statistic that fully leverages the particular alternative model. 
 
 For both types of examples, we will also see that the aCSS test is empirically valid (the rejection probability is almost exactly the nominal level $\alpha=0.05$ under the null hypothesis) and that it has only slightly less power than an oracle method---this oracle method is given extra information about the distribution that reduces the composite null to a simple null, and computes a p-value~\eqref{eqn:pval} by applying the same statistic function $T$ as aCSS to $M$ copies $\Xt^{(m)}$ drawn independently (unconditionally) from that simple null.

 \subsubsection{Simulations with a parametric alternative}\label{sec:examples_sim_with_score}

We use Examples~\ref{ex:behrens_fisher} (Behrens--Fisher) and \ref{ex:multivariate_t} (multivariate t) to demonstrate similar power between the aCSS test and the score test under parametric alternatives. The results, plotted in  Figure~\ref{fig:examples_with_score}, show the aCSS tests have very similar power to both the oracle and score tests. The simulation setups for the two examples are summarized below; the choice of the proposal distributions for the MCMC samplers, and chain lengths $L$, are described in detail in Appendix~\ref{app:compute}.

\begin{figure}\centering
\includegraphics[width=0.49\textwidth]{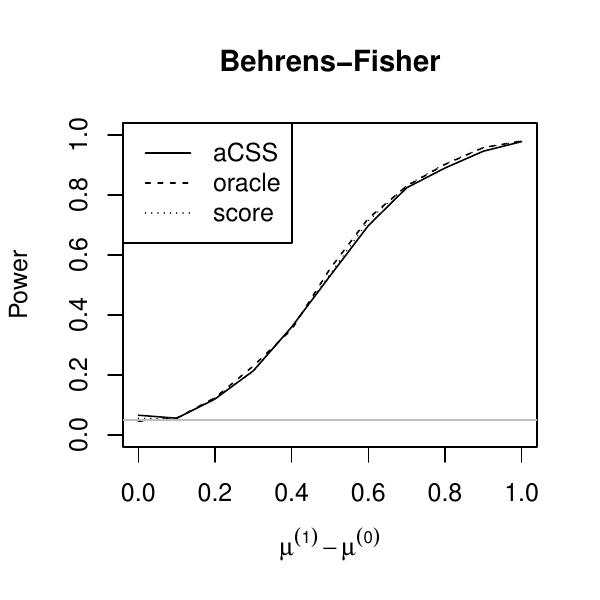}
\includegraphics[width=0.49\textwidth]{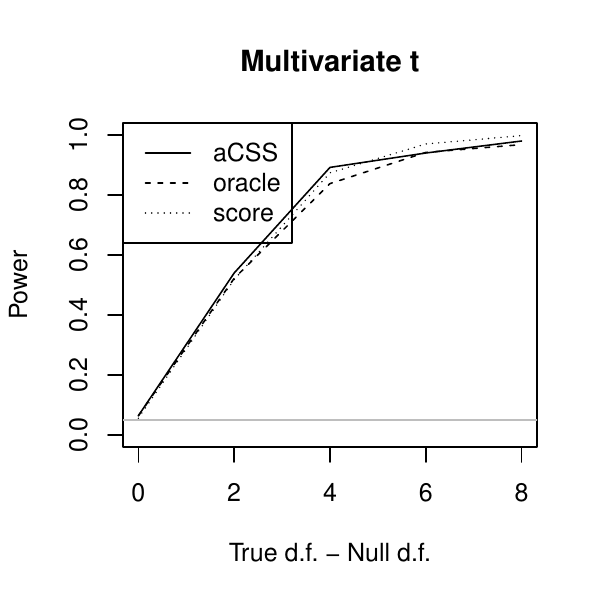}
\caption{Power of the aCSS test compared to an unconditional oracle that knows the (simple) null hypothesis, and compared also to the score test, for the two examples discussed in Section~\ref{sec:examples_sim_with_score}. The aCSS test controls the Type I error at the nominal 5\% level (dotted line) under the null (represented by 0 on the x-axis in each plot), and has very similar power to the oracle and score test under the alternatives. Each point represents 500 independent replications, with the maximum standard error $\approx 2\%$ and the standard error at the left edge of each plot (under the null) below $1\%$.}
\label{fig:examples_with_score}
\end{figure}

\paragraph{ Example~\ref{ex:behrens_fisher} (Behrens--Fisher)} For the Behrens--Fisher example,
the alternative model is as described in Section~\ref{sec:ex2} but with $(\mu^{(0)},\mu^{(1)},\gamma^{(0)}, \gamma^{(1)})$ unconstrained in $\R\times\R\times\R_+\times\R_+$.
\begin{itemize}
    \item To generate the data,
    we take $n^{(0)}=n^{(1)}=50$, $\mu^{(0)}=0$, $\gamma^{(0)}=1$, $\gamma^{(1)}=2$, and $\mu^{(1)} \in \{0,0.1,0.2,\dots,1\}$ (with $\mu^{(1)}=0$ corresponding to the case where the null hypothesis holds). 
    \item The test statistic $T$ (used both for aCSS and for the oracle)
 is given by the absolute difference in sample means between the two halves of the data. 
\item    aCSS is run with the hub-and-spoke sampler with parameters $\sigma^2 = 1$ and $M=500$. The oracle method is given all parameter values except for $\mu^{(1)}$, so that the null $\mu^{(1)}=0$ is simple.
    \end{itemize}
\paragraph{Example~\ref{ex:multivariate_t} (multivariate t)} For the multivariate t example, the alternative model is as described in Section~\ref{sec:ex4} but with the degrees-of-freedom parameter $\gamma$ unknown and unconstrained (aside from being positive). 

\begin{itemize}
    \item To generate the data, we
    let $n=100$, $\theta_0 =\left(\begin{array}{cc} 1 & -0.5\\ -0.5 & 2\end{array}\right)$, and $\gamma = 2$ be the assumed degrees of freedom under the null hypothesis (``d.f.${}_{\text{null}}$''). The distribution of the data is given by $t_{\text{d.f.}}(0,\theta_0^{-1})$, where the degrees of freedom ``d.f.'' is taken from $\{2,4,6,8,10\}$. Therefore d.f. $=2$ represents the case where the null is true, and $ \text{d.f.} - \text{d.f.${}_{\text{null}}$}$ measures the deviation from 
    the null hypothesis. 
    \item The test statistic $T$ (used both for aCSS and for the oracle) is chosen to be the same as for the score test.
   \item  aCSS is run with the hub-and-spoke sampler with parameters $\sigma^2 = 1$ and $M=100$. The oracle method is given all parameter values except for $\gamma$, so that the null $\gamma=2$ is simple.
\end{itemize}

 \subsubsection{Simulations without a parametric alternative}\label{sec:examples_sim_without_score}

We use Examples~\ref{ex:canonical_glm} and \ref{ex:gaussian_spatial} to demonstrate the power of aCSS testing under more complex alternative models for which no existing methods (including the score test) are suitable. The results, plotted in Figure~\ref{fig:examples_without_score}, show the aCSS tests have very similar power to the oracle.
 For the four examples, the settings of the simulation are as follows.  In each case, the choice of the proposal distribution for the MCMC sampler, and chain length $L$, are described in detail in Appendix~\ref{app:compute}.
 
 \begin{figure}\centering
\includegraphics[width=0.49\textwidth]{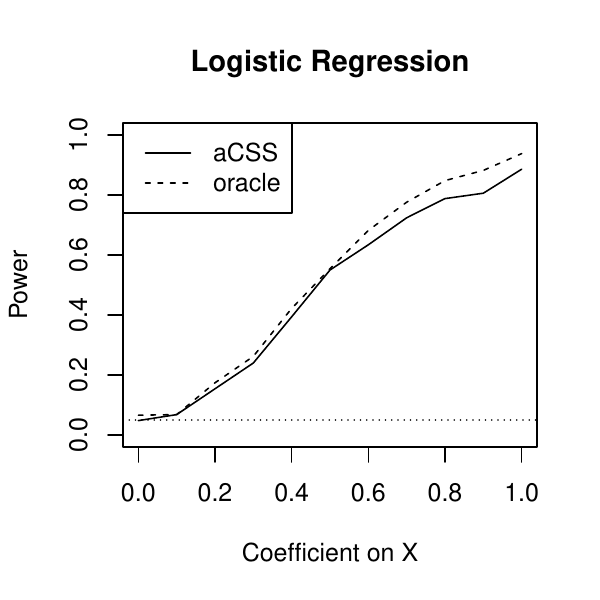}
\includegraphics[width=0.49\textwidth]{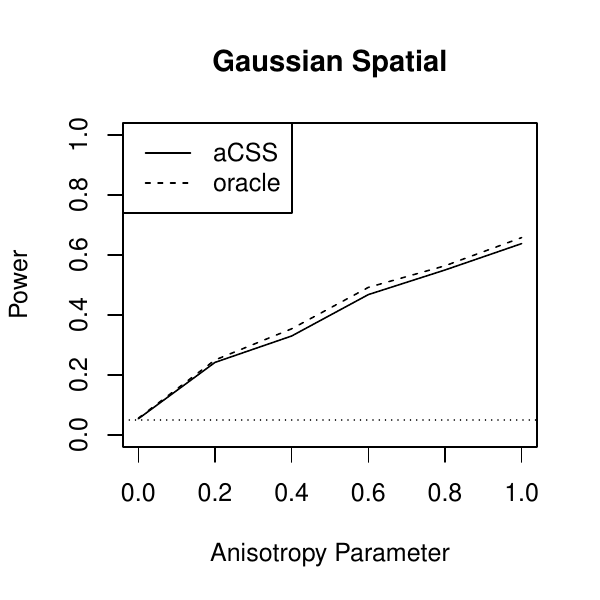}
\caption{Power of the aCSS test compared to an unconditional oracle that knows the (simple) null hypothesis, for the two examples discussed in Section~\ref{sec:examples_sim_with_score}. The aCSS test controls the Type I error at the nominal 5\% level (dotted line) under the null (represented by 0 on the x-axis in each plot), and has very similar power to the oracle and score test under the alternatives. Each point represents 500 independent replications, with the maximum standard error $\approx 2\%$ and the standard error at the left edge of each plot (under the null) below $1\%$.}
\label{fig:examples_without_score}
\end{figure}

\paragraph{Example~\ref{ex:canonical_glm} (logistic regression)} For the logistic regression example, we use aCSS to test a conditional independence hypothesis, so there is a response variable $Y$ that, under the alternative, changes the conditional distribution of $X\giv Z$ given in Section~\ref{sec:ex1}. $Y$ is drawn from a nonparametric model which is well approximated by a  single index model, but does not exactly follow this model.
\begin{itemize}
    \item To generate the data, we take
    $n=100$, and $X\giv Z$ follows 5-dimensional logistic regression with coefficient vector $\theta_0=0.2\cdot \bf{1}$. $Y$'s conditional distribution is given by: $Y\giv (Z,X=0) = f_0(g_0(Z) + \beta_0^\top Z) + \normal(0,1)$ and $Y\giv (Z,X=1) = f_1(g_1(Z) + \beta_1^\top Z) + \normal(0,1)$. We choose $f_0=f_1 = t\mapsto t+0.5t^3$, $g_0=g_1 = z\mapsto 0.5\sum_{j=1}^5 (z_j)_+$, $\beta_0 = c\cdot {\bf e}_1$, and $\beta_1 = c\cdot {\bf e}_5$, where $c\in\{0,0.1,0.2,\dots,1\}$ indicates the signal strength (with $c=0$ corresponding to the null hypothesis). The nonlinearity of $g_0$ and $g_1$ means that the single index model does not exactly describe the conditional distribution of $Y$.
    \item The test statistic $T$ (used both for aCSS and for the oracle) is computed by estimating the coefficient vector on $Z$ in a single index model via sliced inverse regression \citep{doi:10.1080/01621459.1991.10475035} separately on the data sets $\{(Y_i,Z_i) : X_i=0\}$ and $\{(Y_i,Z_i) : X_i=1\}$, respectively (though recall that the single index model does not strictly hold for either data set), and then computing the angle between these estimated coefficient vectors. 
    \item aCSS is run with the hub-and-spoke sampler with parameters $\sigma^2 = 10$ and $M=500$. To implement the oracle method in this example, the oracle is given the distribution of $X\giv Z$, i.e., the true coefficient vector $\theta_0$ for the logistic regression model; under the null hypothesis, $X\giv Z,Y$ follows the same distribution as $X\giv Z$, and thus the oracle is given full knowledge of the distribution of $X\giv Z,Y$ under the null.
\end{itemize}
\paragraph{Example~\ref{ex:gaussian_spatial} (Gaussian spatial)} For the Gaussian spatial process example, we take a 2-dimensional 10$\times$10 integer lattice $\{1,\dots,10\}^2$ for the spatial points.
\begin{itemize}
    \item The distribution of the data is as described in Example~\ref{ex:gaussian_spatial} with the exception that there exists a line $\mathcal{L}$ bisecting the lattice, and for two points $i$ and $j$ whose positions ($z_i$ and $z_j$, respectively) are on opposite sides of $\mathcal{L}$, instead of their covariance being given by $e^{-\theta_0\|z_i-z_j\|}$, it is instead given by $(1-c)e^{-\theta_0\|z_i-z_j\|}$. For instance, the data points could come from soil samples, and $\mathcal{L}$ might be a possible geological ridge reducing the dependence between points on either side of it. In our experiments, $\theta_0 = 0.2$, $\mathcal{L}$ is horizontal with intercept 5.5 so that 50 of the lattice points lie below it and the other 50 lie above it, and $c\in\{0,0.2,\dots,1\}$ is an anisotropy parameter, with $c=0$ indicating an isotropic spatial process so that the null hypothesis holds. 
    \item The test statistic $T$ (used both for aCSS and for the oracle) is computed as follows. We first compute a thresholded kernel matrix $\Delta\in\mathbb{R}^{n\times n}$ with entries $\Delta_{i,j} = e^{-|X_i-X_j|}\one{\|z_i-z_j\|=1}$ and then use $\Delta$ as the kernel matrix for spectral clustering with two clusters. Denoting the two clusters as $S$ and $S^c$, the value of $T$ is then computed as the normalized negative sum of kernel distances between the two groups:
\[ -\left( \frac{1}{\sum_{i\in S^c, j\in S\cup S^c} \Delta_{i,j}} + \frac{1}{\sum_{i\in S, j\in S\cup S^c} \Delta_{i,j}} \right)\sum_{i\in S, j\in S^c} \Delta_{i,j}. \]
    \item aCSS is run with the hub-and-spoke sampler with parameters $\sigma^2 = 1$ and $M=100$. The oracle method is given $\theta_0$, $\mathcal{L}$, and the functional form for $\Sigma$ in terms of $c$, so that the null $c=0$ is simple.
\end{itemize}


\section{Discussion}
Approximate co-sufficient sampling offers a new framework for 
inference on goodness-of-fit and related problems such as conditional independence testing and inference
on target parameters, under mild assumptions on a composite null model. In this section, we will first
revisit the construction of aCSS to develop a deeper intuition for the ideas behind the method,
and will then examine some open questions and directions that remain.

\subsection{The importance of conditioning: comparison to the parametric bootstrap}\label{sec:parboot}
Here we return to the construction of the aCSS method, with new insights obtained from the proof of our main result, Theorem~\ref{thm:main}.
In particular, why is it important to condition on $\thetah$ when we sample the copies?

In the construction of aCSS, after conditioning on $\thetah$, we sample copies $\Xt^{(m)}$ that are approximately exchangeable
with $X$ as long as it holds that $p_{\thetah}(\cdot\giv \thetah)\approx p_{\theta_0}(\cdot\giv \thetah)$. This is because, conditional on $\thetah$,
the copies are sampled from the density $p_{\thetah}(\cdot\giv \thetah)$, while the unknown true null density of $X\giv \thetah$ is instead $p_{\theta_0}(\cdot\giv \thetah)$;
we simply use $\thetah$ as a plug-in estimator of $\theta_0$ to define the distribution from which we sample the copies.
In our proofs, we saw
that aCSS leads to asymptotically valid tests as long as $\dtv(p_{\theta_0}(\cdot\giv \thetah),p_{\thetah}(\cdot\giv \thetah))$ is vanishing.

It is tempting to ask whether the same idea can be used without conditioning on $\thetah$.
That is, since the true data is distributed as $X\sim P_{\theta_0}$ under the null,
can we plug in $\thetah$ for $\theta_0$ and sample the copies $\Xt^{(m)}$ from $P_{\thetah}$? In fact, this non-conditional version of the procedure is simply recovering the parametric bootstrap---and, as we observed in Section~\ref{sec:intro}, the parametric bootstrap may result in inflated Type I error rates in certain settings, depending on the test statistic $T$ that we use.
This is because,
in general, it will not be the case that  $\dtv(P_{\theta_0},P_{\thetah})$ is vanishing, even for $\thetah$ chosen to be the MLE,
and therefore, if we define the copies $\Xt^{(1)},\dots,\Xt^{(M)}$ by sampling (unconditionally) from $P_{\thetah}$, rather than 
from the conditional distribution estimate $p_{\thetah}(\cdot\giv \thetah)$, it will generally be the case that, for some adversarially
chosen test statistic $T(X)$, we may have Type I error that exceeds the nominal level $\alpha$ by a nonvanishing amount.

\subsection{Can we condition on less information?}

More generally, what if we consider conditioning on a different statistic $S=S(X)$ (or a perturbed version $S=S(X,W)$), 
which contains strictly less information about the data $X$ than the (perturbed) MLE $\thetah$?
Of course, the above unconditional distribution is simply the extreme case of this idea, since it conditions
on no information at all. Can we choose $S$ so that it reveals less information about $X$ and thus yields potentially higher
power against the alternative, while still retaining  approximate validity of our test?
To run such a test, we would need to sample the copies
from the plug-in estimated distribution $P_{\thetah}(\cdot\giv S)$ rather than the true conditional null distribution $P_{\theta_0}(\cdot\giv S)$ of $X\giv S$,
and 
in order for the copies to be approximately exchangeable with $X$ under the null, 
we will need this plug-in estimate to be accurate, i.e., $P_{\thetah}(\cdot\giv S)\approx P_{\theta_0}(\cdot\giv S)$---in other words,
$S$ needs to be (approximately) sufficient. As discussed earlier 
in Section~\ref{sec:theory_asymptotic_view}, the perturbed MLE $\thetah$ is asymptotically sufficient under 
standard conditions; since $\thetah$ has the same dimension $d$ as the true parameter $\theta_0$, it is clear that
it is also (asymptotically) a {\em minimal} sufficient statistic.
Therefore, if we choose to condition on any other statistic $S$, if $S$ contains strictly less information about the data $X$ than $\thetah$,
the approximate validity of aCSS would no longer hold.

\subsection{Open questions}
Given our new framework for inference via approximate co-sufficient sampling, many open questions remain 
regarding the properties of this framework, and the settings in which it can be applied.
\begin{enumerate}
\item {\em Power.} How does the choice of statistic $T$ interact with the aCSS framework, to offer the best possible power?
In particular, might it be the case that the choice of $T$ that is most powerful under an aCSS test is not the same
as the $T$ that is most powerful
for an oracle test (with a known point null hypothesis, i.e., $\theta_0$ known)?
\item {\em Computation.} Are there particular algorithms
that enable efficient sampling of the copies $\Xt^{(m)}$, or are there statistics $T$
that allow us to calculate $T(\Xt^{(m)})$ without needing to fully observe $\Xt^{(m)}$---for example, through leveraging symmetries in the model
and the conditional distribution?
\item {\em Additional models.} In addition to the examples
described in this paper, can the aCSS framework be applied to similar problems such as non-canonical generalized linear
models, low-rank regression, or rank-based data? Moving to more challenging settings, does the aCSS
framework extend to latent variable models, errors-in-variables models, or models with missing data?
\item {\em Broader settings.} Can aCSS be applied in a nonparametric setting (perhaps with constraints on the 
statistics $T$ allowed)? Is aCSS robust to model misspecification?  
\item {\em Relaxing regularity conditions and extending to high dimensions.} Can aCSS be applied in settings where the null model is $d$-dimensional, 
but cannot be
represented as a convex and open subset of $\R^d$? For 
 instance, we may have sparsity constraints (with the parameter space given by all $s$-sparse vectors in $\R^p$) or rank constraints (with the parameter space consisting of all matrices with rank at most $r$ in $\R^{a\times b}$). It would appear that any extension of aCSS testing to high dimensions would require incorporating some such low-dimensional structure, in order to ensure the existence of a non-degenerate approximately sufficient statistic, as well as a consistent estimator $\thetah$.
\end{enumerate}

\appendix
\section{Proofs of main results}

Before presenting the proofs of our theoretical results,
we first establish some notation that we will use throughout these proofs.
Let \[\Omega_{\textnormal{SSOSP}} = \left\{(x,w)\in\Xcal\times\R^d : \textnormal{$\thetah(x,w)$ is a SSOSP of $\Lcal(\theta;x,w)$}\right\},\]
and let
\[\Psi_{\textnormal{SSOSP}} = \left\{(x,\theta)\in\Xcal\times\Theta : x\in\Xcal_\theta\right\},\]
where $\Xcal_\theta$ is defined as in~\eqref{eqn:density_support}. The following lemma (proved in Appendix~\ref{app:proof_lem:bijection}) establishes a bijection between these sets:
\begin{lemma}\label{lem:bijection}
Under Assumption~\ref{asm:family},
the map 
\[\psi: (x,w)\mapsto (x,\thetah(x,w))\]
defines a bijection between $\Omega_{\textnormal{SSOSP}}$ and $\Psi_{\textnormal{SSOSP}}$, with inverse 
\[\psi^{-1}:(x,\theta)\mapsto \left(x,-\frac{\nabla_\theta\Lcal(\theta;x)}{\sigma}\right).\]
\end{lemma}

\subsection{Proof of Theorem~\ref{thm:main}}\label{app:proof_thm:main}
Define
 $P^*_{\theta_0}$ to be the distribution of $(X,W)\sim P_{\theta_0}\times\normal(0,\tfrac{1}{d}\ident_d)$ conditional on the event $(X,W)\in\Omega_{\textnormal{SSOSP}}$.
(If this event has probability 0 then the theorem holds trivially, so we can ignore this case.)
Consider the joint distribution
\[\textnormal{Distrib.\,(a):\quad\quad}\begin{cases}
(X,W)\sim P^*_{\theta_0},\\
\thetah = \thetah(X,W),\\
\Xt^{(1)},\dots,\Xt^{(M)}\giv X,\thetah \sim \Pt_M(\cdot;X,\thetah),\end{cases}\]
which is clearly equivalent to the aCSS procedure~\eqref{eqn:aCSS_alg} if we condition on the event $(X,W)\in\Omega_{\textnormal{SSOSP}}$.
On the other hand, on the event that $(X,W)\not\in\Omega_{\textnormal{SSOSP}}$, then by definition 
we set $\Xt^{(1)}=\dots=\Xt^{(M)}=X$, and so exchangeability can only be violated on the event $\Omega_{\textnormal{SSOSP}}$.
Therefore, we have
\begin{equation}\label{eqn:thm:main_dex_for_mixture}\dex(X,\Xt^{(1)},\dots,\Xt^{(M)})  \leq \dex\big(\textnormal{Distribution of $X,\Xt^{(1)},\dots,\Xt^{(M)}$ under Distrib.~(a)}\big).\end{equation}
(We formalize this intuition in Lemma~\ref{lem:dex_for_mixture} in Appendix~\ref{app:dexch_mixture_lem}.)

Next, let $Q^*_{\theta_0}$ be the marginal distribution of $\thetah(X,W)$ under $(X,W)\sim P^*_{\theta_0}$, and define
\[\textnormal{Distrib.\,(b):\quad\quad}\begin{cases}
\thetah\sim Q^*_{\theta_0},\\
X\giv\thetah\sim p_{\theta_0}(\cdot\giv \thetah),\\
\Xt^{(1)},\dots,\Xt^{(M)}\giv X,\thetah \sim \Pt_M(\cdot;X,\thetah).\end{cases}\]
where $p_{\theta_0}(\cdot\giv \thetah)$ is defined as in 
Lemma~\ref{lem:compute_conditional}. By definition of $Q^*_{\theta_0}$, together with Lemma~\ref{lem:compute_conditional},
we can see that the joint distribution of $(X,\Xt^{(1)},\dots,\Xt^{(M)})$
under Distrib.~(b), is equal to its joint distribution under Distrib.~(a), and therefore
\[\dex(X,\Xt^{(1)},\dots,\Xt^{(M)})  \leq \dex\big(\textnormal{Distribution of $X,\Xt^{(1)},\dots,\Xt^{(M)}$ under Distrib.~(b)}\big).\]
Finally, we define another distribution,
\[\textnormal{Distrib.\,(c):\quad\quad}\begin{cases}
\thetah\sim Q^*_{\theta_0},\\
X\giv\thetah\sim p_{\thetah}(\cdot\giv \thetah),\\
\Xt^{(1)},\dots,\Xt^{(M)}\giv X,\thetah \sim \Pt_M(\cdot;X,\thetah).\end{cases}\]
(As mentioned earlier, Lemma~\ref{lem:density_hat} in Appendix~\ref{app:check_density}  will verify that the density $ p_{\thetah}(\cdot\giv \thetah)$
exists almost surely over $\thetah$.) Since $\Pt_M(\cdot;X,\theta)$ was constructed to satisfy~\eqref{eqn:exch_some_theta}, it holds 
that under Distrib.~(c), the random variables $(X,\Xt^{(1)},\dots,\Xt^{(M)})$ are exchangeable (in fact, they are exchangeable conditional on $\thetah$).
Therefore, by definition of $\dex$, we have
\[\dex\big(\textnormal{Distribution of $X,\Xt^{(1)},\dots,\Xt^{(M)}$ under Distrib.~(b)}\big) \leq \dtv\big(\textnormal{Distrib.\,(b)},\textnormal{Distrib.\,(c)}\big),\]
and comparing the definitions of Distrib.~(b) and Distrib.~(c),  it is easy to verify that
\[\dtv(\textnormal{Distrib.\,(b)},\textnormal{Distrib.\,(c)}) = \Ewrt{Q^*_{\theta_0}}{\dtv\big(p_{\theta_0}(\cdot\giv \thetah),p_{\thetah}(\cdot\giv \thetah)\big)}.\]
Combining everything, we have shown that the aCSS procedure~\eqref{eqn:aCSS_alg} satisfies
\begin{equation}\label{eqn:relate_to_tv}\dex(X,\Xt^{(1)},\dots,\Xt^{(M)}) \leq   \Ewrt{Q^*_{\theta_0}}{\dtv\big(p_{\theta_0}(\cdot\giv \thetah),p_{\thetah}(\cdot\giv \thetah)\big)}.\end{equation}
We next need to bound this expected total variation.

We begin with the well known expression for total variation distance between two densities $g,h$, which is given by $\dtv(g,h) = \Ewrt{g}{\left(1 - \frac{h(X)}{g(X)}\right)_+}$. Therefore,
\begin{equation}\label{eqn:dtv_step1}\Ewrt{Q^*_{\theta_0}}{\dtv\big(p_{\theta_0}(\cdot\giv \thetah),p_{\thetah}(\cdot\giv \thetah)\big)}= \Ewrt{Q^*_{\theta_0}}{\Ewrt{p_{\theta_0}(\cdot\giv \thetah)}{\left(1 -\frac{p_{\thetah}(X\giv \thetah)}{p_{\theta_0}(X\giv \thetah)} \right)_+}}.\end{equation}
Recalling the definitions~\eqref{eqn:density_realX} and~\eqref{eqn:density_hat} (and noting in particular
that these two densities have the same support by definition), after calculating normalizing constants we can verify that
\begin{equation}\label{eqn:ratio_of_densities}\frac{p_{\thetah}(x\giv \thetah)}{p_{\theta_0}(x\giv \thetah)} =  \frac{\frac{f(x;\thetah)}{f(x;\theta_0)}}{\Ewrt{p_{\theta_0}(\cdot\giv \thetah)}{\frac{f(X;\thetah)}{f(X;\theta_0)}}}.\end{equation}
Next we take a Taylor series for the function $\theta\mapsto \log f(X;\theta)$. For any $x,\theta$ we can calculate
\[\log \left(\frac{f(x;\theta_0)}{f(x;\theta)}\right)= (\theta_0-\theta)^\top\nabla_\theta\log f(x;\theta)  + \int_{t=0}^1(1-t)\cdot (\theta_0-\theta)^\top \nabla^2_\theta\log f(x;\theta_t)(\theta_0-\theta)\;\mathsf{d}t,\]
where we write $\theta_t = (1-t)\theta_0+t\theta$. 
Therefore, 
 for any $x,x'$ we have
\begin{align*}
\frac{\frac{f(x';\theta)}{f(x';\theta_0)}}{\frac{f(x;\theta)}{f(x;\theta_0)}}
&=\exp\left\{\log \left(\frac{f(x;\theta_0)}{f(x;\theta)}\right) - \log \left(\frac{f(x';\theta_0)}{f(x';\theta)}\right)\right\}\\
&= \exp\bigg\{-(\theta_0-\theta)^\top\left(\nabla_\theta\log f(x';\theta) -\nabla_\theta\log f(x;\theta)\right)  \\
&\hspace{.5in}{}- \int_{t=0}^1(1-t)\cdot (\theta_0-\theta)^\top \left(\nabla^2_\theta\log f(x';\theta_t) - \nabla^2_\theta\log f(x;\theta_t)\right) (\theta_0-\theta)\;\mathsf{d}t\bigg\}\\
&= \exp\bigg\{(\theta_0-\theta)^\top\left(\nabla_\theta\Lcal(\theta;x')-\nabla_\theta\Lcal(\theta;x)\right)  \\
&\hspace{.5in}{}+ \int_{t=0}^1(1-t)\cdot (\theta_0-\theta)^\top \left(H(\theta_t;x') - H(\theta_t;x)\right) (\theta_0-\theta)\;\mathsf{d}t\bigg\}\\
&\leq \exp\bigg\{(\theta_0-\theta)^\top\left(\nabla_\theta\Lcal(\theta;x')-\nabla_\theta\Lcal(\theta;x)\right)  \\
&\hspace{.5in}{} + \frac{1}{2} \sup_{t\in[0,1]}  (\theta_0-\theta)^\top \left(H(\theta_t;x') - H(\theta_t;x)\right) (\theta_0-\theta)\bigg\},\end{align*}
where the  inequality  holds since $\int_{t=0}^1 (1-t)\cdot h(t)\;\mathsf{d}t \leq\int_{t=0}^1 (1-t)\;\mathsf{d}t\cdot \sup_{t\in[0,1]}h(t) = \frac{1}{2} \sup_{t\in[0,1]}h(t)$ for any function $h:\R\rightarrow\R$.
For any $\theta\in\ball(\theta_0,r(\theta_0))\cap\Theta$, it therefore holds that, for all $x,x'$,
\begin{multline*}
\frac{\frac{f(x';\theta)}{f(x';\theta_0)}}{\frac{f(x;\theta)}{f(x;\theta_0)}}
\leq \exp\bigg\{r(\theta_0) \left(\norm{\nabla_\theta\Lcal(\theta;x')}+\norm{\nabla_\theta\Lcal(\theta;x)}\right)  \\
{} + \frac{r(\theta_0)^2}{2} \sup_{\theta'\in \ball(\theta_0,r(\theta_0))\cap\Theta}  \lambda_{\max} \left(H(\theta';x') - H(\theta';x)\right)\bigg\}
\leq \exp\left\{\Delta_1(x,\theta) + \Delta'_1(x',\theta) \right\},\end{multline*}
where we define
\[\Delta_1(x,\theta) = r(\theta_0) \norm{\nabla_\theta\Lcal(\theta;x)} + \frac{r(\theta_0)^2}{2} \sup_{\theta'\in\ball(\theta_0,r(\theta_0))\cap\Theta} \big(\lambda_{\max}(H(\theta') - H(\theta';x))\big)_+,\]
and
\[\Delta'_1(x,\theta) = r(\theta_0) \norm{\nabla_\theta\Lcal(\theta;x)} + \frac{r(\theta_0)^2}{2} \sup_{\theta'\in\ball(\theta_0,r(\theta_0))\cap\Theta}  \big(\lambda_{\max}(H(\theta';x) - H(\theta'))\big)_+.\]
Applying this calculation with $x'=X$, we obtain
\[
 \frac{\frac{f(x;\theta)}{f(x;\theta_0)}}{\Ewrt{p_{\theta_0}(\cdot\giv \theta)}{\frac{f(X;\theta)}{f(X;\theta_0)}}}= \left(\Ewrt{p_{\theta_0}(\cdot\giv \theta)}{\frac{\frac{f(X;\theta)}{f(X;\theta_0)}}{\frac{f(x;\theta)}{f(x;\theta_0)}}}\right)^{-1}\\
\geq \frac{1}{\Ewrt{p_{\theta_0}(\cdot\giv \theta)}{e^{\Delta'_1(X,\theta)}}\cdot e^{\Delta_1(x,\theta)}}
\]
for all $x$ and for all $\theta\in\Theta$ such that $\norm{\theta-\theta_0}\leq r(\theta_0)$.
Returning to~\eqref{eqn:dtv_step1} and~\eqref{eqn:ratio_of_densities}  above, 
and defining $\Ecal_{\textnormal{ball}}$ to be the event that $\norm{\thetah-\theta_0}\leq r(\theta_0)$,
 we therefore have
\begin{align*}&\hspace{-.5in} \Ewrt{Q^*_{\theta_0}}{\dtv\big(p_{\theta_0}(\cdot\giv \thetah),p_{\thetah}(\cdot\giv \thetah)\big)}
=  \Ewrt{Q^*_{\theta_0}}{\Ewrt{p_{\theta_0}(\cdot\giv \thetah)}{\left(1 -\frac{p_{\thetah}(X\giv \thetah)}{p_{\theta_0}(X\giv \thetah)} \right)_+}}\\
&\leq   \prwrt{Q^*_{\theta_0}}{\Ecal_{\textnormal{ball}}^c}+ \Ewrt{Q^*_{\theta_0}}{\Ewrt{p_{\theta_0}(\cdot\giv \thetah)}{\one{\Ecal_{\textnormal{ball}}}\cdot \left(1 -\frac{p_{\thetah}(X\giv \thetah)}{p_{\theta_0}(X\giv \thetah)} \right)_+}}\\
&\leq \prwrt{Q^*_{\theta_0}}{\Ecal_{\textnormal{ball}}^c}+\Ewrt{Q^*_{\theta_0}}{\Ewrt{p_{\theta_0}(\cdot\giv \thetah)}{1 -\frac{1}{\Ewrt{p_{\theta_0}(\cdot\giv \thetah)}{e^{\Delta'_1(X,\thetah)}}\cdot e^{\Delta_1(X,\thetah)}}}} \\
&\leq \prwrt{Q^*_{\theta_0}}{\Ecal_{\textnormal{ball}}^c}+\Ewrt{Q^*_{\theta_0}}{\Ewrt{p_{\theta_0}(\cdot\giv \thetah)}{\Delta_1(X,\thetah)} + 1 -\frac{1}{\Ewrt{p_{\theta_0}(\cdot\giv \thetah)}{e^{\Delta'_1(X,\thetah)}}} },\end{align*}
where the last step holds since $1 - ab \leq (1-a)+(1-b)\leq \log(1/a) + (1-b)$ for any $a,b\in (0,1]$. (Note that, in
the next-to-last line,  the two random variables $X$ appearing in the denominator are different---they are sampled independently conditional on $\thetah$ from the distribution $p_{\theta_0}(\cdot\giv \thetah)$.)

Next, recall that by Lemma~\ref{lem:compute_conditional} together with the definition of $Q^*_{\theta_0}$, the joint distribution of $(X,\thetah)$
in this calculation above (i.e., $\thetah\sim Q^*_{\theta_0}$ and $X\giv\thetah \sim p_{\theta_0}(\cdot\giv \thetah)$), is equivalent to the joint distribution of $(X,\thetah(X,W))$ when $(X,W)\sim P^*_{\theta_0}$.
Therefore, our calculation above can be rewritten as follows (where we also apply Jensen's inequality to the last term):
\[\Ewrt{Q^*_{\theta_0}}{\dtv\big(p_{\theta_0}(\cdot\giv \thetah),p_{\thetah}(\cdot\giv \thetah)\big)}
\leq \prwrt{P^*_{\theta_0}}{\Ecal_{\textnormal{ball}}^c}+ \Ewrt{P^*_{\theta_0}}{\Delta_1(X,\thetah(X,W))} + \left(1 -\frac{1}{\Ewrt{P^*_{\theta_0}}{e^{\Delta'_1(X,\thetah(X,W))}}}\right).\]
Next let
\[\Delta_2(x,w) = r(\theta_0) \sigma\norm{w}  + \frac{r(\theta_0)^2}{2} \sup_{\theta\in\ball(\theta_0,r(\theta_0))\cap\Theta} \big(\lambda_{\max}(H(\theta) - H(\theta;x))\big)_+,\]
and
\[\Delta'_2(x,w) = r(\theta_0) \sigma\norm{w} + \frac{r(\theta_0)^2}{2} \sup_{\theta\in\ball(\theta_0,r(\theta_0))\cap\Theta} \big(\lambda_{\max}(H(\theta;x) - H(\theta))\big)_+,\]
and observe that $\Delta_1(x,\thetah(x,w))=\Delta_2(x,w)$ and $\Delta'_1(x,\thetah(x,w))=\Delta'_2(x,w)$
 for all $(x,w)\in\Omega_{\textnormal{SSOSP}}$, since $0 = \nabla_\theta\Lcal(\thetah(x,w);x,w) = \nabla_\theta\Lcal(\thetah(x,w);x) + \sigma w$
for all $(x,w)$ in this set by definition. Therefore, since $(X,W)\in\Omega_{\textnormal{SSOSP}}$ almost surely under $P^*_{\theta_0}$ by definition, we have
\[\Ewrt{Q^*_{\theta_0}}{\dtv\big(p_{\theta_0}(\cdot\giv \thetah),p_{\thetah}(\cdot\giv \thetah)\big)}
\leq \prwrt{P^*_{\theta_0}}{\Ecal_{\textnormal{ball}}^c}+ \Ewrt{P^*_{\theta_0}}{\Delta_2(X,W)} +  \left(1 -\frac{1}{\Ewrt{P^*_{\theta_0}}{e^{\Delta'_2(X,W)}}}\right).\]
Now let $\Ecal_{\textnormal{SSOSP}}$ be the event that $(X,W)\in\Omega_{\textnormal{SSOSP}}$. Recall that $P^*_{\theta_0}$
is the joint distribution of $(X,W)\sim P_{\theta_0}\times\normal(0,\tfrac{1}{d}\ident_d)$ conditional on $\Ecal_{\textnormal{SSOSP}}$.
Therefore, we can write this as follows where we now take all probabilities and expectations with respect to $(X,W)\sim P_{\theta_0}\times\normal(0,\tfrac{1}{d}\ident_d)$:
\begin{align*}&\hspace{-.25in}\Ewrt{Q^*_{\theta_0}}{\dtv\big(p_{\theta_0}(\cdot\giv \thetah),p_{\thetah}(\cdot\giv \thetah)\big)}\\
&\leq \prst{\Ecal_{\textnormal{ball}}^c}{\Ecal_{\textnormal{SSOSP}}}+ \Est{\Delta_2(X,W)}{\Ecal_{\textnormal{SSOSP}}} +
\left(1-\frac{1}{\Est{e^{\Delta'_2(X,W)}}{\Ecal_{\textnormal{SSOSP}}}}\right)\\
&\le \frac{ \pr{\Ecal_{\textnormal{ball}}^c\cap\Ecal_{\textnormal{SSOSP}}}+\E{\Delta_2(X,W)}}{\pr{\Ecal_{\textnormal{SSOSP}}}} +\left( 1-\frac{\pr{\Ecal_{\textnormal{SSOSP}}}}{\E{e^{\Delta'_2(X,W)}\cdot\one{\Ecal_{\textnormal{SSOSP}}}}}\right)\\
&\leq \frac{ \pr{\Ecal_{\textnormal{ball}}^c\cap\Ecal_{\textnormal{SSOSP}}}+\E{\Delta_2(X,W)}}{\pr{\Ecal_{\textnormal{SSOSP}}}} + \left(1-\frac{1-\pr{\Ecal_{\textnormal{SSOSP}}^c}}{\E{e^{\Delta'_2(X,W)}} - \pr{\Ecal_{\textnormal{SSOSP}}^c}}\right)\\
&\leq \frac{ \pr{\Ecal_{\textnormal{ball}}^c\cap\Ecal_{\textnormal{SSOSP}}}+\E{\Delta_2(X,W)} + \log\E{e^{\Delta'_2(X,W)}}}{\pr{\Ecal_{\textnormal{SSOSP}}}} ,\end{align*}
where the next-to-last step holds since $\Delta'_2(X,W)\geq 0$ by definition, and the last step holds since $1 - \frac{1-a}{b-a} \leq \frac{1-1/b}{1-a} \leq \frac{\log(b)}{1-a}$ for all $a\in[0,1)$ and  $b\ge 1$.
Finally, we apply our assumptions. By Assumption~\ref{asm:thetah},  we have
$\pr{\Ecal_{\textnormal{ball}}\cap\Ecal_{\textnormal{SSOSP}}}\geq 1 - \delta(\theta_0)$,
and so
\[\pr{\Ecal_{\textnormal{ball}}^c\cap\Ecal_{\textnormal{SSOSP}}} \leq \delta(\theta_0) - \pr{\Ecal_{\textnormal{SSOSP}}^c}.\]
Next,
\begin{multline*}
\E{\Delta_2(X,W)}
= \E{r(\theta_0) \sigma\norm{W}} + \E{\frac{r(\theta_0)^2}{2} \sup_{\theta\in\ball(\theta_0,r(\theta_0))\cap\Theta} \big(\lambda_{\max}(H(\theta) - H(\theta;X))\big)_+}\\
\leq  \frac{1}{2} \log\E{e^{2 r(\theta_0) \sigma \norm{W} }} + \frac{\eps(\theta_0)}{2},
\end{multline*}
where the last step holds by Jensen's inequality for the first term and by the bound~\eqref{eqn:asm:hessian_conc1} in Assumption~\ref{asm:hessian} for the second term.
And, by Cauchy--Schwarz,
\begin{multline*}
\log\E{e^{\Delta'_2(X,W)}}\leq \frac{1}{2} \log\E{e^{2 r(\theta_0) \sigma \norm{W} }} + \frac{1}{2}\log\E{e^{r(\theta_0)^2 \sup_{\theta\in\ball(\theta_0,r(\theta_0))\cap\Theta} \big(\lambda_{\max}(H(\theta;X) - H(\theta))\big)_+}}\\
\leq \frac{1}{2} \log\E{e^{2 r(\theta_0) \sigma \norm{W} }} + \frac{\eps(\theta_0)}{2},\end{multline*}
where the last step holds by the bound~\eqref{eqn:asm:hessian_conc2} in Assumption~\ref{asm:hessian}.
Finally, since $W\sim \normal(0,\tfrac{1}{d}\ident_d)$ we know that $\E{e^{t\norm{W}}}\leq e^{t+t^2/2d}$ for any $t>0$ (see, e.g., \cite[Theorem 5.5]{boucheron2013concentration}).
Therefore,
\[\log\E{e^{2r(\theta_0)\sigma \norm{W}}}
\leq 2\sigma\cdot r(\theta_0) + \frac{2\sigma^2\cdot r(\theta_0)^2}{d} \leq 3\sigma\cdot r(\theta_0),\]
where the last step holds since $d\geq 1$ and we can assume $2\sigma\cdot r(\theta_0) \leq 1$ (as otherwise, the result of the theorem holds trivially).
Combining everything, we have
\[\Ewrt{Q^*_{\theta_0}}{\dtv\big(p_{\theta_0}(\cdot\giv \thetah),p_{\thetah}(\cdot\giv \thetah)\big)}
\leq \frac{3\sigma\cdot r(\theta_0)+ \delta(\theta_0)+   \eps(\theta_0)  - \pr{\Ecal_{\textnormal{SSOSP}}^c}}{1 - \pr{\Ecal_{\textnormal{SSOSP}}^c}}.\]
Since total variation distance is bounded by 1, trivially we can relax this to
\[\Ewrt{Q^*_{\theta_0}}{\dtv\big(p_{\theta_0}(\cdot\giv \thetah),p_{\thetah}(\cdot\giv \thetah)\big)}
\leq 3\sigma\cdot r(\theta_0) + \delta(\theta_0) +   \eps(\theta_0) .\]
Returning to~\eqref{eqn:relate_to_tv}, we see that the aCSS procedure~\eqref{eqn:aCSS_alg} satisfies
\[\dex(X,\Xt^{(1)},\dots,\Xt^{(M)}) \leq 3\sigma\cdot r(\theta_0) + \delta(\theta_0) +   \eps(\theta_0) ,\]
as desired.

\subsection{Proof of Lemma~\ref{lem:compute_conditional}}\label{app:proof_lem:compute_conditional}
Consider the joint distribution $(X,W)\sim P_{\theta_0}\times\normal(0,\tfrac{1}{d}\ident_d)$ conditioned on the event that
$(X,W)\in\Omega_{\textnormal{SSOSP}}$, which is assumed to occur with positive probability. The joint density of $(X,W)$, after conditioning
on this event, is therefore proportional to the function
\begin{equation}\label{eqn:joint_density_XW}g_{\theta_0}(x,w) = f(x;\theta_0)\cdot\exp\left\{-\tfrac{d}{2}\norm{w}^2\right\}\cdot\one{(x,w)\in\Omega_{\textnormal{SSOSP}}},\end{equation}
with respect to the measure $\nu_\Xcal\times\leb$.
We will consider the induced joint distribution of $(X,\thetah(X,W))$, and will calculate its joint density.

Define $\psi$ and $\psi^{-1}$ as in Lemma~\ref{lem:bijection}. Fix any measurable subset $A\subseteq\Psi_{\textnormal{SSOSP}}$.
Then, writing $\psi^{-1}(A) = \{(x,w)\in \Omega_{\textnormal{SSOSP}}:\psi(x,w)\in A\}\subseteq \Omega_{\textnormal{SSOSP}}$,
\[\pr{(X,\thetah(X,W))\in A} = \pr{(X,W)\in \psi^{-1}(A)} \\=  \frac{\int_{\psi^{-1}(A)} g_{\theta_0}(x,w) \;\mathsf{d}\nu_\Xcal(x)\mathsf{d}w}{\int_{\Xcal\times\R^d}g_{\theta_0}(x',w') \;\mathsf{d}\nu_\Xcal(x')\mathsf{d}w'},\]
where the probability is taken with respect to
$(X,W)\sim P_{\theta_0}\times\normal(0,\tfrac{1}{d}\ident_d)$ conditioned on the event that
$(X,W)\in\Omega_{\textnormal{SSOSP}}$. (Note that, since $g_{\theta_0}$ is proportional to a density on $(X,W)$ with respect to $\nu_\Xcal\times\leb$,
this implies that the denominator $\int_{\Xcal\times\R^d}g_{\theta_0}(x',w') \;\mathsf{d}\nu_\Xcal(x')\mathsf{d}w'$ in the last expression above must be finite and positive.)

From this point on, the result essentially follows from a change-of-variables calculation, under the transformation $\theta = \thetah(x,w)$.
However, with our weak assumptions, we cannot assume standard conditions (such as, e.g., the support of $\thetah\giv X$ being an open subset of $\R^d$---it may even
be the case that this set does not contain any open subset), so we will need to be careful.
Fixing any $x\in\Xcal$, a change-of-variables calculation, proved formally in Appendix~\ref{app:lem:change_of_vars}  below,
establishes that
\begin{align}
\notag&\hspace{-.25in}\int_{\Theta} \exp\left\{-\tfrac{d}{2\sigma^2}\norm{\nabla_\theta\Lcal(\theta;x)}^2\right\}\cdot\det(\nabla_\theta^2\Lcal(\theta;x))\cdot\one{(x,\theta)\in A\cap\Psi_{\textnormal{SSOSP}}}\;\mathsf{d}\theta \\
\label{eqn:lem:change_of_vars}&=\sigma^d  \int_{\R^d}\exp\left\{-\tfrac{d}{2\sigma^2}\norm{\nabla_\theta\Lcal(\thetah(x,w);x)}^2\right\}\cdot\one{(x,\thetah(x,w))\in A}\cdot\one{(x,w)\in\Omega_{\textnormal{SSOSP}}}\;\mathsf{d}w\\
\notag&=\sigma^d  \int_{\R^d}\exp\left\{-\tfrac{d}{2}\norm{w}^2\right\}\cdot\one{(x,\thetah(x,w))\in A}\cdot\one{(x,w)\in\Omega_{\textnormal{SSOSP}}}\;\mathsf{d}w\\
\label{eqn:claim_for_joint_density}&=\sigma^d  \int_{\R^d}\exp\left\{-\tfrac{d}{2}\norm{w}^2\right\}\cdot\one{(x,w)\in \psi^{-1}(A)\cap \Omega_{\textnormal{SSOSP}}}\;\mathsf{d}w,
\end{align}
where the second step uses the fact that $w = -\frac{\nabla_\theta\Lcal(\thetah(x,w);x)}{\sigma}$ for any $(x,w)\in\Omega_{\textnormal{SSOSP}}$ by the SSOSP conditions, and the
 last step applies the definition of $\psi$ as in Lemma~\ref{lem:bijection}.
Now define the function
\[h_{\theta_0}(x,\theta) := \frac{ f(x;\theta_0)\exp\left\{-\tfrac{d}{2\sigma^2}\norm{\nabla_\theta\Lcal(\theta;x)}^2\right\}\cdot\det(\nabla_\theta^2\Lcal(\theta;x))\cdot\one{x\in\Xcal_\theta}}{\sigma^d\int_{\Xcal\times\R^d}g_{\theta_0}(x',w') \;\mathsf{d}\nu_\Xcal(x')\mathsf{d}w'} \]
on $(x,\theta)\in\Xcal\times\Theta$.
We then have
\begin{align*}
&\pr{(X,\thetah(X,W))\in A} 
=\frac{\int_{\psi^{-1}(A)} g_{\theta_0}(x,w) \;\mathsf{d}\nu_\Xcal(x)\mathsf{d}w}{\int_{\Xcal\times\R^d}g_{\theta_0}(x',w') \;\mathsf{d}\nu_\Xcal(x')\mathsf{d}w'}\\
&=\frac{\int_\Xcal \int_{\R^d}   f(x;\theta_0)\cdot\exp\left\{-\tfrac{d}{2}\norm{w}^2\right\} \cdot  \one{(x,w)\in \psi^{-1}(A)\cap \Omega_{\textnormal{SSOSP}}}\;\mathsf{d}w\;\mathsf{d}\nu_\Xcal(x)}{\int_{\Xcal\times\R^d}g_{\theta_0}(x',w') \;\mathsf{d}\nu_\Xcal(x')\mathsf{d}w'}\textnormal{\quad by~\eqref{eqn:joint_density_XW}}\\
&=\frac{\int_\Xcal   f(x;\theta_0)\int_{\Theta} \exp\left\{-\tfrac{d}{2\sigma^2}\norm{\nabla_\theta\Lcal(\theta;x)}^2\right\}\cdot\det(\nabla_\theta^2\Lcal(\theta;x))\cdot\one{(x,\theta)\in A\cap\Psi_{\textnormal{SSOSP}}}\;\mathsf{d}\theta\;\mathsf{d}\nu_\Xcal(x)}{\sigma^d\int_{\Xcal\times\R^d}g_{\theta_0}(x',w') \;\mathsf{d}\nu_\Xcal(x')\mathsf{d}w'}\textnormal{\quad by~\eqref{eqn:claim_for_joint_density}}\\
&=\int_A h_{\theta_0}(x,\theta)\;\mathsf{d}\nu_\Xcal(x)\;\mathsf{d}\theta,
\end{align*}
where the last step holds since $\one{x\in\Xcal_\theta}=\one{(x,\theta)\in \Psi_{\textnormal{SSOSP}}}$ for all $(x,\theta)$, by definition of $\Psi_{\textnormal{SSOSP}}$.
Therefore, this calculation establishes that, conditional on the event that $\thetah(X,W)$ is a SSOSP of $\Lcal(\theta;X,W)$, the joint distribution
of $(X,\thetah(X,W))$ has density $h_{\theta_0}(x,\theta)$ with respect to the base measure $\nu_\Xcal\times\leb$.

Finally, since $h_{\theta_0}(x,\theta)$ is the joint density of $(X,\thetah) = (X,\thetah(X,W))$, we therefore see that $X\giv\thetah$
has conditional density equal to
\[\frac{h_{\theta_0}(x,\thetah)}{\int_{x'} h_{\theta_0}(x',\thetah)\;\mathsf{d}\nu_\Xcal(x')} \propto  f(x;\theta_0)\exp\left\{-\tfrac{d}{2\sigma^2}\norm{\nabla_\theta\Lcal(\thetah;x)}^2\right\}\cdot\det(\nabla_\theta^2\Lcal(\thetah;x))\cdot\one{x\in\Xcal_{\thetah}},\]
which verifies the desired expression~\eqref{eqn:density_realX}.

\section{Additional proofs}\label{app:add_proofs}

\subsection{Proof of Lemma~\ref{lem:bijection}}\label{app:proof_lem:bijection}
First we check that $\psi$ is injective on $\Omega_{\textnormal{SSOSP}}$, which holds since for any $(x,\theta)$, if $\psi(x',w) = (x,\theta)$ then we must have $x=x'$ trivially
 and we must have $w=-\frac{\nabla_\theta\Lcal(\theta;x)}{\sigma}$ by definition of the SSOSP conditions.
This establishes that $\psi$ is injective and that the inverse function (on the image of $\psi$) is given by $\psi^{-1}(x,\theta) = 
 \left(x,-\frac{\nabla_\theta\Lcal(\theta;x)}{\sigma}\right)$ as claimed above.
 
Now we verify that $\Psi_{\textnormal{SSOSP}}$ is the image of $\psi$. Fix any $(x,\theta)\in\Xcal\times\Theta$.
First, suppose $(x,\theta)\in  \psi(\Omega_{\textnormal{SSOSP}})$, i.e., we have $\theta = \thetah(x,w)$
for some $w$ such that $(x,w)\in\Omega_{\textnormal{SSOSP}}$. Then 
by definition of $\Omega_{\textnormal{SSOSP}}$, $\theta$ is a SSOSP of $\Lcal(\theta;x,w)$, and so $x\in\Xcal_\theta$ 
and therefore $(x,\theta)\in \Psi_{\textnormal{SSOSP}}$. Conversely suppose that $(x,\theta)\in \Psi_{\textnormal{SSOSP}}$.
Then by definition, $x\in\Xcal_\theta$ and so there exists some $w$ such that $\theta = \thetah(x,w)$ and $\theta$ is a SSOSP of $\Lcal(\theta;x,w)$.
Therefore, for this choice of $w$, we have $(x,w)\in\Omega_{\textnormal{SSOSP}}$ and so $(x,\theta) = \psi(x,w) \in \psi(\Omega_{\textnormal{SSOSP}})$.

\subsection{Distance to exchangeability for mixture distributions}\label{app:dexch_mixture_lem}
In this section, we verify the claim~\eqref{eqn:thm:main_dex_for_mixture} that appears in the proof of Theorem~\ref{thm:main}.
Specifically, we need to show that the distance-to-exchangeability $\dex$ introduced in Definition~\ref{def:dex}
is convex on the space of distributions.
\begin{lemma}\label{lem:dex_for_mixture}
Consider any distributions $P_0,P_1$ on $(A_1,\dots,A_k)$, and any $c\in[0,1]$. Let $P = (1-c)\cdot P_0 + c\cdot P_1$ be the mixture distribution. 
Then 
\[\dex(P) \leq (1-c)\cdot \dex(P_0) + c\cdot \dex(P_1).\]
\end{lemma}
With this lemma in place, we have
\begin{multline*}\dex(X,\Xt^{(1)},\dots,\Xt^{(M)})  \leq
\pr{(X,W)\in\Omega_{\textnormal{SSOSP}}} \cdot  \dex\left(\textnormal{\begin{tabular}{c}Distrib.~of $X,\Xt^{(1)},\dots,\Xt^{(M)}$ \\ 
condl.~on $(X,W)\in \Omega_{\textnormal{SSOSP}}$\end{tabular}}\right)\\
+\pr{(X,W)\not\in\Omega_{\textnormal{SSOSP}}} \cdot  \dex\left(\textnormal{\begin{tabular}{c}Distrib.~of $X,\Xt^{(1)},\dots,\Xt^{(M)}$ \\ 
condl.~on $(X,W)\not\in \Omega_{\textnormal{SSOSP}}$\end{tabular}}\right).\end{multline*}
Furthermore, we know that
\[\dex\left(\textnormal{\begin{tabular}{c}Distrib.~of $X,\Xt^{(1)},\dots,\Xt^{(M)}$ \\ 
condl.~on $(X,W)\not\in \Omega_{\textnormal{SSOSP}}$\end{tabular}}\right) = 0\]
since, on the event that $(X,W)\not\in \Omega_{\textnormal{SSOSP}}$, we set $\Xt^{(1)} = \dots = \Xt^{(M)} = X$ by definition of the method.
Therefore, 
the claim~\eqref{eqn:thm:main_dex_for_mixture} must hold.

\begin{proof}[Proof of Lemma~\ref{lem:dex_for_mixture}]
Fix any $\eps>0$. By definition of $\dex$, for each $\ell=0,1$ we can find some exchangeable distribution $Q_\ell$ on $(B_1,\dots,B_k)$ such that 
\[\dtv(P_\ell,Q_\ell) \leq \dex(P_\ell) + \eps.\]
Next define the mixture distribution
$Q = (1-c)\cdot Q_0 + c\cdot Q_1$.
Clearly $Q$ is exchangeable, inheriting this property from $Q_0$ and $Q_1$, and therefore
$\dex(P) \leq \dtv(P,Q)$. Furthermore, for any measurable subset $A$, we have
\begin{multline*}\big|P(A) - Q(A)\big| = \left|\big((1-c)\cdot P_0(A) + c\cdot P_1(A)\big) - \big((1-c)\cdot Q_0(A) + c\cdot Q_1(A)\big)\right|\\ \leq (1-c) \cdot |P_0(A)-Q_0(A)| + c\cdot |P_1(A) -Q_1(A)|\leq (1-c)\cdot \dtv(P_0,Q_0) + c\cdot \dtv(P_1,Q_1).\end{multline*}
This establishes that
$\dtv(P,Q) \leq (1-c)\cdot \dtv(P_0,Q_0) + c\cdot \dtv(P_1,Q_1)$, and therefore,
\[\dex(P) \leq (1-c) \cdot (\dex(P_0) +\eps) + c\cdot(\dex(P_1) + \eps).\] Since $\eps>0$ can be taken to be arbitrarily small, this proves the lemma.
\end{proof}

\subsection{Verifying that~\eqref{eqn:density_hat} defines a density}\label{app:check_density}
To ensure that our procedure is well defined, we need to check that
\[ p_{\thetah}(x\giv \thetah) \propto p^{\textnormal{un}}_{\thetah}(x)\]
defines a valid density with respect to $\nu_\Xcal$,
where the unnormalized function is given by
\[p^{\textnormal{un}}_\theta(x) := f(x;\theta)\cdot \exp\left\{ - \frac{\norm{\nabla_\theta\Lcal(\theta;x)}^2}{2\sigma^2/d}\right\}\cdot \det\left(\nabla^2_\theta \Lcal(\theta;x)\right)\cdot\one{x\in\Xcal_{\theta}}.\] The following lemma verifies all the necessary conditions:
\begin{lemma}\label{lem:density_hat} If Assumptions~\ref{asm:family} and~\ref{asm:hessian} hold, then
for all $\theta\in\Theta$ the function $x\mapsto p^{\textnormal{un}}_\theta(x)$ is nonnegative and integrable with respect to $\nu_\Xcal$.
Furthermore, if the event that $\thetah=\thetah(X,W)$ is a SSOSP of $\Lcal(\theta;X,W)$ has positive probability, then 
conditional on this event,
 \[\int_\Xcal p^{\textnormal{un}}_{\thetah}(x)\;\mathsf{d}\nu_\Xcal(x)>0.\]
holds almost surely.
\end{lemma}
\begin{proof}
First we check nonnegativity. For any $\theta$ and any $x$, $ f(x;\theta)> 0$ by Assumption~\ref{asm:family}. Furthermore, if $x\in\Xcal_\theta$ then $\nabla^2_\theta\Lcal(\theta;x)\succ 0$ and so $\det(\nabla^2_\theta\Lcal(\theta;x))>0$ by definition of the SSOSP conditions. This verifies that $p^{\textnormal{un}}_\theta(x)\geq 0$ for all $(x,\theta)$.
Next we check integrability. We have
\begin{align*}
&\int_\Xcal p^{\textnormal{un}}_{\thetah}(x)\;\mathsf{d}\nu_\Xcal(x)
\leq \int_\Xcal  f(x;\theta)\cdot \det\left(\nabla^2_\theta \Lcal(\theta;x)\right)\cdot \one{\nabla^2_\theta \Lcal(\theta;x) \succ 0}\;\mathsf{d}\nu_\Xcal(x)\\
&\leq \int_\Xcal  f(x;\theta)\cdot \big(\lambda_{\max}(\nabla^2_\theta \Lcal(\theta;x))\big)_+^d\;\mathsf{d}\nu_\Xcal(x)\\
&\leq \frac{d!}{r(\theta)^{2d}}\int_\Xcal  f(x;\theta)\cdot  \exp\left\{r(\theta)^2\big(\lambda_{\max}(\nabla^2_\theta \Lcal(\theta;x))\big)_+\right\}\;\mathsf{d}\nu_\Xcal(x)\\
&\leq \frac{d!}{r(\theta)^{2d}} \int_\Xcal  f(x;\theta)\cdot \exp\left\{ r(\theta)^2\big(\lambda_{\max}(H(\theta;x) - H(\theta)\big)_+ + r(\theta)^2\big(\lambda_{\max}(H(\theta) + \nabla^2_\theta \pen(\theta)\big)_+\right\}\;\mathsf{d}\nu_\Xcal(x)\\
&\leq \frac{d!}{r(\theta)^{2d}} \cdot e^{\eps(\theta)}\cdot \exp\left\{r(\theta)^2\big(\lambda_{\max}(H(\theta) + \nabla^2_\theta \pen(\theta)\big)_+\right\},
\end{align*}
where the last step holds by Assumption~\ref{asm:hessian}. This proves that $\int_\Xcal p^{\textnormal{un}}_{\thetah}(x)\;\mathsf{d}\nu_\Xcal(x)$ is finite.

Finally we check that $\int_\Xcal p^{\textnormal{un}}_{\thetah}(x)\;\mathsf{d}\nu_\Xcal(x)>0$ almost surely.
Since $f(x;\theta)>0$ for all $x,\theta$ by Assumption~\ref{asm:family}, it is equivalent to verify that
$\int_\Xcal \frac{f(x;\theta_0)}{f(x;\thetah)}p^{\textnormal{un}}_{\thetah}(x)\;\mathsf{d}\nu_\Xcal(x)>0$ almost surely.
Recalling from~\eqref{eqn:density_realX} that $p_{\theta_0}(x|\thetah)\propto  \frac{f(x;\theta_0)}{f(x;\thetah)}\cdot p^{\textnormal{un}}_{\thetah}(x)$ is the conditional
density of $X\giv \thetah$, this must be true.
\end{proof}

\subsection{Change of variables calculation}\label{app:lem:change_of_vars}
In this section, we verify the change-of-variables calculation needed in the proof of Lemma~\ref{lem:compute_conditional}.
Specifically, the step~\eqref{eqn:lem:change_of_vars} follows by applying the lemma below to the function
\[\rho(x,\theta) =  \exp\left\{-\tfrac{d}{2\sigma^2}\norm{\nabla_\theta\Lcal(\thetah(x,w);x)}^2\right\}\cdot\one{(x,\theta)\in A}.\]

\begin{lemma}\label{lem:change_of_vars}
Suppose Assumption~\ref{asm:family} holds.
For all nonnegative measurable functions $\rho:\Xcal\times\Theta \rightarrow\R$, it holds  for all $x\in\Xcal$ that
\[\int_{\Theta} \rho(x,\theta)\cdot\det(\nabla_\theta^2\Lcal(\theta;x))\cdot\one{(x,\theta)\in \Psi_{\textnormal{SSOSP}}}\;\mathsf{d}\theta \\=\sigma^d  \int_{\R^d}\rho(x,\thetah(x,w))\cdot\one{(x,w)\in\Omega_{\textnormal{SSOSP}}}\;\mathsf{d}w.\]
\end{lemma}
\begin{proof}[Proof of Lemma~\ref{lem:change_of_vars}]
Define
\[A_x = \{\theta\in\Theta : \nabla^2\Lcal(\theta;x)\succ 0\},\]
which is an open set since $\Lcal(\theta;x)$ is continuously twice differentiable in $\theta$.
If this set is empty then the lemma is trivial (since the left- and right-hand side are both equal to zero),
so from this point on we will assume $A_x$ is nonempty. Let
\[B_x = \{w\in\R^d : \thetah(x,w)\in A_x\}.\]
 By definition, if $(x,\theta)\in \Psi_{\textnormal{SSOSP}}$ then we must have $\theta\in A_x$, and similarly if 
$(x,w)\in\Omega_{\textnormal{SSOSP}}$ then we must have $\thetah(x,w)\in A_x$ and so $w\in B_x$. Therefore, 
to prove the lemma, it is sufficient to show that
\begin{equation}\label{eqn:claim_for_joint_density_0}
 \int_{A_x} \rho(x,\theta)\cdot\det(\nabla_\theta^2\Lcal(\theta;x))\cdot\one{(x,\theta)\in \Psi_{\textnormal{SSOSP}}}\;\mathsf{d}\theta \\=\sigma^d  \int_{B_x}\rho(x,\thetah(x,w))\cdot\one{(x,w)\in\Omega_{\textnormal{SSOSP}}}\;\mathsf{d}w.\end{equation}

Next define nested sets 
\[A_{x,\lambda} = \{\theta\in A_x:\textnormal{$\ball(\theta,\lambda)\subseteq\Theta$ and $\nabla^2\Lcal(\theta';x)\succ 0$ for all $\theta'\in\ball(\theta,\lambda)$}\}\]
indexed by $\lambda>0$.
Since $\Theta$ is an open subset of $\R^d$, and $\Lcal(\theta;x)$ is continuously twice differentiable in $\theta$, 
we see that $A_x=\cup_{\lambda>0}A_{x,\lambda}$. 
Similarly we have $B_x = \cup_{\lambda>0} B_{x,\lambda}$ where
\[B_{x,\lambda} = \{w\in\R^d : \thetah(x,w)\in A_{x,\lambda}\}.\]
By the monotone
convergence theorem, this implies that
\[ \int_{A_x} \rho(x,\theta)\cdot\det(\nabla_\theta^2\Lcal(\theta;x))\cdot\one{(x,\theta)\in \Psi_{\textnormal{SSOSP}}}\;\mathsf{d}\theta
= \lim_{\lambda\rightarrow 0} \int_{A_{x,\lambda}} \rho(x,\theta)\cdot\det(\nabla_\theta^2\Lcal(\theta;x))\cdot\one{(x,\theta)\in \Psi_{\textnormal{SSOSP}}}\;\mathsf{d}\theta,\]
and similarly,
\[\int_{B_x}\rho(x,\thetah(x,w))\cdot\one{(x,w)\in\Omega_{\textnormal{SSOSP}}}\;\mathsf{d}w 
=\lim_{\lambda\rightarrow 0} \int_{B_{x,\lambda}}\rho(x,\thetah(x,w))\cdot\one{(x,w)\in\Omega_{\textnormal{SSOSP}}}\;\mathsf{d}w.\]
Therefore, to prove~\eqref{eqn:claim_for_joint_density_0}, it is sufficient to show that, for each $\lambda>0$,
\begin{equation}\label{eqn:claim_for_joint_density_1}
 \int_{A_{x,\lambda}} \rho(x,\theta)\cdot\det(\nabla_\theta^2\Lcal(\theta;x))\cdot\one{(x,\theta)\in \Psi_{\textnormal{SSOSP}}}\;\mathsf{d}\theta\\ = 
\sigma^d \int_{B_{x,\lambda}}\rho(x,\thetah(x,w))\cdot\one{(x,w)\in\Omega_{\textnormal{SSOSP}}}\;\mathsf{d}w.\end{equation}

From this point on we will treat $\lambda>0$ as fixed.
Let $S_1,S_2,\dots$ be a countable collection of disjoint open sets, each of diameter $\leq\lambda$, such that $\leb\big(\R^d\backslash(\cup_{k\geq 1}S_k)\big) =0$
(for example, we can partition $\R^d$ into countably many sufficiently small hypercubes). Then
\[ \int_{A_{x,\lambda}} \rho(x,\theta)\cdot\det(\nabla_\theta^2\Lcal(\theta;x))\cdot\one{(x,\theta)\in \Psi_{\textnormal{SSOSP}}}\;\mathsf{d}\theta 
= \sum_{k\geq 1} \int_{A_{x,\lambda,k}} \rho(x,\theta)\cdot\det(\nabla_\theta^2\Lcal(\theta;x))\cdot\one{(x,\theta)\in \Psi_{\textnormal{SSOSP}}}\;\mathsf{d}\theta,\]
where $A_{x,\lambda,k} = A_{x,\lambda}\cap S_k$, and similarly
\[ \int_{B_{x,\lambda}}\rho(x,\thetah(x,w))\cdot\one{(x,w)\in\Omega_{\textnormal{SSOSP}}}\;\mathsf{d}w
=\sum_{k\geq 1}  \int_{B_{x,\lambda,k}}\rho(x,\thetah(x,w))\cdot\one{(x,w)\in\Omega_{\textnormal{SSOSP}}}\;\mathsf{d}w,\]
where 
\[B_{x,\lambda,k} = \{w\in\R^d : \thetah(x,w)\in A_{x,\lambda,k}\}.\]
Therefore,  to prove~\eqref{eqn:claim_for_joint_density_1}, it is sufficient to show that, for each $\lambda>0$ and each $k\geq 1$,
\begin{multline}\label{eqn:claim_for_joint_density_2}
 \int_{A_{x,\lambda,k}} \rho(x,\theta)\cdot\det(\nabla_\theta^2\Lcal(\theta;x))\cdot\one{(x,\theta)\in \Psi_{\textnormal{SSOSP}}}\;\mathsf{d}\theta\\ = 
\sigma^d \int_{B_{x,\lambda,k}}\rho(x,\thetah(x,w))\cdot\one{(x,w)\in\Omega_{\textnormal{SSOSP}}}\;\mathsf{d}w.\end{multline}

From this point on we will treat both $\lambda>0$ and $k\geq 1$ as fixed, and will prove~\eqref{eqn:claim_for_joint_density_2}.
First, by definition of the SSOSP conditions, if $(x,\theta)\in \Psi_{\textnormal{SSOSP}}$ then we must have 
\[\theta = \thetah(x,\phi_x(\theta))\textnormal{\quad where \quad}\phi_x(\theta):= - \frac{\nabla_\theta \Lcal(\theta;x)}{\sigma},\]
and furthermore, $\nabla_\theta^2\Lcal(\theta;x)\succ 0$ and so $\det(\nabla_\theta^2\Lcal(\theta;x))>0$.
Therefore, we can calculate that
\[ \rho(x,\theta)\cdot\det(\nabla_\theta^2\Lcal(\theta;x)) =  \rho(x,\thetah(x,\phi_x(\theta)))\cdot \sigma^d |\det(\nabla_\theta\phi_x(\theta))|\]
for all $(x,\theta)\in\Psi_{\textnormal{SSOSP}}$, and so~\eqref{eqn:claim_for_joint_density_2} is equivalent to the claim that
\begin{multline}\label{eqn:claim_for_joint_density_3}
 \int_{A_{x,\lambda,k}} \rho(x,\thetah(x,\phi_x(\theta)))\cdot |\det(\nabla_\theta\phi_x(\theta))|\cdot\one{(x,\thetah(x,\phi_x(\theta)))\in \Psi_{\textnormal{SSOSP}}}\;\mathsf{d}\theta\\ = 
 \int_{B_{x,\lambda,k}}\rho(x,\thetah(x,w))\cdot\one{(x,w)\in \Omega_{\textnormal{SSOSP}}}\;\mathsf{d}w.\end{multline}

Next, we show that $\phi_x:A_{x,\lambda,k}\rightarrow \phi_x(A_{x,\lambda,k})$ is a diffeomorphism.
$\phi_x$ is clearly differentiable, and its derivative is invertible since $\nabla_\theta \phi_x(\theta) = (-\sigma)^{-d} \nabla^2_\theta\Lcal(\theta;x)$, and $ \nabla^2_\theta\Lcal(\theta;x)\succ 0$
on  $A_{x,\lambda,k}$ by definition. To check injectivity, if $\phi_x(\theta)=\phi_x(\theta')$ for some $\theta,\theta'\in A_{x,\lambda,k}$,
then by Taylor's theorem we must have $\nabla^2_\theta \Lcal((1-t)\theta + t\theta';x) \cdot(\theta'-\theta) = 0$ for some $t\in[0,1]$.
Since the diameter of $S_k$ (and therefore, of $A_{x,\lambda,k}$) is $\leq\lambda$,
we must have $\norm{\theta-\theta'}\leq \lambda$ and therefore $(1-t)\theta + t\theta'\in\ball(\theta,\lambda)$.
By definition of $A_{x,\lambda}$, this implies that $\nabla^2_\theta \Lcal((1-t)\theta + t\theta';x) \succ 0$, and we conclude that $\theta'-\theta =0$,
thus establishing injectivity. Therefore,  $\phi_x:A_{x,\lambda,k}\rightarrow \phi_x(A_{x,\lambda,k})$ is a diffeomorphism.
Since  $A_{x,\lambda,k}\subseteq \R^d$ is an open set, by the change-of-variables formula we therefore have
\begin{multline*}
 \int_{A_{x,\lambda,k}} \rho(x,\thetah(x,\phi_x(\theta)))\cdot |\det(\nabla_\theta\phi_x(\theta))|\cdot\one{(x,\thetah(x,\phi_x(\theta)))\in \Psi_{\textnormal{SSOSP}}}\;\mathsf{d}\theta\\ = 
 \int_{\phi_x(A_{x,\lambda,k})}\rho(x,\thetah(x,w))\cdot\one{(x,\thetah(x,w))\in \Psi_{\textnormal{SSOSP}}}\;\mathsf{d}w,\end{multline*}
 Therefore, to prove~\eqref{eqn:claim_for_joint_density_3}, we now only need to check that
\[\One{w\in\phi_x(A_{x,\lambda,k}), \ (x,\thetah(x,w))\in \Psi_{\textnormal{SSOSP}}} = \One{w\in B_{x,\lambda,k},\ (x,w)\in\Omega_{\textnormal{SSOSP}}}\]
for all $(x,w)$.
First suppose $w\in\phi_x(A_{x,\lambda,k})$ and $(x,\thetah(x,w))\in \Psi_{\textnormal{SSOSP}}$.
Then we have $w = \phi_x(\theta)$ for some $\theta\in A_{x,\lambda,k}$. By definition, this means $(x,\theta)\in\Psi_{\textnormal{SSOSP}}$,
and so we must have some $w'$ such that $\theta = \thetah(x,w')$ and $\theta$ is a SSOSP of $\Lcal(\theta;x,w')$.
By the SSOSP conditions, this implies that $0=\nabla_\theta\Lcal(\theta;x,w')$ and so $w'=\phi_x(\theta)$, and therefore $w=w'$. Therefore, $(x,w)\in\Omega_{\textnormal{SSOSP}}$, 
and $\thetah(x,w) \in A_{x,\lambda,k}$ which implies $w\in B_{x,\lambda,k}$.
Conversely, suppose that $w\in B_{x,\lambda,k}$ and $(x,w)\in\Omega_{\textnormal{SSOSP}}$.
Then by definition of $B_{x,\lambda,k}$, we have $\thetah(x,w)\in A_{x,\lambda,k}$. 
Furthermore, by the SSOSP conditions we must have $0=\nabla_\theta\Lcal(\thetah(x,w);x,w)$ and so $w=\phi_x(\thetah(x,w))$,
and therefore, $w\in \phi_x(A_{x,\lambda,k})$ and  $(x,\thetah(x,w))\in \Psi_{\textnormal{SSOSP}}$.
This completes the proof of~\eqref{eqn:claim_for_joint_density_3}, and therefore proves the lemma. 
\end{proof}

\section{Proofs for examples}\label{app:proofs_for_examples}

We now turn to establishing that  our examples all satisfy the assumptions needed for aCSS to control Type I error.
The regularity conditions (Assumption~\ref{asm:family}) hold by definition for all of our examples,
so we only need to verify 
the properties of the estimator $\thetah$ (Assumption~\ref{asm:thetah}) and the Hessian conditions (Assumption~\ref{asm:hessian}).

\subsection{Checking Assumption~\ref{asm:hessian}}
The Hessian conditions~\eqref{eqn:asm:hessian_conc1} and~\eqref{eqn:asm:hessian_conc2} are immediately implied 
by the stronger condition
\begin{equation}\label{eqn:asm:hessian_conc_strong}\Ewrt{\theta_0}{\exp\left\{\sup_{\theta\in\ball(\theta_0,r(\theta_0))\cap\Theta}r(\theta_0)^2\cdot\norm{H(\theta;X)-H(\theta)}\right\}} \leq e^{\eps(\theta_0)}.\end{equation}
We will check that this stronger condition holds for each of our examples.
Specifically, fixing $\theta_0\in\Theta$ we will prove that, for any $c>0$ we can find $c'>0$ such that
\begin{equation}\label{eqn:asm:hessian_conc_strong_for_examples}\Ewrt{\theta_0}{\exp\left\{\sup_{\theta\in\ball(\theta_0,c\sqrt{\frac{\log n}{n}})\cap\Theta}\tfrac{c^2\log n}{n}\cdot\norm{H(\theta;X)-H(\theta)}\right\}} \leq e^{c' \eps_n}\end{equation}
for all sufficiently large $n$, where $\eps_n$ is some vanishing term (specified below for each example)  that does not depend on our choice of $c$. Since later on we will verify that Assumption~\ref{asm:thetah} holds with $r(\theta_0)\asymp \sqrt{\frac{\log n}{n}}$,
this will be sufficient to verify that~\eqref{eqn:asm:hessian_conc_strong} holds.

\subsubsection{Checking Assumption~\ref{asm:hessian} for Example~\ref{ex:canonical_glm}}
For the canonical GLM setting (Example~\ref{ex:canonical_glm}), we can calculate
\[H(\theta;x) =  \sum_{i=1}^n Z_iZ_i^\top \cdot a''(Z_i^\top\theta),\]
which does not depend on $x$. Therefore, $H(\theta) = H(\theta;x)$ for all $x$, or in other words, $\norm{H(\theta;X) - H(\theta)}=0$ almost surely.
Therefore~\eqref{eqn:asm:hessian_conc_strong_for_examples} holds trivially with $\eps_n = 0$.

\subsubsection{Checking Assumption~\ref{asm:hessian} for Example~\ref{ex:behrens_fisher}}
For the Behrens--Fisher problem (Example~\ref{ex:behrens_fisher}), we can calculate
\[H(\theta;x) = \left(\begin{array}{ccc}\frac{n^{(0)}}{\gamma^{(0)}}+\frac{n^{(1)}}{\gamma^{(1)}} &\frac{  \sum_{i=1}^{n^{(0)}} (x^{(0)}_i - \mu)}{(\gamma^{(0)})^2}&\frac{  \sum_{i=1}^{n^{(1)}} (x^{(1)}_i - \mu)}{(\gamma^{(1)})^2}\\ 
\frac{  \sum_{i=1}^{n^{(0)}} (x^{(0)}_i - \mu)}{(\gamma^{(0)})^2}& -\frac{n^{(0)}}{2(\gamma^{(0)})^2} + \sum_{i=1}^{n^{(0)}}\frac{ (x^{(0)}_i-\mu)^2}{(\gamma^{(0)})^3} &0\\
\frac{  \sum_{i=1}^{n^{(1)}} (x^{(1)}_i - \mu)}{(\gamma^{(1)})^2}&0&-\frac{n^{(1)}}{2(\gamma^{(1)})^2}  + \sum_{i=1}^{n^{(1)}}\frac{ (x^{(1)}_i-\mu)^2}{(\gamma^{(1)})^3} \end{array}\right),\]
which we can rewrite in the form
\[H(\theta;x) = A(\theta) +  \sum_{k=0,1} \sum_{\ell = 1,2} \left(\sum_{i=1}^{n^{(k)}} (x_i^{(k)} - \mu)^\ell \right) \cdot A_{k,\ell}(\theta),\]
where $A(\theta)\in\R^{3\times 3}$ and each $A_{k,\ell}(\theta)\in\R^{3\times 3}$ are all continuous matrix-valued functions of $\theta$.
Therefore, we can calculate
\[H(\theta;x) - H(\theta) = \sum_{k=0,1} \left(\sum_{i=1}^{n^{(k)}} (x_i^{(k)} - \mu)  \right) \cdot A_{k,1}(\theta) +  \sum_{k=0,1} \left(\sum_{i=1}^{n^{(k)}} ((x_i^{(k)}- \mu)^2 - \gamma^{(k)})  \right) \cdot A_{k,2}(\theta),\]
and so
\[\norm{H(\theta;x) - H(\theta)} \leq \sum_{k=0,1} \left|\sum_{i=1}^{n^{(k)}} (x_i^{(k)} - \mu) \right| \cdot \norm{A_{k,1}(\theta)} + 
\sum_{k=0,1} \left|\sum_{i=1}^{n^{(k)}} ((x_i^{(k)}- \mu)^2 - \gamma^{(k)})  \right| \cdot \norm{A_{k,2}(\theta)}.\]
Now let $r>0$ be any constant so that $\ball(\theta_0,r)\subseteq\Theta$, and let \[c_r = \sup_{\theta\in\ball(\theta_0,r)}\max_{k=0,1}\max_{\ell=1,2}\norm{A_{k,\ell}(\theta)},\]
which is finite since the $A_{k,\ell}$'s are continuous functions of $\theta$. Then
\[\sup_{\theta \in \ball(\theta_0,r)\cap\Theta}\norm{H(\theta;x) - H(\theta)} \leq c_r \left(\sum_{k=0,1} \left|\sum_{i=1}^{n^{(k)}} (x_i^{(k)} - \mu) \right| + 
\sum_{k=0,1} \left|\sum_{i=1}^{n^{(k)}} ((x_i^{(k)}- \mu)^2 - \gamma^{(k)})  \right| \right).\]
By definition of the distribution of the data we see that the terms $(x_i^{(k)} - \mu) $ are independent and Gaussian,
while the terms $((x_i^{(k)}- \mu)^2 - \gamma^{(k)})$ are independent centered and scaled $\chi^2$ (and therefore subexponential). An elementary calculation then verifies
that
\[\Ewrt{\theta_0}{\exp\left\{\sup_{\theta \in \ball(\theta_0,r)\cap\Theta} t\cdot \norm{H(\theta;X) - H(\theta)}\right\}} \leq e^{c''t^2 n}\textnormal { for all $|t|\leq c'''$,}\]
where $c''$ is chosen to be sufficiently large and $c'''>0$ is chosen to be sufficiently small. Taking $t = \frac{c^2\log n}{n}$, and choosing $n$ sufficiently large so that $c\sqrt{\frac{\log n}{n}}\leq r$ and $t\leq c'''$,
we have established the desired bound~\eqref{eqn:asm:hessian_conc_strong_for_examples} with $\eps_n = \frac{\log^2 n}{n}$ and $c'$ chosen appropriately.

\subsubsection{Checking Assumption~\ref{asm:hessian} for Example~\ref{ex:gaussian_spatial}}
For the Gaussian spatial process (Example~\ref{ex:gaussian_spatial}),
we can calculate
\[H(\theta;x) = \frac{1}{2} x^\top \left(\frac{\partial^2}{\partial \theta^2}\Sigma_\theta^{-1}\right) x + \frac{1}{2}\frac{\partial^2}{\partial \theta^2}\log \det(\Sigma_\theta),\]
and therefore writing $\tilde{x} = \Sigma_{\theta_0}^{-1/2}x$, we have
\begin{align*}
&\norm{H(\theta;x) - H(\theta)}
 = \frac{1}{2}\left\langle xx^\top - \Sigma_{\theta_0}, {\frac{\partial^2}{\partial \theta^2}\Sigma_\theta^{-1}} \right\rangle\\
& = \frac{1}{2}\left\langle \tilde{x}\tilde{x}^\top - \ident_d, \Sigma_{\theta_0}^{1/2}\cdot{\frac{\partial^2}{\partial \theta^2}\Sigma_\theta^{-1}}\cdot\Sigma_{\theta_0}^{1/2} \right\rangle\\
&\leq \frac{1}{2}\left\langle \tilde{x}\tilde{x}^\top - \ident_d, \Sigma_{\theta_0}^{1/2}\cdot{\frac{\partial^2}{\partial \theta^2}\Sigma_{\theta_0}^{-1}}\cdot\Sigma_{\theta_0}^{1/2} \right\rangle
+ \bignorm{\tilde{x}\tilde{x}^\top - \ident_d}
\cdot\frac{1}{2}\norm{\Sigma_{\theta_0}}\Bignorm{\frac{\partial^2}{\partial \theta^2}\Sigma_\theta^{-1} - \frac{\partial^2}{\partial \theta^2}\Sigma_{\theta_0}^{-1}}\\
&\leq  \left\langle\tilde{x}\tilde{x}^\top - \ident_d, \frac{1}{2}\Sigma_{\theta_0}^{1/2} \cdot{\frac{\partial^2}{\partial \theta^2}\Sigma_{\theta_0}^{-1}}\cdot\Sigma_{\theta_0}^{1/2} \right\rangle + \bignorm{\tilde{x}\tilde{x}^\top - \ident_d}
\cdot\frac{1}{2}\norm{\Sigma_{\theta_0}}|\theta - \theta_0|\cdot\sup_{t\in[0,1]}\Bignorm{\frac{\partial^3}{\partial \theta^3}\Sigma_{(1-t)\theta_0+t\theta}^{-1}}.
\end{align*}
Therefore, taking $n$ sufficiently large so that $\ball(\theta_0,c\sqrt{\frac{\log n}{n}})\subseteq\Theta$,
\begin{multline*}
\sup_{\theta\in\ball(\theta_0,c\sqrt{\frac{\log n}{n}})}
\norm{H(\theta;x) - H(\theta)}\\
\leq  \left\langle\tilde{x}\tilde{x}^\top - \ident_d, \frac{1}{2}\Sigma_{\theta_0}^{1/2}\cdot{\frac{\partial^2}{\partial \theta^2}\Sigma_{\theta_0}^{-1}}\cdot\Sigma_{\theta_0}^{1/2} \right\rangle + \bignorm{\tilde{x}\tilde{x}^\top - \ident_d}
\cdot \frac{c}{2}\sqrt{\frac{\log n}{n}}\norm{\Sigma_{\theta_0}}\cdot\sup_{\theta\in\ball(\theta_0,c\sqrt{\frac{\log n}{n}})}\Bignorm{\frac{\partial^3}{\partial \theta^3}\Sigma_\theta^{-1}}.
\end{multline*}
By  \cite[Proposition D.7]{bachoc2014asymptotic}, 
the eigenvalues of $\Sigma_{\theta_0}^{1/2}\cdot{\frac{\partial^2}{\partial \theta^2}\Sigma_{\theta_0}^{-1}}\cdot\Sigma_{\theta_0}^{1/2} $ are bounded
above  by  a constant not depending on $n$, and furthermore $\norm{\Sigma_{\theta_0}}$
and (for sufficiently large $n$) $\sup_{\theta\in\ball(\theta_0,c\sqrt{\frac{\log n}{n}})}\bignorm{\frac{\partial^3}{\partial \theta^3}\Sigma_\theta^{-1}}$ are bounded
by  constants  not depending on $n$. Since $\tilde{x}\sim\normal(0,\ident_n)$, 
 standard tail bounds on the $\chi^2$ distribution
(e.g., \cite[Lemma 1]{laurent2000adaptive}) establish that 
\[\Ewrt{\theta_0}{\exp\left\{t \cdot \sup_{\theta\in\ball(\theta_0,c\sqrt{\frac{\log n}{n}})}
\norm{H(\theta;x) - H(\theta)}\right\} }\leq \exp\left\{c'' \cdot t^2 n + t\cdot \sqrt{\frac{\log n}{n}} \cdot n\right\}\textnormal{ for all $|t|\leq c'''$},\]
where $c''$ is chosen to be sufficiently large and $c'''>0$ is chosen to be sufficiently small. Taking $t = \frac{c^2\log n}{n}$, and choosing $n$ sufficiently large,
we have established the desired bound~\eqref{eqn:asm:hessian_conc_strong_for_examples} with $\eps_n \asymp \sqrt{\frac{\log^3 n}{n}}$
 and $c'$ chosen appropriately.

\subsubsection{Checking Assumption~\ref{asm:hessian} for Example~\ref{ex:multivariate_t}}
For the multivariate t distribution (Example~\ref{ex:multivariate_t}), we first note that since $\theta\in\R^{k\times k}$ is a matrix parameter, the Euclidean norm 
is given by the matrix Frobenius norm, $\fronorm{M} = \sqrt{\sum_{ij}M_{ij}^2}$.
To avoid confusion, when discussing Example~\ref{ex:multivariate_t} we will write $\opnorm{M}$ for the operator
norm on matrices (both for a $k\times k$ matrix, such as the parameter $\theta$ itself, or for a $k^2\times k^2$ linear operator from $\R^{k\times k}$
to $\R^{k\times k}$, such as the Hessian). 

We can first calculate the Hessian, which in this setting will be a linear operator mapping from $\R^{k\times k}$ to $\R^{k\times k}$.
We calculate $H(\theta;x)$ applied to any $A,B\in\R^{k\times k}$ as
\[\big[H(\theta;x)\big](A,B)=   \frac{n}{2}\left\langle \theta^{-1/2} A \theta^{-1/2} ,\theta^{-1/2} B \theta^{-1/2}  \right\rangle  - \frac{\gamma + k}{2} \sum_{i=1}^n \frac{(x_i^\top A x_i)\cdot (x_i^\top B x_i)}{(\gamma + x_i^\top\theta x_i)^2}.\]
For any $\theta$ and any $a\in(0,\tfrac{1}{2})$, if
$(1-a)\theta_0 \preceq \theta \preceq (1+a)\theta_0$, we can verify that
\[ (1-a)^2\cdot (\gamma + z^\top\theta z)^2 \leq (\gamma + z^\top\theta z)^2 \leq (1+a)^2 \cdot (\gamma + z^\top\theta z)^2\]
for all $z\in\R^k$, and therefore 
\[\left| \sum_{i=1}^n \frac{(x_i^\top A x_i)\cdot (x_i^\top B x_i)}{(\gamma + x_i^\top\theta x_i)^2} - \sum_{i=1}^n \frac{(x_i^\top A x_i)\cdot (x_i^\top B x_i)}{(\gamma + x_i^\top\theta_0 x_i)^2}\right| 
\leq n \cdot  \frac{2a+a^2}{1-2a}\cdot \lambda_{\min}(\theta_0)^{-2} \cdot\opnorm{A} \cdot\opnorm{B} .\]
for all $A,B$, where $\lambda_{\min}(\theta_0)>0$ is the minimum eigenvalue of $\theta_0$. Since $\opnorm{}\leq\fronorm{}$,
we have
\[\left| \sum_{i=1}^n \frac{(x_i^\top A x_i)\cdot (x_i^\top B x_i)}{(\gamma + x_i^\top\theta x_i)^2} - \sum_{i=1}^n \frac{(x_i^\top A x_i)\cdot (x_i^\top B x_i)}{(\gamma + x_i^\top\theta_0 x_i)^2}\right| 
\leq n \cdot  \frac{2a+a^2}{1-2a}\cdot \lambda_{\min}(\theta_0)^{-2}\]
for all $A,B$ with $\fronorm{A},\fronorm{B}\leq 1$.
This is sufficient to verify that
\[\sup_{\theta\in\ball(\theta_0,c\sqrt{\frac{\log n}{n}})} \opnorm{H(\theta;x) - H(\theta)}
\leq \opnorm{H(\theta_0;x) - H(\theta_0)} + \sqrt{n \log n} \cdot 3c\lambda_{\min}(\theta_0)^{-2}\]
for all sufficiently large $n$. Therefore, for sufficiently large $n$,
\begin{multline*}
\Ewrt{\theta_0}{\exp\left\{\sup_{\theta\in\ball(\theta_0,c\sqrt{\frac{\log n}{n}})\cap\Theta}\tfrac{c^2\log n}{n}\cdot\opnorm{H(\theta;X)-H(\theta)}\right\}}\\
 \leq \exp\left\{\tfrac{c^2\log n}{n} \cdot  \sqrt{n \log n} \cdot 3c\lambda_{\min}(\theta_0)^{-2}\right\} \cdot \Ewrt{\theta_0}{\exp\left\{\tfrac{c^2\log n}{n}\cdot\opnorm{H(\theta_0;X)-H(\theta_0)}\right\}}.
\end{multline*}
Next, $H(\theta_0;X)$ is equal to a constant plus a sum of $n$ i.i.d.~terms, with each term bounded uniformly, since
\[\left|\frac{\gamma+k}{2}  \frac{(X_i^\top A X_i)\cdot (X_i^\top B X_i)}{(\gamma + X_i^\top\theta_0 X_i)^2}\right| \leq \frac{\gamma+k}{2}\cdot \lambda_{\min}(\theta_0)^{-2}\]
holds for all $A,B$ with $\fronorm{A},\fronorm{B}\leq 1$,  almost surely over $X_i$.
Therefore, by the matrix Hoeffding inequality \citep[Theorem 1.3]{tropp2012user}, we have
\[\prwrt{\theta_0}{n^{-1/2}\opnorm{H(\theta_0;X) - H(\theta_0)} > t} \leq  2k^2\exp\left\{ - \frac{t^2}{8\cdot (\frac{\gamma+k}{2})^2\cdot \lambda_{\min}(\theta_0)^{-4}} \right\}\]
for any $t>0$. In other words, $n^{-1/2}\opnorm{H(\theta_0;X) - H(\theta_0)}$ is subgaussian with parameter not depending on $n$
and therefore
\[\Ewrt{\theta_0}{\exp\left\{\tfrac{c^2\log n}{n}\cdot\opnorm{H(\theta_0;X)-H(\theta_0)}\right\}} \leq \exp\left\{\frac{c''\log^2 n}{n}\right\}\]
for an appropriately chosen $c''$.

Combining everything, we have established that
the bound~\eqref{eqn:asm:hessian_conc_strong_for_examples} holds with $\eps_n \asymp \sqrt{\frac{\log^3 n}{n}}$
 and $c'$ chosen appropriately.

\subsection{Checking Assumption~\ref{asm:thetah}}
Before giving proofs for our specific examples, we pause to discuss  Assumption~\ref{asm:thetah} more generally,
 to see that this assumption will be plausible for many common settings (beyond the few that we study here).
We consider the following general scenario. Suppose that we have access to a consistent initial estimate $\thetahi(X)$ 
of $\theta_0$. Then under some standard conditions on the negative log-likelihood surface, 
by constraining $\thetah(X,W)$ to a neighborhood of $\thetahi(X)$,
we can ensure that $\thetah(X,W)$ will satisfy the needed assumptions.

\begin{lemma}\label{lem:theta_init}
Let
\[\thetahi:\Xcal\rightarrow\Theta\textnormal{\quad and\quad }
\widehat{r}_{\textnormal{init}}:\Xcal\rightarrow\R_+\]
be any maps such that $\ball(\thetahi(x),\widehat{r}_{\textnormal{init}}(x))\subseteq\Theta$ for all $x\in\Xcal$.
Suppose that, under the distribution $X\sim P_{\theta_0}$,
 the following statements all hold with probability at least $1-\delta_{\textnormal{init}}(\theta_0)$:
\begin{equation}\label{eqn:theta_init}\begin{cases}
\norm{\thetahi(X)-\theta_0}\leq r_{\textnormal{init}}(\theta_0),\\
\textnormal{$\Lcal(\theta;X)$ has a FOSP in $\ball(\theta_0, r_{\textnormal{init}}(\theta_0))$},\\
\textnormal{$\nabla_\theta^2\Lcal(\theta;X)\succeq \lambda_{\textnormal{cvx}}(\theta_0)\ident_d$ for all $\theta\in\ball(\theta_0,r_{\textnormal{cvx}}(\theta_0))$},\\
3r_{\textnormal{init}}(\theta_0) \leq \widehat{r}_{\textnormal{init}}(X) \leq  r_{\textnormal{cvx}}(\theta_0) - r_{\textnormal{init}}(\theta_0) ,
\end{cases}\end{equation}
for some constants $r_{\textnormal{init}}(\theta_0) ,r_{\textnormal{cvx}}(\theta_0),\lambda_{\textnormal{cvx}}(\theta_0)>0$.
If $\thetah:\Xcal\times\R^d\mapsto\Theta$ is any function that 
 maps each point $(x,w)$ to some FOSP of the constrained optimization problem
\[\argmin_{\theta\in \ball(\thetahi(x),r(\thetahi(x)))} \Lcal(\theta;x,w),\]
then Assumption~\ref{asm:thetah} is satisfied with 
\[r(\theta_0)=2r_{\textnormal{init}}(\theta_0)\textnormal{\ and \ }\delta(\theta_0)=\delta_{\textnormal{init}}(\theta_0) + \exp\left\{ -\frac{1}{2} \max\left\{\frac{r_{\textnormal{init}}(\theta_0)\lambda_{\textnormal{cvx}}(\theta_0)}{\sigma} - 1,0\right\}^2\right\}.\]
\end{lemma}

With this lemma in place, we will now turn to verifying that its conditions hold for each of our four examples.
Specifically, for each example, we will propose an initial estimator $\thetah_{\textnormal{init}}(X)$ such that the 
conditions of the lemma are satisfied with $r_{\textnormal{init}}(\theta_0)\asymp \sqrt{\frac{\log n}{n}}$ and $r_{\textnormal{cvx}}(\theta_0)\asymp 1$
and $\lambda_{\textnormal{cvx}}(\theta_0)\asymp n$.

\subsubsection{Checking the conditions of Lemma~\ref{lem:theta_init}: general recipe}
After fixing some $\theta_0\in\Theta$, each proof will follow the same general recipe:
\begin{itemize}
\item We will verify that 
\begin{equation}\label{eqn:recipe1}H(\theta_0)\succeq C_1 n \mathbf{I}_d,\end{equation}
where $C_1>0$ does not depend on $n$.
Combined with Assumption~\ref{asm:hessian} (which we verified above for each of our examples), this means that
for sufficiently large $n$ it holds that
$\nabla_\theta^2\Lcal(\theta;X) = H(\theta;X) \succeq  C_2n\ident_3$ for all $\theta\in\ball(\theta_0,C_3)$
for appropriately chosen $C_2,C_3>0$, with probability at least $1-n^{-1}$. Thus we can take $\lambda_{\textnormal{cvx}} = C_2$ and  $r_{\textnormal{cvx}} = C_3$.
\item We will define an initial estimator $\thetahi(x) $ and will prove that we can find a constant $C_4$ not depending on $n$ such that
\begin{equation}\label{eqn:recipe2}\prwrt{\theta_0}{\norm{\thetahi(X) - \theta_0} \leq C_4\sqrt{\frac{\log n}{n}}} \geq 1- n^{-1}\end{equation}
for all sufficiently large $n$. Thus we can take $r_{\textnormal{init}} = C_4\sqrt{\frac{\log n}{n}}$.
Furthermore, choosing $\widehat{r}_{\textnormal{init}}(x)$ to be any function of $n$ that vanishes slower than $\sqrt{\frac{\log n}{n}}$ (e.g., $\widehat{r}_{\textnormal{init}}(x)\equiv n^{-1/4}$),
we have verified that $3r_{\textnormal{init}}(\theta_0) \leq \widehat{r}_{\textnormal{init}}(X) \leq  r_{\textnormal{cvx}}(\theta_0) - r_{\textnormal{init}}(\theta_0)$ holds.
\item Finally we will show that we can find a constant $C_5$ not depending on $n$ such that
\begin{equation}\label{eqn:recipe3}\prwrt{\theta_0}{\norm{\nabla_\theta \log f(X;\theta_0)}\leq C_5\sqrt{n\log n}}\geq 1-n^{-1},\end{equation}
for all sufficiently large $n$. Combined with the bound $\nabla_\theta^2\Lcal(\theta;X) \succeq  C_2n\ident_3$ for all $\theta\in\ball(\theta_0,C_3)$
that is already established, this means that $\Lcal(\theta;X) = -\log f(X;\theta)$ has a FOSP in $\ball(\theta_0,C_5C_2^{-1}\sqrt{\frac{\log n}{n}})$
and so we can take $r_{\textnormal{init}}(\theta_0) = C_5C_2^{-1}\sqrt{\frac{\log n}{n}}$.
\end{itemize}

\subsubsection{Checking the conditions of Lemma~\ref{lem:theta_init} for Example~\ref{ex:canonical_glm}}
For the canonical GLM setting (Example~\ref{ex:canonical_glm}), first we have
\[H(\theta_0;x) = \sum_{i=1}^n Z_iZ_i^\top \cdot a''(Z_i^\top\theta_0) \succeq n \cdot C_1 \cdot \mathbf{I}\]
for some $C_1>0$ that does not depend on $n$, since we have assumed $\max_{ij}|Z_{ij}|$ is bounded by a constant
and $\frac{1}{n}\sum_i Z_iZ_i^\top\succeq \lambda_0 \ident_d$. Thus~\eqref{eqn:recipe1} holds.
Next we verify~\eqref{eqn:recipe2}. Since the negative log-likelihood is strictly convex everywhere,
we can define $\thetahi(x)$ to equal a global minimizer of $-\log f(x;\theta)$, if one exists (i.e., finding a global minimizer is computationally feasible
since it is a differentiable and strictly convex minimization problem). 
Therefore, if a FOSP exists in a $\bigo\Big(\sqrt{\frac{\log n}{n}}\Big)$ neighborhood of $\theta_0$ (as we will establish next), then~\eqref{eqn:recipe2} is satisfied.
Finally we check~\eqref{eqn:recipe3} to verify the existence of the FOSP. We calculate
\[\nabla_\theta [-\log f(x;\theta_0)] = \sum_{i=1}^n Z_i \left(a'(Z_i^\top\theta_0) - x_i\right),\]
and by standard calculations for GLMs, $X$ is subexponential with
\[\Ewrt{\theta_0}{e^{tX_i}} = e^{a(Z_i^\top\theta_0 + t) - a(Z_i^\top\theta_0)}\]
for any $t\in\R$ and for each $i=1,\dots,n$.
Since we have assumed $\max_{ij}|Z_{ij}|$ is bounded by a constant, proving~\eqref{eqn:recipe3} is a standard calculation.

\subsubsection{Checking the conditions of Lemma~\ref{lem:theta_init} for Example~\ref{ex:behrens_fisher}}
For the Behrens--Fisher problem (Example~\ref{ex:behrens_fisher}), write $\theta_0 = (\mu_0,\gamma^{(0)}_0,\gamma^{(1)}_0)$.
We first calculate
\[H(\theta_0) = \Ewrt{\theta_0}{H(\theta_0;X)} =\left(\begin{array}{ccc}\frac{n^{(0)}}{\gamma^{(0)}_0}+\frac{n^{(1)}}{\gamma^{(1)}}_0 &0&0\\
0& \frac{n^{(0)}}{2(\gamma^{(0)}_0)^2} &0\\
0&0&\frac{n^{(1)}}{2(\gamma^{(1)}_0)^2}  \end{array}\right)\succeq  cn \cdot \frac{\min\{n^{(0)},n^{(1)}\}}{\max\{n^{(0)},n^{(1)}\}} \cdot\ident_3,\]
where the  inequality holds for some $c>0$ depending only on $\theta_0$. Recalling that we have assumed $\frac{\max\{n^{(0)},n^{(1)}\}}{\min\{n^{(0)},n^{(1)}\}}$ 
is bounded by a constant, this means that
\[H(\theta_0) \succeq c' n \ident_3\]
for some $c'>0$ that does not depend on $n$, which verifies~\eqref{eqn:recipe1}.

Next we define an initial estimator
\[\thetahi(x) = \big(\widehat\mu_{\textnormal{init}}(x), \widehat\gamma^{(0)}_{\textnormal{init}}(x), \widehat\gamma^{(1)}_{\textnormal{init}}(x)\big)\]
where
\[\widehat\mu_{\textnormal{init}}(x) = \frac{1}{n}\left(\sum_{i=1}^{n^{(0)}}x^{(0)}_i + \sum_{i=1}^{n^{(1)}}x^{(1)}_i\right)\]
for $n=n^{(0)}+n^{(1)}$,
and
\[\widehat\gamma^{(k)}_{\textnormal{init}}(x) =  \frac{1}{n^{(k)}}\sum_{i=1}^{n^{(k)}}(x^{(k)}_i -\widehat\mu_{\textnormal{init}}(x))^2\]
for each $k=0,1$. By standard Gaussian and $\chi^2$ tail bounds, we can easily see that for sufficiently large $c''$ (not depending on $n$) it holds that
\[\prwrt{\theta_0}{\norm{\thetahi(X) - \theta_0} \leq c''\sqrt{\frac{\log n}{n}}} \geq 1- n^{-1}\]
for sufficiently large $n$, which verifies~\eqref{eqn:recipe2}.

Finally, 
we calculate
\[\nabla \log f(x;\theta) = - \left(\begin{array}{c}-\sum_{k=0,1} \frac{  \sum_{i=1}^{n^{(k)}} (x^{(k)}_i - \mu)}{\gamma^{(k)}} \\ 
 \frac{n^{(0)}}{2\gamma^{(0)}_0} - \sum_{i=1}^{n^{(0)}}\frac{ (x^{(0)}_i-\mu_0)^2}{2(\gamma^{(0)}_0)^2} \\
\frac{n^{(1)}}{2\gamma^{(1)}_0}  - \sum_{i=1}^{n^{(1)}}\frac{ (x^{(1)}_i-\mu_0)^2}{2(\gamma^{(1)}_0)^2} \end{array}\right),\]
and therefore each entry of $\nabla \log f(X;\theta_0)$ is a sum of $n$ or $n^{(0)}$ or $n^{(1)}$ many i.i.d.~zero-mean subexponential terms.
Therefore,  we can find a constant $c'''$ such that
\[\prwrt{\theta_0}{\norm{\nabla \log f(X;\theta_0) } \leq c'''\sqrt{n\log n}}\geq 1-n^{-1}\]
for sufficiently large $n$, which verifies~\eqref{eqn:recipe3} and thus completes the proof.

\subsubsection{Checking the conditions of Lemma~\ref{lem:theta_init} for Example~\ref{ex:gaussian_spatial}}
For the Gaussian spatial process (Example~\ref{ex:gaussian_spatial}),
first, recall our calculation
\[H(\theta;x) = \frac{1}{2} x^\top \left(\frac{\partial^2}{\partial \theta^2}\Sigma_{\theta_0}^{-1}\right) x + \frac{1}{2}\frac{\partial^2}{\partial \theta^2}\log \det(\Sigma_{\theta_0}),\]
which we can calculate explicitly as
\begin{multline*}H(\theta;X)
=  \frac{1}{2}x^\top\left( - \Sigma_\theta^{-1}(D\circ D\circ \Sigma_\theta)\Sigma_\theta^{-1} + 2\Sigma_\theta^{-1}(D\circ \Sigma_\theta)\Sigma_\theta^{-1}(D\circ \Sigma_\theta)\Sigma_\theta^{-1}  \right)x \\
 + \frac{1}{2}\textnormal{trace}(\Sigma_\theta^{-1/2}(D\circ D\circ \Sigma_\theta) \Sigma_\theta^{-1/2})
 -\frac{1}{2}\norm{\Sigma_\theta^{-1/2}(D\circ\Sigma_\theta)\Sigma_\theta^{-1/2}}^2,
\end{multline*}
and so since $\Ewrt{\theta_0}{XX^\top}=\Sigma_{\theta_0}$, we have
\[H(\theta_0) = \Ewrt{\theta_0}{H(\theta_0;X)} = 
\frac{1}{2}\norm{\Sigma_{\theta_0}^{-1/2}(D\circ\Sigma_{\theta_0})\Sigma_{\theta_0}^{-1/2}}^2
\geq \frac{1}{2}\lambda_{\min}(\Sigma_{\theta_0})^{-2}\norm{D\circ\Sigma_{\theta_0}}^2.\]
We know from \cite[Proposition D.7]{bachoc2014asymptotic} that $\Sigma_{\theta_0}$ has eigenvalues bounded above and below by positive constants.
Furthermore,
\[\norm{D\circ\Sigma_{\theta_0}}^2 = \sum_{i=1}^n\sum_{j=1}^n D_{ij}^2\cdot(\Sigma_{\theta_0})_{ij}^2 \geq \sum_{(i,j)\in E} D_{ij}^2\cdot(\Sigma_{\theta_0})_{ij}^2 =\sum_{(i,j)\in E} 1 \cdot e^{-2\theta_0} = e^{-2\theta_0}\cdot|E|,\]
where  $E\subseteq \{1,\dots,n\}\times\{1,\dots,n\}$ be the set of all pairs $(i,j)$ such that $D_{ij} = 1$.
Since $|E|\geq n$, we have shown that~\eqref{eqn:recipe1} holds for some appropriately
chosen $C_1$ that does not depend on $n$.

Next we need to define our initial estimator to satisfy~\eqref{eqn:recipe2}. We will define a simple choice for intuition (this choice is of course not necessarily
optimal in any sense). Define
\[\thetahi(x) = - \log\left(\frac{1}{|E|} \sum_{(i,j)\in E} x_i x_j\right).\]
We need to check that, with  probability at least $1-n^{-1}$, $|\thetahi(X) - \theta_0|\leq C_4\sqrt{\frac{\log n}{n}}$ for some constant $C_4$ not depending on $n$. 
Since $\theta_0>0$, it is equivalent to check that, with probability at least $1-n^{-1}$,
\[\left|\frac{1}{|E|} \sum_{(i,j)\in E} X_i X_j - e^{-\theta_0}\right| \leq C'\sqrt{\frac{\log n}{n}}\]  for some constant $C'$ not depending on $n$. 
Let $A$ be the adjacency matrix, with entry $A_{ij} = \One{(i,j)\in E}$, and  let $U\Lambda U^\top = \Sigma_{\theta_0}^{1/2} A  \Sigma_{\theta_0}^{1/2} $ be an eigendecomposition. Then
\[\left|\frac{1}{|E|} \sum_{(i,j)\in E} X_i X_j - e^{-\theta_0}\right| = \frac{1}{|E|}\left|\left\langle XX^\top - \Sigma_{\theta_0} , A\right\rangle\right|
=\frac{1}{|E|} \left|\left\langle (U^\top \Sigma_{\theta_0}^{-1/2} X)(U^\top \Sigma_{\theta_0}^{-1/2} X)^\top - \mathbf{I}_n , \Lambda\right\rangle\right|.\]
Since $U^\top \Sigma_{\theta_0}^{-1/2} X\sim \normal(0,\ident_n)$, while $|E|\geq n$,
the desired bound holds as long as the values $\Lambda_{11},\dots,\Lambda_{nn}$ (i.e., the eigenvalues of $\Sigma_{\theta_0}^{1/2} A  \Sigma_{\theta_0}^{1/2}$)
are bounded by some constant $C''$ not depending on $n$. 
Since the eigenvalues of $\Sigma_{\theta_0}$ are bounded by a constant not depending on $n$ by \cite[Proposition D.7]{bachoc2014asymptotic},
equivalently we need to verify that $\norm{A}\leq C'''$ for some constant $C'''$ not depending on $n$---in
 fact, since $A$ is the adjacency matrix of a graph where each vertex has at most $2k$ many neighbors,
we have $\norm{A}\leq 2k$. This establishes~\eqref{eqn:recipe2}.

Finally we verify~\eqref{eqn:recipe3}.  We calculate
\begin{multline*}\frac{\partial}{\partial \theta}\log f(x;\theta_0) = -  \frac{1}{2}x^\top \left(\frac{\partial}{\partial \theta}\Sigma_{\theta_0}^{-1}\right)x  - \frac{1}{2}\frac{\partial}{\partial\theta}\log\det(\Sigma_{\theta_0}) \\
 = -  \frac{1}{2}(\Sigma_{\theta_0}^{-1/2}x)^\top \cdot\Sigma_{\theta_0}^{1/2} \left(\frac{\partial}{\partial \theta}\Sigma_{\theta_0}^{-1}\right)\Sigma_{\theta_0}^{1/2}\cdot (\Sigma_{\theta_0}^{-1/2}x)  - \frac{1}{2}\frac{\partial}{\partial\theta}\log\det(\Sigma_{\theta_0}) 
 \end{multline*}
We know that $\Ewrt{\theta_0}{\frac{\partial}{\partial \theta}\log f(X;\theta_0)}=0$, and moreover,
 $\Sigma_{\theta_0}^{-1/2}X\sim \normal(0,\ident_n)$ and so this quantity has distribution equal  to a weighted sum of centered $\chi^2$ random variables.
By \cite[Proposition D.7]{bachoc2014asymptotic} we know that the eigenvalues of the matrix $\Sigma_{\theta_0}^{1/2} \left(\frac{\partial}{\partial \theta}\Sigma_{\theta_0}^{-1}\right)\Sigma_{\theta_0}^{1/2}$
are bounded by a constant that does not depend on $n$, 
standard $\chi^2$ tail bounds (see, e.g., \cite[Lemma 1]{laurent2000adaptive}) establish that~\eqref{eqn:recipe3} holds for an appropriately chosen $C_5$ not depending on $n$.

\subsubsection{Checking the conditions of Lemma~\ref{lem:theta_init} for Example~\ref{ex:multivariate_t}}
For the multivariate t distribution (Example~\ref{ex:multivariate_t}), 
calculations in \cite[Appendix B]{lange1989robust} show that
\[\big[H(\theta_0)\big]\big(M,M\big)
 = \frac{n}{2} \left(\frac{\gamma + k}{\gamma + k + 2}\fronorm{\theta_0^{-1/2}M\theta_0^{-1/2}}^2 - \frac{1}{\gamma+k + 2}\textnormal{trace}(\theta_0^{-1/2}M\theta_0^{-1/2})^2\right).\]
 Since $\textnormal{trace}(A)\leq \sqrt{k}\fronorm{A}$ for any $A\in\R^{k\times k}$, then we have
\[\big[H(\theta_0)\big]\big(M,M\big)\geq\frac{n}{2} \cdot \frac{\gamma}{\gamma + k + 2}\fronorm{\theta_0^{-1/2}M\theta_0^{-1/2}}^2 ,\]
and so~\eqref{eqn:recipe1} holds with $C_1 = \frac{1}{2}\cdot \frac{\gamma}{\gamma + k + 2}\cdot \opnorm{\theta_0}^{-2}$.

Next, we define our initial estimator.
We will work with the Kendall's $\tau$ correlation: given a data point $x\in(\R^k)^n$, for each $j,j'\in\{1,\dots,k\}$ define
\[T_{jj'} = \frac{1}{{n\choose 2}}\sum_{1\leq i < i' \leq n} \textnormal{sign}\left((x_{ij} - x_{i'j})\cdot(x_{ij'} - x_{i'j'})\right),\]
let $S = S(x)\in\R^{k\times k}$ be defined with entries
$S_{jj'} = \sin\left(\frac{\pi}{2}\cdot T_{jj'}\right)$.
It is well known that for a continuous elliptical distribution (such as the multivariate t), this transformation 
yields an unbiased estimate of the correlation matrix.
We will also estimate $V_j = \frac{\textnormal{Median of $|x_{1j}|,\dots,|x_{nj}|$}}{\textnormal{$0.75$-quantile of $t_{\gamma}$}}$, and let $\Sigma = \Sigma(x)$ have entries
\[\Sigma_{jj'} = S_{jj'}\sqrt{V_j V_{j'}}.\]
Next let $\thetahi(x) = \Sigma(x)^{-1}$ (or define it to take any value if $\Sigma(x)$ is not invertible).

By \cite[Corollary 4.8]{barber2018rocket},  if $n\geq k\log n$, then with probability at least $1-n^{-1}$,
\[\opnorm{S(X) - S_*} \leq C  \sqrt{\frac{k\log n}{n}} \]
for a universal constant $C$,
where $S_*$ is the true correlation matrix, i.e.,
\[(S_*)_{jj'} = \frac{(\theta_0^{-1})_{jj'}}{\sqrt{(\theta_0)^{-1}_{jj}(\theta_0^{-1})_{j'j'}}}.\]
We also have $X_{ij}\iidsim \sqrt{(\theta_0^{-1})_{jj}}\cdot t_\gamma$ (a univariate $t$ distribution) and so we can easily verify that
\[\max_{j=1,\dots,k}|V_j - (\theta_0^{-1})_{jj}| \leq C' \sqrt{\frac{\log n}{n}}\]
with probability at least $1-n^{-1}$.
Combining these bounds, this means that
\[\opnorm{\thetahi(X) - \theta_0} \leq C''  \sqrt{\frac{k\log n}{n}} .\]
Since this is a $k\times k$ matrix, therefore
\[\fronorm{\thetahi(X) - \theta_0} \leq C''\sqrt{\frac{k^2\log n}{n}} = C'''\sqrt{\frac{\log n}{n}} ,\]
which verifies~\eqref{eqn:recipe2}.

Finally we check~\eqref{eqn:recipe3}.
We calculate
\[\nabla\log f(X;\theta_0)
 = \frac{n}{2}\theta_0^{-1} - \frac{\gamma+k}{2} \sum_{i=1}^n \frac{X_iX_i^\top}{\gamma + X_i^\top \theta_0 X_i},\]
 which is a sum of $n$ i.i.d.~mean-zero terms.
 Observe also that for any $z\in\R^k$, $\opnorm{ \frac{zz^\top}{\gamma + z^\top \theta_0 z}} \leq \lambda_{\min}(\theta_0)^{-1}$, so the terms
are uniformly bounded. 
By the matrix Hoeffding inequality \citep[Theorem 1.3]{tropp2012user} along with the bound  
$\fronorm{\nabla\log f(X;\theta_0)}\leq \sqrt{k}\opnorm{\nabla\log f(X;\theta_0)}$,
 we therefore have
\[\prwrt{\theta_0}{\fronorm{\nabla\Lcal(\theta_0;X)} \geq t\sqrt{k}} \leq 2k\exp\left\{-\frac{t^2}{8n\cdot (\frac{\gamma+k}{2})^2 \cdot \lambda_{\min}(\theta_0)^{-2}}\right\}\]
for any $t>0$. Taking $t\asymp \sqrt{n\log n}$ is sufficient to establish~\eqref{eqn:recipe3}.

\subsection{Proof of Lemma~\ref{lem:theta_init}}
Suppose that the statements~\eqref{eqn:theta_init} all hold, which is satisfied with probability at least $1-\delta_{\textnormal{init}}(\theta_0)$ by assumption.
Suppose also that the random vector $W$ satisfies
 \[\norm{W}< \frac{r_{\textnormal{init}}(\theta_0)\lambda_{\textnormal{cvx}}(\theta_0)}{\sigma}.\]
 Since $\norm{W}^2\sim \tfrac{1}{d}\chi^2_d$ by definition, using standard $\chi^2$ tail bounds (see, e.g., \cite[Lemma 1]{laurent2000adaptive})
we can calculate
\[\pr{\norm{W} <  \frac{r_{\textnormal{init}}(\theta_0)\lambda_{\textnormal{cvx}}(\theta_0)}{\sigma}} \geq 1- \exp\left\{ -\frac{1}{2} \max\left\{\frac{r_{\textnormal{init}}(\theta_0)\lambda_{\textnormal{cvx}}(\theta_0)}{\sigma} - 1,0\right\}^2\right\}.\]
Therefore, with probability at least $1-\delta(\theta_0)$ (where $\delta(\theta_0)$ is defined as in the statement of the lemma),
the bounds~\eqref{eqn:theta_init} all hold and $\norm{W}$ satisfies the bound above.
From this point on we will assume these bounds all hold.

Let $\theta_*\in\ball(\theta_0,r_{\textnormal{init}}(\theta_0))$ be a FOSP of $\Lcal(\theta;X)$,
and let $\thetah=\thetah(X,W)\in\ball(\thetahi(X),\widehat{r}_{\textnormal{init}}(X))$ be a FOSP of the constrained problem
\[\min_{\theta\in\ball(\thetahi(X),\widehat{r}_{\textnormal{init}}(X))} \Lcal(\theta;X,W).\]
Then \[\norm{\theta_* - \thetahi(X)}\leq \norm{\thetahi(X)-\theta_0}+\norm{\theta_*-\theta_0}\leq 2r_{\textnormal{init}}(\theta_0)\leq \widehat{r}_{\textnormal{init}}(X),\]
and so $\theta_*$ also lies in the convex constraint set $\ball(\thetahi(X),\widehat{r}_{\textnormal{init}}(X))$.
Since $\norm{\thetahi(X)-\theta_0}\leq r_{\textnormal{init}}(\theta_0) \leq  r_{\textnormal{cvx}}(\theta_0) - \widehat{r}_{\textnormal{init}}(X)$, this means that $\thetah$ and $\theta_*$
both lie in $\ball(\theta_0,r_{\textnormal{cvx}}(\theta_0))$, and so we have $\lambda_{\textnormal{cvx}}(\theta_0)$-strong convexity in this region.
Therefore we have
\begin{align*}
0&\leq (\theta_*-\thetah)^\top \nabla_\theta\Lcal(\thetah;X,W)\\
&=  (\theta_*-\thetah)^\top \nabla_\theta\Lcal(\thetah;X) + \sigma (\theta_*-\thetah)^\top W\\
&\leq  (\theta_*-\thetah)^\top \nabla_\theta\Lcal(\theta_*;X)  - \lambda_{\textnormal{cvx}}(\theta_0)\norm{\theta_*-\thetah}^2+ \sigma \norm{\theta_*-\thetah}\norm{W}\\
&\leq  - \lambda_{\textnormal{cvx}}(\theta_0)\norm{\theta_*-\thetah}^2 + \sigma \norm{\theta_*-\thetah}\norm{W}\\
&<  - \lambda_{\textnormal{cvx}}(\theta_0)\norm{\theta_*-\thetah}^2 + r_{\textnormal{init}}(\theta_0)\lambda_{\textnormal{cvx}}(\theta_0)\norm{\theta_*-\thetah},
\end{align*}
where the next-to-last step holds since $\theta_*$ is a FOSP of the unconstrained problem $\min_\theta\Lcal(\theta;X)$,
and the last step holds
as long as $\norm{\theta_*-\thetah}>0$ by our bound on $\norm{W}$. Therefore, we must have
\[\norm{\thetah-\theta_*} <  r_{\textnormal{init}}(\theta_0).\]
In particular this implies
\[\norm{\thetah-\theta_0}\leq \norm{\thetah-\theta_*}+\norm{\theta_*-\theta_0} <  2r_{\textnormal{init}}(\theta_0)\leq \widehat{r}_{\textnormal{init}}(X).\]
Finally, we need to check that $\thetah$ is a SSOSP.
We have
\[\norm{\thetah - \thetahi(X)}\leq \norm{\thetah-\theta_0} + \norm{\thetahi(X)-\theta_0} <3r_{\textnormal{init}}(\theta_0) \leq \widehat{r}_{\textnormal{init}}(X),\]
which means that $\thetah$ is in the interior of the constraint set $\ball(\thetahi(X),\widehat{r}_{\textnormal{init}}(X))$. Therefore, $\thetah$ must be a FOSP
of the unconstrained problem $\min_\theta\Lcal(\theta;X,W)$.
Finally, since $\thetah\in\ball(\theta_0,r_{\textnormal{cvx}}(\theta_0))$ as calculated above, $\Lcal(\theta;X)$ (and therefore also $\Lcal(\theta;X,W)$) has strong convexity at $\theta=\thetah$.
This completes the proof.


\section{Computational considerations}\label{app:compute}

\subsection{Optimization of \eqref{eqn:thetah_def}}\label{app:optimization}
If the unperturbed penalized maximum likelihood problem is (strongly) convex, then \eqref{eqn:thetah_def} is also (strongly) convex. Since the linear perturbation only changes the gradient by a fixed constant and does not affect the Hessian, any convex solver that relies on first- and second-order derivatives to solve the unperturbed problem can be immediately adapted to run on \eqref{eqn:thetah_def}. Note that even strong convexity does not guarantee the unperturbed penalized maximum likelihood problem has any local optima, since $\Theta$ could be constrained to a region with no minima. However, as long as the unperturbed problem is convex and has a local optimum, the perturbation can only lead to a lack of local optima if there exists a direction $z\in\R^d$ such that 
\begin{equation}\label{eq:tipover}
-\sigma W^\top z \ge \max_{\theta\in\Theta}\{\nabla_\theta\Lcal(\theta; X)^\top z\}.
\end{equation}
Since we control $\sigma$ in the aCSS algorithm, we can always choose it to be sufficiently small as to make \eqref{eq:tipover} very unlikely (and
moreover, if $\Theta=\R^d$ is unconstrained and the unperturbed problem is strongly convex, then~\eqref{eq:tipover} cannot occur at any $\sigma$).
 Indeed, when $X$ is composed of $n$ i.i.d. samples, the right-hand side of \eqref{eq:tipover} will grow at a rate of $\sqrt{n}$, while in Section~\ref{sec:theory_asymptotic_view}, we noted that Theorem~\ref{thm:main} required $\sigma$, and hence the left-hand side of \eqref{eq:tipover}, to grow at a rate that is vanishing compared to $\sqrt{n}$. The same story holds for non-convex functions locally for a well-behaved basin of attraction: the random perturbation can cause problems but not if you choose it sufficiently small. Note that the cost of $\thetah$ failing to return a SSOSP of \eqref{eqn:thetah_def} is conservativeness of the aCSS test (but not loss of validity!), since when it fails to return a SSOSP the test will return a $p$-value of 1.

Although $\sigma$ can always be chosen to be very small, this can incur a different computational cost in terms of sampling the copies $\Xt^{(m)}$. In particular, as we will see in the next subsection, reducing $\sigma$ leads to ``smaller" MCMC steps, i.e., starting at some state $X'$ and taking a single step in the reversible Markov chain we will use for sampling will produce a state that is highly-related to $X'$ or may even be identical to it with high probability. One solution to this is to simply take $L$, the number of steps we take in the Markov chain between samples, to be very large, so at least with sufficient computational resources it should always be possible to choose $\sigma$ sufficiently small so as to not adversely affect the optimization of \eqref{eqn:thetah_def} relative to the unperturbed maximum likelihood problem.

\subsection{Sampling the conditional randomizations}\label{app:sampling_randomizations}
Due to the conditioning on $\thetah$, the solution to an optimization problem, we only expect to be able to perform the exact sampling i.i.d. from Equation~\eqref{eqn:density_hat} in special cases when both the conditional distribution of $X$ is very simple and $\thetah$ can be found in closed form. Aside from very special cases, we expect almost any model and/or estimator to require one of the MCMC samplers.

Recall the density we are targeting in Equation~\eqref{eqn:density_hat}:
\begin{equation*} 
p_{\thetah}(x\giv \thetah) \propto f(x;\thetah)\cdot \exp\left\{ - \frac{\norm{\nabla_\theta\Lcal(\thetah;x)}^2}{2\sigma^2/d}\right\}\cdot \det\left(\nabla^2_\theta \Lcal(\thetah;x)\right)\cdot\one{x\in\Xcal_{\thetah}},
\end{equation*}
with respect to the base measure $\nu_\Xcal$.
Both MCMC sampling schemes assume the ability to take steps in a reversible Markov chain
whose stationary distribution has the above density.
We will now show that it is feasible to construct an efficient sampling scheme using Metropolis--Hastings (MH).

Given $\thetah$, we first choose a proposal distribution $q_{\thetah}(x\giv x')$---we will discuss this choice below.
Fixing $q_{\thetah}(x\giv x')$, 
 we can write the MH acceptance probability for a proposal $x$ from a previous iteration $x'$ as
\[
A_{\thetah}(x\giv x')
:= \min\left\{1, \frac{p_{\thetah}(x\giv\thetah)\,q_{\thetah}(x'\giv x)}{p_{\thetah}(x'\giv\thetah) \,q_{\thetah}(x\giv x')} \right\}.\]
Our reversible MCMC is then given by the following: 
\begin{itemize}
\item Starting at state $x'$, generate a proposal $x$ according to the proposal distribution $q_{\thetah}(\cdot\giv x')$.
\item With probability $A_{\thetah}(x\giv x')$, set the next state to equal $x$. Otherwise, the next state is set to equal $x'$.
\end{itemize}

To verify that this yields a computationally feasible method, we need to check two things: first,
that the acceptance probability $A_{\thetah}(x\giv x')$ is not too low (i.e., its average value is bounded away from zero), in order to ensure
that our chain length $L$ does not need to be taken to be too large,
and second, that the acceptance probability $A_{\thetah}(x\giv x')$ can be calculated efficiently.
The first consideration, ensuring that $A_{\thetah}(x\giv x')$ is not too low, will be specific to the problem and will discuss
this for specific examples below. To check that we can efficiently calculate the acceptance probability $A_{\thetah}(x\giv x')$, 
by definition of $p_{\thetah}(\cdot\giv\thetah)$ we see that $A_{\thetah}(x\giv x')$ can be written as
\[
A_{\thetah}(x\giv x')
=\min\left\{1, \ \frac{q_{\thetah}(x'\giv x)}{ q_{\thetah}(x\giv x')} 
\cdot
 \frac{f(x;\thetah)\exp\left\{ - \frac{\norm{\nabla_\theta\Lcal(\thetah;x)}^2}{2\sigma^2/d}\right\}\det\left(\nabla^2_\theta \Lcal(\thetah;x)\right)}
{f(x';\thetah)\exp\left\{-\frac{\norm{\nabla_\theta\Lcal(\thetah;x')}^2}{2\sigma^2/d}\right\}\det\left(\nabla^2_\theta \Lcal(\thetah;x')\right)} 
\cdot
 \frac{\one{x\in\Xcal_{\thetah}}}
{\one{x'\in\Xcal_{\thetah}}} \right\}.
\]
We consider the three fractions appearing in this expression. The first two are generally straightforward
to calculate, but the
 last ratio, with the indicator variables, requires more careful consideration.
In the denominator, we will have $\one{x'\in\Xcal_{\thetah}}=1$ always, since $x'$ denotes the current state which is therefore
a draw from the density~\eqref{eqn:density_hat} supported on $\Xcal_{\thetah}$. Turning to the numerator, however,
we see that we do need to verify that our proposed state $x$
also lies in $\Xcal_{\thetah}$. To do so, we observe that for any $\theta$,
\begin{align*} 
\one{x\in\Xcal_{\theta}}&= \One{\text{for some }w\in\R^d\text{, }\thetah(x,w)=\theta\text{ and }\theta\text{ is a SSOSP of }\Lcal(\theta;x,w)} \\
&=\One{ \thetah\left(x, -\frac{\nabla_\theta\Lcal(\theta;x)}{\sigma}\right) = \theta\text{, and }\theta\text{ is a SSOSP of }\Lcal\left(\theta;x,-\frac{\nabla_\theta\Lcal(\theta,x)}{\sigma}\right)}\\
&=\One{ \thetah\left(x, -\frac{\nabla_\theta\Lcal(\theta;x)}{\sigma}\right) = \theta\text{, and }\nabla^2_\theta\Lcal(\theta;x)\succ 0}.
\end{align*}
In other words, given the proposed state $x$, we need only verify (1) that $\nabla^2_\theta\Lcal(\thetah;x)\succ 0$, which is a
simple calculation, and (2) that the estimator $(x,w)\mapsto \thetah(x,w)$, when calculated with this proposed $x$
and with $w=  -\frac{\nabla_\theta\Lcal(\thetah;x)}{\sigma}$, indeed returns the observed value $\thetah$.
We note that, in the special case that $\Lcal$ is strictly convex, then
this verification is trivial---if we take the map $(x,w)\mapsto\thetah(x,w)$ to be the output of any solver guaranteed to return the unique FOSP (if it exists), then (2) is automatically verified since we know that $\thetah$ is a FOSP of $\Lcal(\theta;x,w)$ by definition of $w$, while (1) holds by strict convexity of $\Lcal$.

\subsubsection{Choosing the proposal distribution}
To choose the proposal distribution $q_{\thetah}(x\giv x')$, we will bear in mind the following considerations.
First, we need to be able to efficiently draw a sample from $q_{\thetah}(\cdot\giv x')$.
Second, we need to trade off between the following two goals:
given our current state $X_{\textnormal{curr}}$ and a proposed state $X_{\textnormal{prop}}\sim q_{\thetah}(\cdot\giv X_{\textnormal{curr}})$,
\begin{itemize}
\item The acceptance probability $A_{\thetah}(X_{\textnormal{prop}}\giv X_{\textnormal{curr}})$ should not be too close to zero.
\item There should not be too much similarity or dependence between $X_{\textnormal{curr}}$ and $X_{\textnormal{prop}}$.
\end{itemize}
To illustrate this tradeoff, if we define $q_{\thetah}(\cdot\giv X_{\textnormal{curr}})$ as the point mass at $X_{\textnormal{curr}}$ (i.e., we never move to a new state),
then the  acceptance probability $A_{\thetah}(X_{\textnormal{prop}}\giv X_{\textnormal{curr}})$ will be equal to 1 almost surely,
but the algorithm will return copies $\Xt^{(1)} = \dots = \Xt^{(M)} = X$, leading to a powerless procedure.
On the other hand, if we define $q_{\thetah}(\cdot\giv X_{\textnormal{curr}})$ to draw
$X_{\textnormal{prop}}$ to be independent or nearly independent of $X_{\textnormal{curr}}$ (for example, $X_{\textnormal{prop}}\sim P_{\thetah}$),
then it may be hard to ensure that $p_{\thetah}(X_{\textnormal{prop}}\giv\thetah)$ is sufficiently large
to bound $A_{\thetah}(X_{\textnormal{prop}}\giv X_{\textnormal{curr}})$ away from zero.

Given the well-known challenge of hyperparameter tuning in the field of MCMC \citep{roberts2009examples}, we can expect that this will be highly non-trivial and problem-dependent. But one appealing aspect of aCSS testing is that we can tune the MCMC hyperparameters \emph{after} looking at $\thetah$ without violating any of our theory. We demonstrate how we did so in our four examples below.

\paragraph{Examples~\ref{ex:canonical_glm},~\ref{ex:behrens_fisher}, and~\ref{ex:multivariate_t}}
First, we consider the three examples where our model $P_\theta$ for $X$ consists of $n$ independent draws---that is,
$P_\theta$ is a product distribution with density
\[f_\theta(x) = \prod_{i=1}^n f^{(i)}_\theta(x_i).\]
In this setting, we begin by fixing a parameter $s\in\{1,\dots,n\}$ (we will discuss the choice of $s$ shortly).
Then the proposal distribution $q_{\thetah}(x|x')$ is defined as follows:
\begin{itemize}
\item Draw a subset $\mathcal{S}\subseteq\{1,\dots,n\}$ of size $s$, uniformly at random.
\item For each $i=1,\dots,n$,
\begin{itemize}
\item If $i\in \mathcal{S}$, draw $x_i \sim f^{(i)}_{\thetah}(\cdot)$.
\item If $i\not\in \mathcal{S}$, set $x_i = x'_i$.
\end{itemize}
\end{itemize}
We can see that the parameter $s$ controls the tradeoff---a larger $s$ ensures
then the proposed state $x=X_{\textnormal{prop}}$ will not be too similar to the previous state $x'=X_{\textnormal{curr}}$,
but a smaller $s$ ensures that the acceptance ratio $A_{\thetah}(X_{\textnormal{prop}}\giv X_{\textnormal{curr}})$ will not be too low (since,
when most entries $i=1,\dots,n$ of $X_{\textnormal{prop}}$ coincide with those of $X_{\textnormal{curr}}$,
the ratio $ \frac{p_{\thetah}(X_{\textnormal{prop}}\giv\thetah)}{p_{\thetah}(X_{\textnormal{curr}}\giv\thetah)}$ 
should be close to 1).

Next, how can we choose $s$ to balance between these two considerations?
For these examples, we will choose $s$ from the data itself. First, we observe that allowing $s$ to depend on $\thetah$ does not violate
the validity of our procedure. This is because the mechanism $\Pt_M(\cdot\giv X,\thetah)$
for sampling the copies is only required to satisfy assumption~\eqref{eqn:exch_some_theta}; 
it is allowed to depend arbitrarily on $\thetah$, as long as exchangeability between $X$ and $\Xt^{(1)},\dots,\Xt^{(M)}$ is not violated. 
(In particular,
this means that we {\em cannot} use the data $X$ itself to choose $s$.)
We will choose $s$ by simulating the procedure with $\thetah$ in place of $\theta_0$:
\begin{itemize}
\item Let $\theta_0^{\textnormal{sim}} = \thetah$.
\item Draw $X^{\textnormal{sim}}\sim P_{\theta_0^{\textnormal{sim}}}$.
\item For each candidate choice of $s$, run Metropolis--Hastings initialized at $X^{\textnormal{sim}}$, 
and compute the average acceptance probability.
\item Repeat for many draws of $X^{\textnormal{sim}}$ to get an average acceptance probability $\bar{A}_s$ for each $s$, and among all values of $s$ such that $\bar{A}_s\ge 0.2$, choose the value of $s$
that maximizes $s\bar{A}_s$ (thus maximizing the expected number of elements that change at each MH step).
\end{itemize}
With this choice of $s$, 
we have completed our $\thetah$-dependent definition of the proposal distribution $q_{\thetah}(x\giv x')$ for this setting. Then we choose $L$ to be at least $\frac{n}{s\bar{A}_s}$ to ensure that (most) entries will be resampled within $L$ steps; in our simulations we chose $L$ to be $\min\{500,\frac{2n}{s\bar{A}_s}\}$ (rounded to an integer).

\paragraph{Example~\ref{ex:gaussian_spatial}}
Next we consider the Gaussian spatial process.
Here we will again define a parametrized proposal distribution, and will then choose the parameter
by simulation. For any $\rho\in(0,1)$, define the proposal distribution 
 $q_{\thetah}(x|x')$ as follows:
\begin{itemize}
\item Draw $x_{\textnormal{tmp}} \sim \normal(0,\Sigma_{\thetah})$.
\item Set $x =  \rho\cdot x' + \sqrt{1-\rho^2}\cdot x_{\textnormal{tmp}}$.
\end{itemize}
As for the examples above, the value of $\rho$ governs the tradeoff---in this case,
a smaller $\rho$ ensures
then the proposed state $x=X_{\textnormal{prop}}$ will not be too similar to the previous state $x'=X_{\textnormal{curr}}$,
but a larger $\rho$ ensures that the acceptance ratio $A_{\thetah}(X_{\textnormal{prop}}\giv X_{\textnormal{curr}})$ will not be too low.
In each trial, we will choose $\rho$ with a simulation, analogous to the choice of $s$ for the other examples:
\begin{itemize}
\item Let $\theta_0^{\textnormal{sim}} = \thetah$.
\item Draw $X^{\textnormal{sim}}\sim P_{\theta_0^{\textnormal{sim}}}$, $W\sim\normal(0,\tfrac{1}{d}\ident_d)$, 
and calculate $\thetah=\thetah(X^{\textnormal{sim}},W)$.
\item For each candidate choice of $\rho$, run one step of Metropolis--Hastings initialized at $X^{\textnormal{sim}}$,
to generate $X^{\textnormal{new}}$.
\item Repeat for 500 draws of $X^{\textnormal{sim}}$ (discarding any draws for which $\thetah(X^{\textnormal{sim}},W)$ is not a SSOSP). Among all values of $\rho$
that achieve average acceptance probability $\geq 0.05$, find the value of $\rho$ that minimizes the average
correlation between $X^{\textnormal{sim}}$ and $X^{\textnormal{new}}$.
\end{itemize}
With this choice of $\rho$, writing $\hat\rho$ to denote  the average
correlation between $X^{\textnormal{sim}}$ and $X^{\textnormal{new}}$, we then set $L=\min\{500,\frac{20}{1-(\hat\rho)_+}\}$ (rounded to an integer).

\section{Details for Figure~\ref{fig:parboot}}\label{app:details_parboot}

In this section we give details for the simulation that generated Figure~\ref{fig:parboot}, comparing the parametric bootstrap versus co-sufficient sampling for a Gaussian linear model setting as described in Section~\ref{sec:intro}. Recall that the null hypothesis for this example is the model
\[X = \theta\cdot Z + \normal(0,\ident_n)\]
for some $\theta\in\R$, where $Z\in\R^n$ is a fixed covariate vector. We are interested in testing the alternative hypothesis that $X$ is in fact more strongly associated with some other covariate $Y\in\R^n$, and so our test statistic is given by 
\[T(X) = \frac{(X^\top Y)^2}{(X^\top Z)^2}.\]

To generate the data, we choose sample size $n=100$, and then independently for each $i=1,\dots,n$, we generate the triple $(X_i,Y_i,Z_i)$ by taking
\[(Y_i,Z_i) \sim \normal\left(0, \left(\begin{array}{cc} 1 &\ \rho \\ \rho &\ 1 \end{array} \right) \right),\]
with correlation parameter $\rho =0.97$, and define
\[X_i = \theta_0 \cdot Z_i + \normal(0,1),\]
where the true parameter is chosen as $\theta_0 =0$.

Next we run parametric bootstrap and CSS to generate copies $\Xt^{(1)},\dots,\Xt^{(M)}$ of the data $X\in\R^n$, for $M=500$. For both methods, the MLE is given by $\thetah = (Z^\top Z)^{-1}Z^\top X$. To run the parametric bootstrap, we generate the copies from the distribution with parameter $\theta=\thetah$, that is, we define the copies as
\[\Xt^{(m)}_{\textnormal{boot}} = \thetah \cdot Z + V_m,\]
where $V_m \iidsim \normal(0,\mathbf{I}_n)$.
To run CSS, we instead condition on the MLE $\thetah$, and the copies can therefore be generated as
\[\Xt^{(m)}_{\textnormal{CSS}} = \thetah \cdot Z + \textnormal{Proj}_Z^{\perp} \cdot V_m,\]
where again $V_m \iidsim \normal(0,\mathbf{I}_n)$.

Finally, we repeat the simulation for 10,000 independent trials to generate the histograms of p-values for each method, as shown in Figure~\ref{fig:parboot}.

\section*{Acknowledgements}
The authors would like to thank Michael Bian for help with some of the computation.
The first author was supported by the National Science Foundation via grant DMS--1654076, and by the Office of Naval Research via grant N00014-20-1-2337.

\bibliographystyle{plainnat}
\bibliography{bibliography}

\end{document}